\documentclass{ws-rv9x6}
\usepackage{ws-rv-thm}
\usepackage[square]{ws-rv-van}
\usepackage{simplewick}
\usepackage{relsize}
\usepackage{alphalph}

\newif\ifarxivsubmission

\arxivsubmissiontrue

\ifarxivsubmission
  \usepackage{color}
  \definecolor{darkblue}{rgb}{0.1,0.1,.7}
  \usepackage[colorlinks, linkcolor=darkblue, citecolor=darkblue, urlcolor=darkblue, linktocpage]{hyperref}
  \crop[off]{}
  \makeatletter
  \def\@makechapterhead#1{%
    {\vbox to 110pt{
    \def\thefootnote{\@fnsymbol\c@footnote}%
    \vspace*{37pt}      
        \parindent\z@\raggedright\reset@font
        {\centering{
         \vskip 0.25in
    \vbox{
    \CTfont #1\par
    }\par}\par}\nobreak\vfill}}}
  \makeatother
\else
  \usepackage[]{hyperref}
\fi

\makeindex

\newcommand\be{\begin{eqnarray}}
\newcommand\ee{\end{eqnarray}}
\newcommand\f\phi
\newcommand\cO{\mathcal{O}}
\newcommand\p[1]{\left(#1\right)}
\newcommand\ptl\partial
\newcommand\e\epsilon
\newcommand\<\langle
\renewcommand\>\rangle
\newcommand\Z{\mathbb{Z}}
\newcommand\de\delta
\newcommand\R{\mathbb{R}}
\newcommand\bx{\mathbf{x}}
\newcommand\nn{\nonumber}
\renewcommand\.{\cdot}
\newcommand\x\times
\newcommand\pdr[2]{\frac{\partial #1}{\partial #2}}
\newcommand\s\sigma
\newcommand\SO{\mathrm{SO}}
\newcommand\De{\Delta}
\newcommand\cS{\mathcal{S}}
\newcommand\oo\infty
\newcommand\SU{\mathrm{SU}}
\newcommand\cH{\mathcal{H}}
\newcommand\bn{\mathbf{n}}
\newcommand\bP{\mathbf{P}}
\renewcommand\b\beta
\renewcommand\a\alpha
\newcommand\Tr{\mathrm{Tr}}
\renewcommand\l\lambda
\newcommand\cL{\mathcal{L}}
\newcommand\cD{\mathcal{D}}
\renewcommand\th{\theta}
\newcommand\tl[1]{\widetilde{#1}}

\newtheorem{exercise}{Exercise}[section]

\begin{document}

\ifarxivsubmission
\chapter{TASI Lectures on the Conformal Bootstrap}\label{ch:bootstrap}
\else
\chapter{The Conformal Bootstrap}
\fi

\author[David Simmons-Duffin]{David Simmons-Duffin}

\address{School of Natural Sciences, Institute for Advanced Study\\
Princeton, NJ 08540, \\
dsd@ias.edu}

\begin{abstract}
These notes are from courses given at TASI and the Advanced Strings School in summer 2015.  Starting from principles of quantum field theory and the assumption of a traceless stress tensor, we develop the basics of conformal field theory, including conformal Ward identities, radial quantization, reflection positivity, the operator product expansion, and conformal blocks. We end with an introduction to numerical bootstrap methods, focusing on the 2d and 3d Ising models.
\end{abstract}
\body

\makeatletter
\newalphalph{\alphmult}[mult]{\@alph}{26}
\makeatother
\renewcommand{\thefootnote}{\alphmult{\value{footnote}}}

\ifarxivsubmission
  \thispagestyle{empty}
  \makeatletter
  \renewcommand*\l@chapter[2]{}
  \renewcommand*\l@schapter[2]{}
  \renewcommand*\l@author[2]{}
  \makeatother
  \newpage
  \setcounter{page}{1}
  \setcounter{tocdepth}{2}
  \smalltoc
  \tableofcontents
\fi

\section*{Other Resources}

This course is heavily inspired by Slava Rychkov's {\it EPFL Lectures on Conformal Field Theory in $d\geq 3$ Dimensions} \cite{Rychkov:2016iqz}. His notes cover similar topics, plus additional material that we won't have time for here, including conformal invariance in perturbation theory, the embedding formalism, and some analytical bootstrap bounds.  By contrast, these lectures spend more time on QFT basics and numerical bootstrap methods. See also lectures by Sheer El-Showk \cite{Busan} and Joshua Qualls \cite{Qualls:2015qjb}.

Our discussion of symmetries and quantization is based on Polchinski's {\it String Theory, Vol.\ 1} \cite{Polchinski:1998rq}: mostly Chapter 2 on 2d CFTs and Appendix A on path integrals.  Appendix A is required reading for any high energy theory student.

The book {\it Conformal Field Theory\/} by Di Francesco, Mathieu, and Senechal \cite{YellowPages} is also a useful reference. It starts with a discussion of CFTs in general spacetime dimensions, but includes much more about 2d CFTs, a topic we unfortunately neglect in this course.

\ifarxivsubmission
  \newpage
\fi

\section{Introduction}

\subsection{Landmarks in the Space of QFTs}

Quantum field theories generically become scale-invariant at long distances. Often, invariance under rescaling actually implies invariance under the larger conformal group, which consists of transformations that locally look like a rescaling and a rotation.\footnote{The question of when scale invariance implies conformal invariance is an important foundational problem in quantum field theory that is still under active study.  In 2d and 4d, it has been proven that Lorentz-invariance and unitarity are sufficient conditions \cite{Polchinski:1987dy,Dymarsky:2013pqa}.  In 3d or $d\geq 5$, the appropriate conditions are not known, but conformal invariance appears in many examples.}  These extra symmetries are powerful tools for organizing a theory. Because their emergence requires no special structure beyond the long distance limit, they are present in a huge variety of physical systems.

We can think of a UV-complete QFT as a renormalization group (RG) flow between conformal field theories (CFTs),\footnote{Having a CFT in the IR is generic.  We do not necessarily have a CFT in the UV, but assuming one is sometimes a useful framework.}
\be
\left.
\begin{array}{c}
\mathrm{CFT}_{UV}\\
\downarrow\\
\mathrm{CFT}_{IR}
\end{array}\right\} \mathrm{QFT}.
\ee
Studying CFTs lets us map out the possible endpoints of RG flows, and thus understand the space of QFTs.

Many of the most interesting RG flows are nonperturbative. A simple example is $\f^4$ theory in 3 dimensions, which has the Euclidean action
\be
S&=&\int d^3 x\p{\frac 1 2(\ptl\phi)^2 + \frac 1 2 m^2 \phi^2 + \frac 1 {4!}g \phi^4}.
\ee
This theory is free in the UV, since $m$ and $g$ have mass dimension 1 and can be ignored at high energies.  The behavior of the theory in the IR depends on the ratio $g^2/m^2$.  If $m^2$ is large and positive, the IR theory is massive and preserves the $\Z_2$ symmetry $\phi\to -\phi$.  If $m^2$ is large and negative, the IR theory is again massive but spontaneously breaks $\Z_2$.  For a special value of $g^2/m^2$, in between these two regimes, the IR theory becomes gapless and is described by a nontrivial interacting CFT.\footnote{The precise value of $g^2/m^2$ that gives a CFT is scheme-dependent: it depends on how one regulates UV divergences.}

It is hard to study this CFT with Feynman diagrams.  By dimensional analysis, naive perturbation theory leads to an expansion in $gx$, where $x$ is a distance scale characterizing the observable we're computing.  At distances $x\gg 1/g$, the expansion breaks down.  Instead, the best perturbative tool we have is the $\e$-expansion, where we compute Feynman diagrams in $4-\e$ dimensions and afterwards continue $\e\to 1$.  This works surprisingly well, but is conceptually a little shaky.  

\subsection{Critical Universality}

In the example above, the UV theory was a continuum QFT: the free boson.  However, IR CFTs can also arise from very different microscopic systems \cite{Polyakov:1970xd}.  An example is the 3d Ising model, which is a lattice of classical spins $s_i \in \{\pm 1\}$ with nearest-neighbor interactions.  The partition function is
\be
Z_\mathrm{Ising} &=& \sum_{\{s_i\}} \exp\p{-J \sum_{\<ij\>} s_i s_j},
\ee
where $i,j$ label lattice points and $\<ij\>$ indicates that $i$ and $j$ are nearest neighbors.
We can think of this sum as a discrete path integral, where the integration variable is the space of functions
\be
s: \mathrm{Lattice} \to \{\pm 1\}.
\ee
For a special value of $J$, this theory also becomes a nontrivial CFT at long distances. Actually {\it it is the same CFT as the one appearing in $\phi^4$ theory!\/}

The Ising CFT also arises in water (and other liquids) at the critical point on its phase diagram, and in uni-axial magnets at their critical temperatures \cite{Cardy:1996xt}. We say that $\phi^4$ theory, the Ising model, water, and uni-axial magnets are {\it IR equivalent\/} at their critical points (figure~\ref{fig:rgflows}), and that they are in the same universality class.  IR equivalences show up all over high-energy and condensed-matter physics, where they are sometimes called ``dualities."  The ubiquity of IR equivalences is the phenomenon of {\it critical universality.} 

\begin{figure}
\begin{center}
\includegraphics[width=0.9\textwidth]{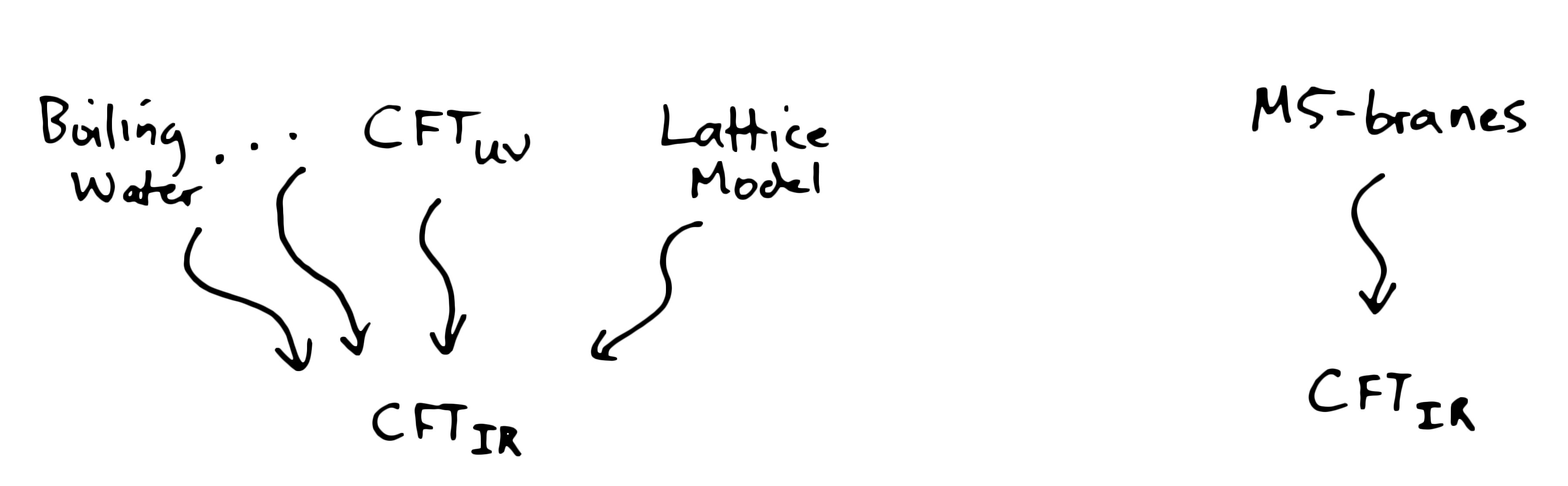}
\end{center}
\caption{Many  microscopic theories can flow to the same IR CFT\@. We say that the theories are IR equivalent, or IR dual. The UV can even be something exotic like a stack of M5-branes in M-theory. \label{fig:rgflows}}
\end{figure}

The above examples are Euclidean field theories. But Lorentzian CFTs also appear in nature, describing quantum critical points. For example, the Lorentzian $O(2)$ model describes thin-film superconductors~\cite{PhysRevB.44.6883,PhysRevLett.95.180603}, while its Wick-rotation, the Euclidean $O(2)$ model, describes the superfluid transition in ${}^4$He~\cite{Lipa:2003zz}. Amazingly, the critical exponents of these theories agree, allowing us to see Wick rotation in nature!

\subsection{The Bootstrap Philosophy}

Critical universality means we can study the 3d Ising CFT by doing computations in any of its microscopic realizations. This is a powerful tool.  For example, we can model critical water by simulating classical spins on a computer, without ever worrying about $10^{23}$ bouncing water molecules!  For analytical results, we can use the $\e$-expansion.  But all of these approaches fail to exploit the emergent symmetries of the IR theory.

The conformal bootstrap philosophy is to:
\begin{enumerate}
\item[0.] focus on the CFT itself and not a specific microscopic realization,
\item[1.] determine the full consequences of symmetries,
\item[2.] impose consistency conditions,
\item[3.] combine (1) and (2) to constrain or even solve the theory.
\end{enumerate}

This strategy was first articulated by Ferrara, Gatto, and Grillo~\cite{Ferrara:1973yt} and Polyakov~\cite{Polyakov:1974gs} in the 70's.  Importantly, it only uses nonperturbative structures, and thus has a hope of working for strongly-coupled theories.  Its effectiveness for studying the 3d Ising model will become clear during this course.  In addition, sometimes bootstrapping is the {\it only\/} known strategy for understanding the full dynamics of a theory. An example is the 6d $\mathcal{N}=(2,0)$ supersymmetric CFT describing the IR limit of a stack of M5 branes in M-theory.  This theory has no known Lagrangian description, but is amenable to bootstrap analysis \cite{Beem:2015aoa}.\footnote{At large central charge, this theory is solved by the AdS/CFT correspondence \cite{Maldacena:1997re}. Supersymmetry also lets one compute a variety of protected quantities (at any central charge).  However, the bootstrap is currently the only known tool for studying non-protected quantities at small central charge.}

A beautiful and ambitious goal of the bootstrap program is to eventually provide a fully nonperturbative  formulation of quantum field theory, removing the need for a Lagrangian.  We are not there yet, but you can help!

\section{QFT Basics}

The first step of the conformal bootstrap is to determine the full consequences of symmetries.  In this section, we quickly review symmetries in quantum field theory, phrasing the discussion in language that will be useful later.  We work in Euclidean signature throughout.

\subsection{The Stress Tensor}

A local quantum field theory has a conserved stress tensor,
\be
\label{eq:conservationofT}
\ptl_\mu T^{\mu\nu}(x) &=& 0 \qquad \textrm{(operator equation)}.
\ee
This holds as an ``operator equation," meaning it is true away from other operator insertions.  In the presence of other operators, (\ref{eq:conservationofT}) gets modified to include contact terms on the right-hand side,
\be
\label{eq:wardidentity}
\partial_\mu \< T^{\mu\nu}(x) \cO_1(x_1)\dots \cO_n(x_n)\> &=& -\sum_i \de(x-x_i)\ptl_i^\nu\<\cO_1(x_1)\dots \cO_n(x_n)\>.\nn\\
\ee

\begin{exercise}
Consider a QFT coupled to a background metric $g$. For concreteness, suppose correlators are given by the path integral
\be
\<\cO_1(x_1)\dots\cO_n(x_n)\>_g &=& \int D\phi\,\cO_1(x_1)\dots\cO_n(x_n)\, e^{-S[g,\f]}.
\ee
A stress tensor insertion is the response to a small metric perturbation,\footnote{This definition  of the stress tensor works in a continuum field theory. If the UV is a lattice model, we must assume (or prove) that a stress tensor emerges in the IR.}
\be
\label{eq:definitionofstresstensor}
\<T^{\mu\nu}(x)\cO_1(x_1)\dots\cO_n(x_n)\>_g &=& \frac{2}{\sqrt g}\frac{\de}{\de g_{\mu\nu}(x)}\<\cO_1(x_1)\dots\cO_n(x_n)\>_g.
\ee
Derive (\ref{eq:wardidentity}) by demanding that $S[g,\phi]$ be diffeomorphism invariant near flat space.  Find how to modify (\ref{eq:wardidentity}) when the $\cO_i$ have spin.
\end{exercise}

Consider the integral of $T^{\mu\nu}$ over a closed surface $\Sigma$,\footnote{The word ``surface" usually refers to a 2-manifold, but we will abuse terminology and use it to refer to a codimension-1 manifold.}\footnote{Our definition of $P^\nu$ differs from the usual one by a factor of $i$.  This convention is much nicer for Euclidean field theories, but it has the effect of modifying some familiar formulae, and also changing the properties of symmetry generators under Hermitian conjugation. More on this in section~\ref{sec:reflectionpositivity}.}
\be
P^\nu(\Sigma) &\equiv& -\int_\Sigma dS_\mu T^{\mu\nu}(x).
\ee
The Ward identity (\ref{eq:wardidentity}) implies that a correlator of $P^\nu(\Sigma)$ with other operators is unchanged as we move $\Sigma$, as long as $\Sigma$ doesn't cross any operator insertions (figure~\ref{fig:topologicalsurfaces}).
We say that $P^\nu(\Sigma)$ is a ``topological surface operator."

\begin{figure}
\begin{center}
\includegraphics[width=0.45\textwidth]{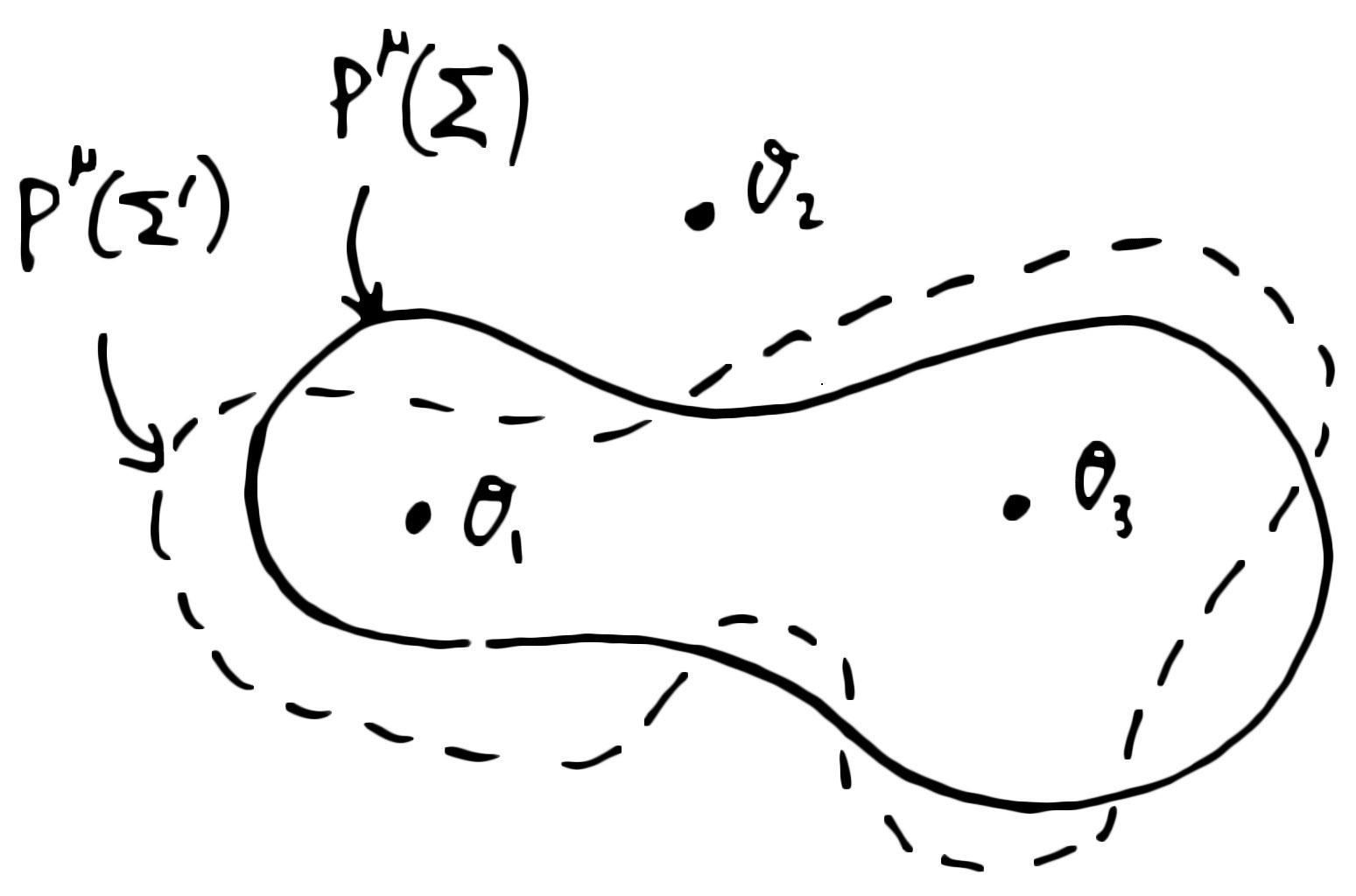}
\end{center}
\caption{A surface $\Sigma$ supporting the operator $P^\mu(\Sigma)$ can be freely deformed $\Sigma\to\Sigma'$ without changing the correlation function, as long as it doesn't cross any operator insertions.
\label{fig:topologicalsurfaces} }
\end{figure}

Let $\Sigma=\ptl B$ be the boundary of a ball $B$ containing $x$ and no other insertions. Integrating (\ref{eq:wardidentity}) over $B$ gives
\be
\label{eq:integratedwardidentity}
\<P^\mu(\Sigma)\cO(x)\dots\> &=& \ptl^\mu\<\cO(x)\dots\>.
\ee
In other words, surrounding $\cO(x)$ with the topological surface operator $P^\mu$ is equivalent to taking a derivative (figure~\ref{fig:surroundoperator}).

\begin{figure}
\begin{center}
\includegraphics[width=0.5\textwidth]{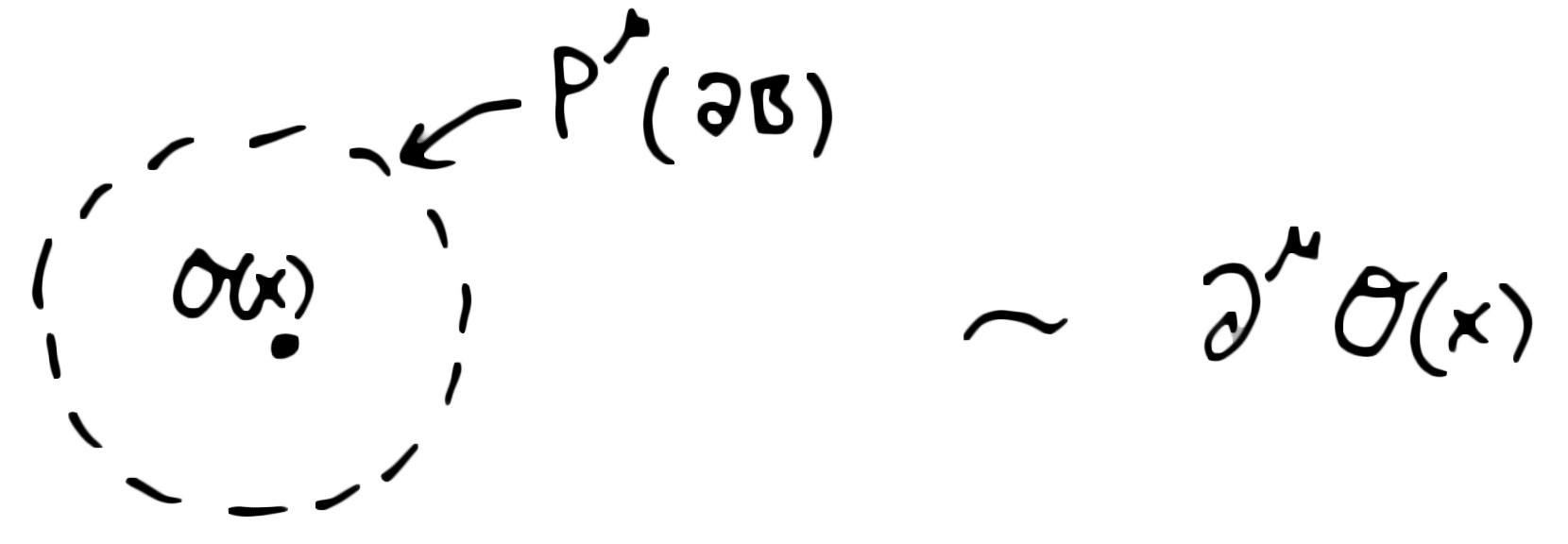}
\end{center}
\caption{\label{fig:surroundoperator} Surrounding $\cO(x)$ with $P^\mu$ gives a derivative.}
\end{figure}

In quantum field theory, having a topological codimension-1 operator is the same as having a symmetry.\footnote{Topological operators with support on other types of manifolds give ``generalized symmetries" \cite{Gaiotto:2014kfa}.}  This may be unfamiliar language, so to connect to something more familiar, let us revisit the relation between the path integral and Hamiltonian formalisms.

\subsection{Quantization}
\label{sec:quantization}

A single path integral can be interpreted in terms of different time evolutions in different Hilbert spaces.  For example, in a rotationally-invariant Euclidean theory on $\R^d$, we can choose any direction as ``time" and think of states living on slices orthogonal to the time direction (figure~\ref{fig:differentquantizations}).  We call each interpretation a ``quantization" of the theory.

\begin{figure}
\begin{center}
\includegraphics[width=0.65\textwidth]{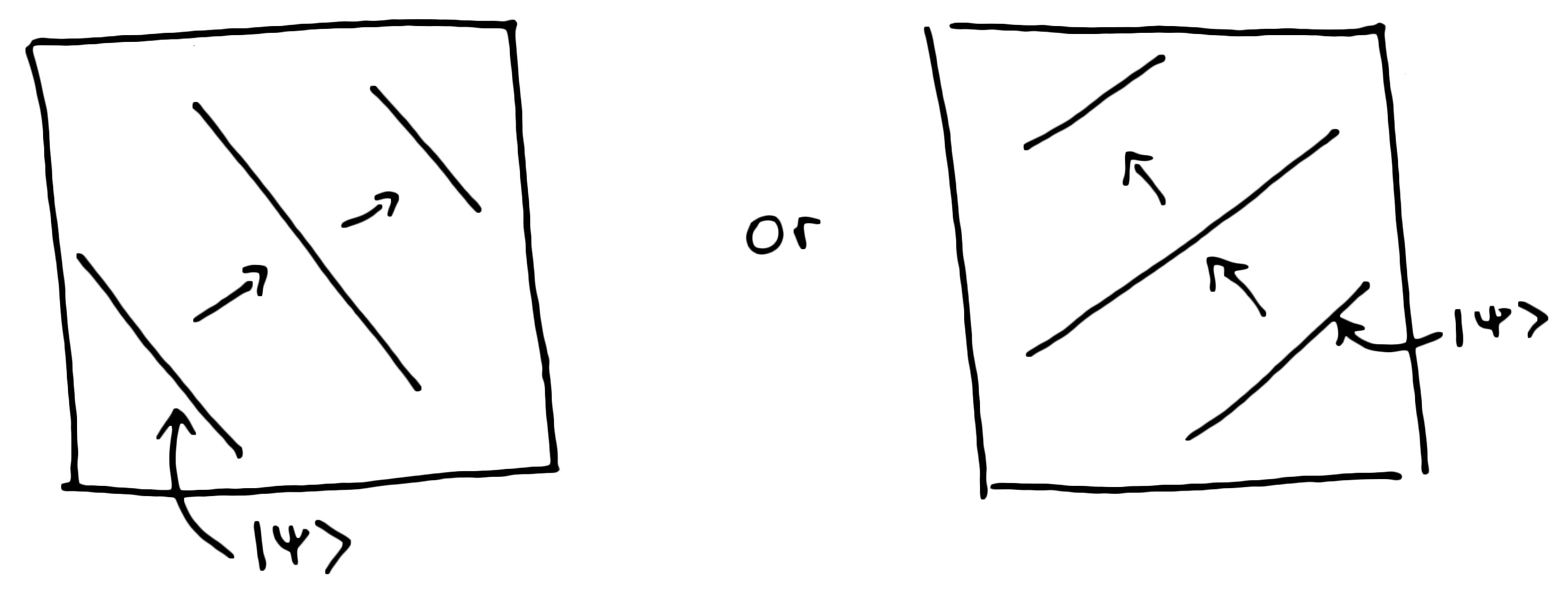}
\end{center}
\caption{\label{fig:differentquantizations} In a rotationally invariant Euclidean theory, we can choose any direction as time.  States live on slices orthogonal to the time direction.}
\end{figure}

To specify a quantization, we foliate spacetime by slices related by an isometry $\ptl_t$. A slice has an associated Hilbert space of states.  A correlation function $\<\cO_1(x_1)\cdots\cO_n(x_2)\>$ gets interpreted as a time-ordered expectation value
\be
\label{eq:timeordered}
\<\cO_1(x_1)\cdots\cO_n(x_n)\> &=& \<0|T\{\widehat \cO_1(t_1,\bx_1)\cdots \widehat \cO_n(t_n,\bx_n)\}|0\>.
\ee
Here, the time ordering $T\{\dots\}$ is with respect to our foliation, $|0\>$ is the vacuum in the Hilbert space $\cH$ living on a spatial slice,\footnote{Other choices of initial and final state correspond to different boundary conditions for the path integral.} and $\widehat\cO_i(x):\cH\to \cH$ are quantum operators corresponding to the path integral insertions $\cO_i(x)$.
  
A different quantization of the theory would give a completely different Hilbert space $\cH'$, a completely different time-ordering, and completely different quantum operators $\widehat \cO_i'$.  However, some equations satisfied by these new operators on this new Hilbert space would be unchanged.  For example, if we arrange the operators as shown on the right-hand side of (\ref{eq:timeordered}), we always get the correlator on the left-hand side.

We demonstrate these ideas explicitly in appendix~\ref{app:latticequantization}, where we show how to (discretely) quantize the lattice Ising model in different ways.

\subsection{Topological Operators and Symmetries}

\begin{figure}
\begin{center}
\includegraphics[width=0.5\textwidth]{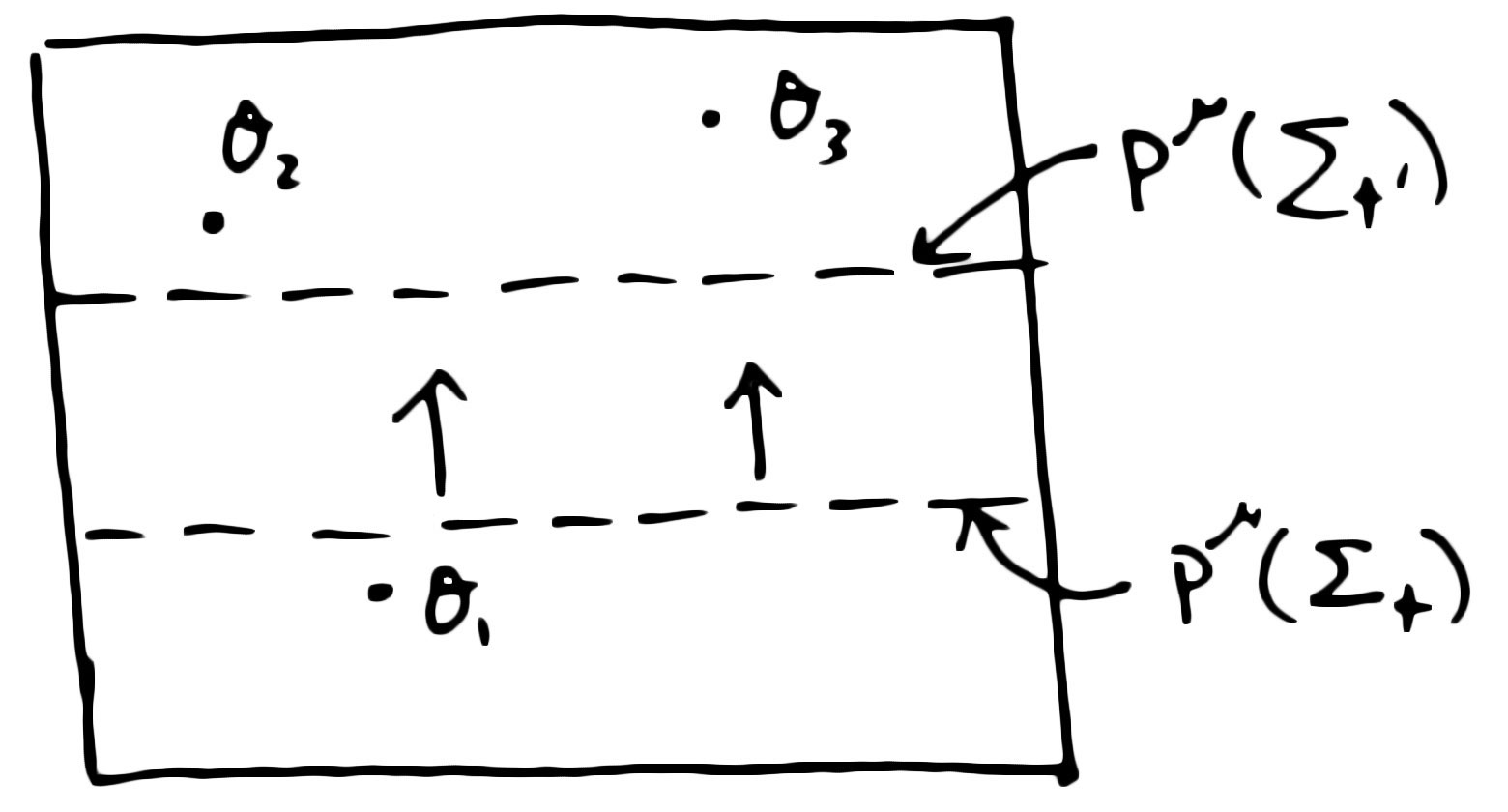}
\end{center}
\caption{\label{fig:slidingcharges} The charge $P^\mu(\Sigma_t)$ can be moved from one time to another $t\to t'$ without changing the correlation function.}
\end{figure}

Let $\Sigma_t$ be a spatial slice at time $t$ and consider the operator $P^\mu(\Sigma_t)$.  Because $P^\mu(\Sigma)$ is topological, we are free to move it forward or backward in time from one spatial slice to another as long as it doesn't cross any operator insertions (figure~\ref{fig:slidingcharges}). In fact, we  often neglect to specify the surface $\Sigma_t$ and just write $P^\mu$ (though we should keep in mind where the surface lives with respect to other operators). We call $P^\mu$ ``momentum," and the fact that it's topological is the path integral version of the statement that momentum is conserved.

Let us understand what happens when we move $P^\mu$ past an operator insertion. Consider a local operator $\cO(x)$ at time $t$. If $\Sigma_1$, $\Sigma_2$ are spatial surfaces at times $t_1<t<t_2$, then when we quantize our theory, the difference $P^\mu(\Sigma_{2})-P^\mu(\Sigma_{1})$ becomes a commutator because of time ordering,
\be
\<(P^\mu(\Sigma_{2})-P^\mu(\Sigma_{1}))\cO(x)\dots\>=\<0|T\{[\widehat P^\mu,\widehat \cO(x)]\dots\}|0\>.
\ee
(We assume that the other insertions ``$\dots$" are not between times $t_1$ and $t_2$.)
Because $P^\mu$ is topological, we can deform $\Sigma_2-\Sigma_1$ to a sphere $S$ surrounding $\cO(x)$, as in figure~\ref{fig:deformingcharges}.  Then using the Ward identity (\ref{eq:integratedwardidentity}), we find
\be
\<0|T\{[\widehat P^\mu, \widehat \cO(x)]\dots\}|0\>&=&\<(P^\mu(\Sigma_2)-P^\mu(\Sigma_1))\cO(x)\dots\>\nn\\
&=&\<P^\mu(S)\cO(x)\dots\>\nn\\
&=& \ptl^\mu\<\cO(x)\dots\>\nn\\
&=& \ptl^\mu\<0|T\{\widehat \cO(x)\dots\}|0\>,
\ee
in other words,
\be
[\widehat P^\mu, \widehat \cO(x)] &=& \ptl^\mu\widehat \cO(x).
\ee

\begin{figure}
\begin{center}
\includegraphics[width=0.8\textwidth]{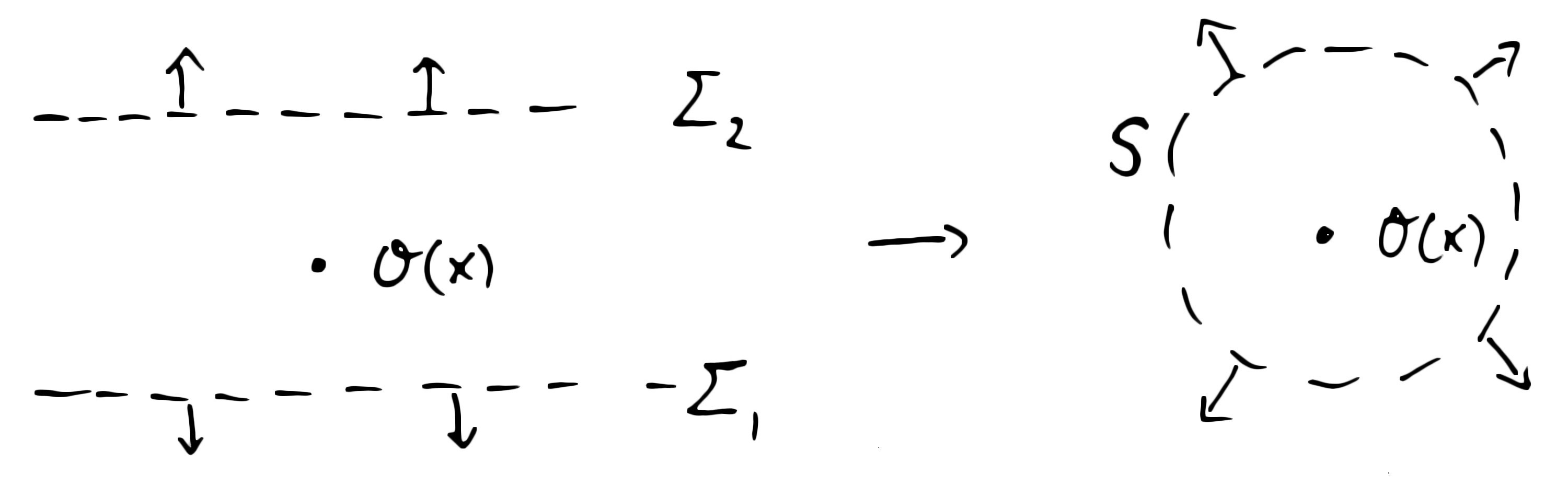}
\end{center}
\caption{\label{fig:deformingcharges} For any charge $Q(\Sigma)$, we can deform $Q(\Sigma_2)-Q(\Sigma_1)=Q(\Sigma_2-\Sigma_1)$ to an insertion of $Q(S)$. Here, arrows indicate the orientation of the surface.}
\end{figure}

Figure~\ref{fig:deformingcharges} makes it clear that this result is independent of how we quantize our theory, since we always obtain a sphere surrounding $\cO(x)$ no matter which direction we choose as ``time."  Thus, we often write
\be
\label{commutator}
[P^\mu,\cO(x)] &=& \ptl^\mu\cO(x),
\ee
without specifying a quantization.  In fact, from now on, we will no longer distinguish between path integral insertions $\cO(x)$ and quantum operators $\widehat \cO(x)$.
The expression $[Q,\cO(x)]$ can be interpreted as either an actual commutator $[\widehat Q,\widehat \cO(x)]$ in any quantization of the theory, or in path-integral language as surrounding $\cO(x)$ with a topological surface operator $Q(S)$.

Figure~\ref{fig:deformingcharges} also explains why the commutator $[Q,\cO(x)]$ of a charge $Q$ with a local operator $\cO(x)$ is local, even though $Q$ is nonlocal (it is the integral of a current). The reason is that the support of $Q$ can be deformed to an arbitrarily small sphere $S$ around $x$, so that the insertion $Q(S)\cO(x)$ only affects the path integral in an infinitesimal neighborhood of $x$.  In general, the way local operators transform under symmetry is always insensitive to IR details like spontaneous symmetry breaking or finite temperature.  This is because commutators with conserved charges can be computed at short distances.

Equation (\ref{commutator}) integrates to
\be
\label{eq:integratedtranslations}
\cO(x) &=& e^{x\.P}\cO(0)e^{-x\.P}.
\ee
This statement is also true in any quantization of the theory.  In path integral language, $e^{x\.P}(\Sigma)$ is another type of topological surface operator.  When we surround $\cO(0)$ with $e^{x\.P}(\Sigma)$, it becomes conjugation $e^{x\.P}(\Sigma)\cO(0)\to e^{\widehat P\. x}\widehat\cO(0)e^{- \widehat P\. x}$ in any quantization.

Consider the time-ordered correlator (\ref{eq:timeordered}) with $t_n>\cdots>t_1$.  Using (\ref{eq:integratedtranslations}), it becomes
\begin{align}
&\<\cO_1(x_1)\cdots\cO_n(x_n)\>\nonumber\\
&= \<0|e^{t_n P^0}\cO_n(0,\bx_n)e^{-t_n P^0}\cdots e^{t_1 P^0}\cO_1(0,\bx_1)e^{-t_1P^0}|0\>\nonumber\\
&=\<0|\cO_n(0,\bx_n)e^{-(t_n-t_{n-1})P^0}\cdots e^{-(t_2-t_1)P^0}\cO_1(0,\bx_1)|0\>.
\end{align}
In other words, the path integral between spatial slices separated by time $t$ computes the time evolution operator $U(t)=e^{-tP^0}$.  In unitary theories (defined in more detail in section~\ref{sec:reflectionpositivity}), $P^0$ has a positive real spectrum, so $U(t)$ causes damping at large time separations.

\subsection{More Symmetries}

Given a conserved current $\ptl_\mu J^\mu=0$ (operator equation), we can always define a topological surface operator by integration.\footnote{It is an interesting question whether the converse is true. When a theory has a Lagrangian description, the Noether procedure gives a conserved current for any continuous symmetry (that is manifest in the Lagrangian).  Proving Noether's theorem without a Lagrangian is an open problem.} For $P^\nu$, the corresponding currents are $T^{\mu\nu}(x)$.  More generally, given a vector field $\e=\e^\mu(x)\ptl_\mu$, the charge
\be
Q_\e(\Sigma) &=& -\int_\Sigma dS_\mu \e_\nu(x) T^{\mu\nu}(x)
\ee
will be conserved whenever
\be
0&=&\ptl_\mu(\e_\nu T^{\mu\nu}) \nn\\
&=&
 \ptl_\mu \e_\nu T^{\mu\nu}+\e_\nu \ptl_\mu T^{\mu\nu}\nn\\
&=& \frac 1 2(\ptl_\mu \e_\nu+\ptl_\nu \e_\mu) T^{\mu\nu},
\ee
or
\be
\label{eq:killingvector}
\ptl_\mu\e_\nu+\ptl_\nu\e_\mu &=& 0.
\ee
This is the Killing equation. In flat space, it has solutions
\begin{align}
\label{eq:poincaregenerators}
p_\mu &= \ptl_\mu &&\textrm{(translations)},\nn\\
m_{\mu\nu} &= x_\nu\ptl_\mu - x_\mu\ptl_\nu && \textrm{(rotations)}.
\end{align}
The corresponding charges
are momentum $P_\mu=Q_{p_\mu}$ and angular momentum $M_{\mu\nu}=Q_{m_{\mu\nu}}$.

\section{Conformal Symmetry}

In a conformal theory, the stress tensor satisfies an additional condition: it is traceless,
\be
T_\mu^\mu(x) &=& 0 \qquad\textrm{(operator equation)}.
\ee
This is equivalent to the statement that the theory is insensitive to position-dependent rescalings of the metric $\de g_{\mu\nu}=\omega(x) g_{\mu\nu}$ near flat space.\footnote{In curved space, there can by Weyl anomalies.} When the stress tensor is traceless, we can relax the requirement (\ref{eq:killingvector}) further to
\be
\label{eq:conformalkilling}
\ptl_\mu\e_\nu + \ptl_\nu \e_\mu = c(x)\de_{\mu\nu},
\ee
where $c(x)$ is a scalar function.  Contracting both sides with $\de^{\mu\nu}$ gives $c(x)=\frac{2}{d}\ptl\.\e(x)$.  Equation (\ref{eq:conformalkilling}) is the {\it conformal\/} Killing equation.  It has two additional types of solutions in $\R^d$,
\begin{align}
\label{eq:extraconformalgenerators}
d &= x^\mu \ptl_\mu &&\textrm{(dilatations)},\nn\\
k_\mu &= 2x_\mu (x\.\ptl)-x^2\ptl_\mu &&\textrm{(special conformal transformations)}.
\end{align}
The corresponding symmetry charges are $D=Q_d$ and $K_\mu=Q_{k_\mu}$.\footnote{The above solutions are present in any spacetime dimension.  In two dimensions, there exists an infinite set of additional solutions to the conformal Killing equation, leading to an infinite set of additional conserved quantities \cite{Belavin:1984vu}.  This is an extremely interesting subject that we unfortunately won't have time for in this course.}

\subsection{Finite Conformal Transformations}

Before discussing the charges $P_\mu,M_{\mu\nu},D,K$, let us take a moment to understand the geometrical meaning of the conformal Killing vectors (\ref{eq:poincaregenerators}) and (\ref{eq:extraconformalgenerators}).  Consider an infinitesimal transformation $x^\mu\to x'^\mu=x^\mu+\e^\mu(x)$.  If $\e^\mu$ satisfies the conformal Killing equation, then
\be
\label{eq:conformalinfinitesimal}
\pdr{x'^\mu}{x^\nu} &=& \de^{\mu}_\nu+\ptl_\nu\e^\mu
\ \ =\ \ \p{1+\frac 1 d (\ptl\.\e)}\p{\de^\mu_\nu + \frac 1 2\p{\ptl_\nu \e^\mu - \ptl^\mu \e_\nu}}.
\ee
The right-hand side is an infinitesimal rescaling times an infinitesimal rotation.  Exponentiating gives a coordinate transformation $x\to x'$ such that
\be
\label{eq:conformalfinite}
\pdr{x'^\mu}{x^\nu} &=& \Omega(x)R^\mu{}_\nu{}(x),\qquad R^TR=I_{d\x d},
\ee
where $\Omega(x)$ and $R^\mu{}_\nu{}(x)$ are finite position-dependent rescalings and rotations.
Equivalently, the transformation $x\to x'$ rescales the metric by a scalar factor,
\be
\delta_{\mu\nu}\pdr{x'^\mu}{x^\alpha}\pdr{x'^\nu}{x^\beta} &=& \Omega(x)^2\de_{\alpha\beta}.
\ee
Such transformations are called {\it conformal}. They comprise the conformal group, a finite-dimensional subgroup of the diffeomorphism group of $\R^d$.

The exponentials of $p_\mu$ and $m_{\mu\nu}$ are translations and rotations. Exponentiating $d$ gives a scale transformation $x\to \lambda x$.  
We can understand the exponential of $k_\mu$ by first considering an inversion
\be
I:x^\mu \to \frac{x^\mu}{x^2}.
\ee
$I$ is a conformal transformation, but it is not continuously connected to the identity, so it can't be obtained by exponentiating a conformal Killing vector. This means that a CFT need not have $I$ as a symmetry.
\begin{exercise}
Show that $I$ is continuously connected to a reflection $x^0\to -x^0$.  Conclude that a CFT is invariant under $I$ if and only if it is invariant under reflections.
\end{exercise}

\begin{exercise}
Show that $k_\mu = -I p_\mu I$. Conclude that $e^{a\.k}$ implements the transformation
\be
\label{eq:finitespecialconformal}
x &\to& x'(x)=\frac{x^\mu-a^\mu x^2}{1-2(a\.x)+a^2 x^2}.
\ee
\end{exercise}
We can think of $k_\mu$ as a ``translation that moves infinity and fixes the origin" in the same sense that the usual translations move the origin and fix infinity, see figure~\ref{fig:translationnearinfinity}.

\begin{figure}
\begin{center}
\includegraphics[width=0.5\textwidth]{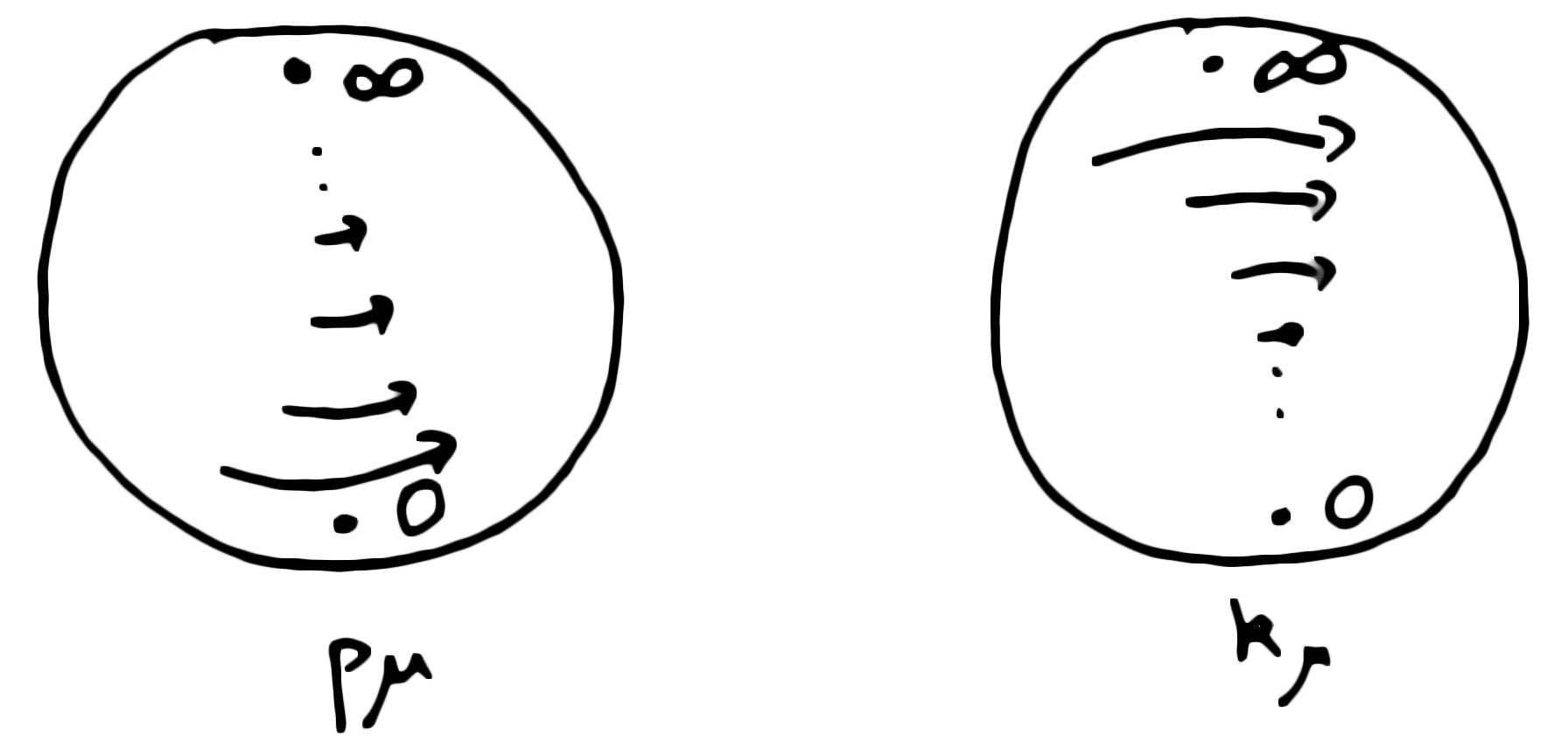}
\end{center}
\caption{\label{fig:translationnearinfinity} $k_\mu$ is analogous to $p_\mu$, with the origin and the point at infinity swapped by an inversion.}
\end{figure}

\subsection{The Conformal Algebra}

The charges $Q_\e$ give a representation of the conformal algebra
\be
\label{eq:conformalalgebra}
[Q_{\e_1},Q_{\e_2}] &=& Q_{-[\e_1,\e_2]},
\ee
where $[\e_1,\e_2]$ is a commutator of vector fields.\footnote{The minus sign in (\ref{eq:conformalalgebra}) comes from the fact that when charges $Q_i$ are represented by differential operators $\cD_i$, repeated action reverses the order $[Q_1,[Q_2,\cO]]=\cD_2 \cD_1 \cO$.  Alternatively, we could have introduced an extra minus sign in the $Q$'s so that $[Q,\cO]=-\cD$ and then $Q,\cD$ would have the same commutation relations.}  This is not obvious and deserves proof. In fact, it is {\it not true\/} in 2-dimensional CFTs, where the algebra of charges is a central extension of the the algebra of conformal killing vectors.

\begin{exercise}
Show that in $d\geq 3$,
\be
\label{eq:conformaltransfofT}
[Q_\e,T^{\mu\nu}] &=& \e^\rho\ptl_\rho T^{\mu\nu}+(\ptl\.\e)T^{\mu\nu}-\ptl_\rho \e^\mu T^{\rho\nu}+\ptl^\nu\e_\rho T^{\rho\mu}.
\ee
Argue as follows. Assume that only the stress tensor appears on the right-hand side.\footnote{Bonus exercise: can other operators appear?} Using linearity in $\e$, dimensional analysis, and the conformal Killing equation, show that (\ref{eq:conformaltransfofT}) contains all terms that could possibly appear.\footnote{The terms on the right-hand side are local in $\e$ because we can evaluate $[Q_{\e},T^{\mu\nu}(x)]$ in an arbitrarily small neighborhood of $x$. Assuming the singularity as two $T^{\mu\nu}$'s coincide is bounded, we can then replace $\e$ by its Taylor expansion around $x$.}  Fix the relative coefficients using conservation, tracelessness, and symmetry under $\mu\leftrightarrow \nu$. Fix the overall coefficient by matching with (\ref{commutator}).
\end{exercise}

\begin{exercise}
Using (\ref{eq:conformaltransfofT}), prove the commutation relation (\ref{eq:conformalalgebra}).
\end{exercise}

\begin{exercise}
When $d=2$, it's possible to add an extra term in (\ref{eq:conformaltransfofT}) proportional to the unit operator that is consistent with dimensional analysis, conservation, and tracelessness.  Find this term (up to an overall coefficient),\footnote{The coefficient can be fixed by comparing with the OPE, see e.g. \cite{Polchinski:1998rq}. It is proportional to the central charge $c$.} and show how it modifies the commutation relations (\ref{eq:conformalalgebra}). This is the Virasoro algebra!
\end{exercise}

As usual, (\ref{eq:conformalalgebra}) is true in any quantization of the theory.  In path integral language, it tells us how to move the topological surface operators $Q_{\e}(\Sigma)$ through each other.

\begin{exercise} Show that
\be
\,[M_{\mu\nu},P_\rho] &=& \de_{\nu\rho}P_\mu - \de_{\mu\rho}P_\nu,\\
\,[M_{\mu\nu},K_\rho] &=& \de_{\nu\rho}K_\mu - \de_{\mu\rho}K_\nu,\\
\,[M_{\mu\nu},M_{\rho\s}] &=& \de_{\nu\rho}M_{\mu\s}-\de_{\mu\rho}M_{\nu\s}+\de_{\nu\s}M_{\rho\mu}-\de_{\mu\s}M_{\rho\nu},\label{eq:mmcommutator}\\
\label{eq:dpcommutator}
\,[D,P_\mu]&=&P_\mu,\\
\label{eq:dkcommutator}
\,[D,K_\mu]&=&-K_\mu,\\
\,[K_\mu,P_\nu]&=&2\de_{\mu\nu}D-2M_{\mu\nu},
\ee
and all other commutators vanish.
\end{exercise}
The first three commutation relations say that $M_{\mu\nu}$ generates the algebra of Euclidean rotations $\SO(d)$ and that $P_\mu,K_\mu$ transform as vectors.  The last three are more interesting.  Equations~(\ref{eq:dpcommutator}) and (\ref{eq:dkcommutator}) say that $P_\mu$ and $K_\mu$ can be thought of as raising and lowering operators for $D$. We will return to this idea shortly.

\begin{exercise} Define the generators
\be
\label{eq:conformalgeneratorssodplus11}
L_{\mu\nu}&=&M_{\mu\nu},\nn\\
L_{-1,0} &=& D,\nn\\
L_{0,\mu} &=& \frac 1 2 (P_\mu+K_\mu),\nn\\
L_{-1,\mu}&=& \frac 1 2 (P_\mu-K_\mu),
\ee
where $L_{ab}=-L_{ba}$ and $a,b\in \{-1,0,1,\dots,d\}$.  Show that $L_{ab}$ satisfy the commutation relations of $\SO(d+1,1)$.
\end{exercise}
The fact that the conformal algebra is $\SO(d+1,1)$ suggests that it might be good to think about its action in terms of $\R^{d+1,1}$ instead of $\R^d$.  This is the idea behind the ``embedding space formalism" \cite{Dirac:1936fq,Mack:1969rr,Boulware:1970ty,Ferrara:1973eg,Weinberg:2010fx,Costa:2011mg}, which provides a simple and powerful way to understand the constraints of conformal invariance. We will be more pedestrian in this course, but I recommend reading about the embedding space formalism in the lecture notes by Penedones \cite{Joao} or Rychkov \cite{Rychkov:2016iqz}.

\section{Primaries and Descendants}

Now that we have our conserved charges, we can classify operators into representations of those charges.  We do this in steps. First we classify operators into Poincare representations, then scale+Poincare representations, and finally conformal representations.

\subsection{Poincare Representations}

In a rotationally-invariant QFT, local operators at the origin transform in irreducible representations of the rotation group $\SO(d)$,
\be
\label{eq:rotationatorigin}
[M_{\mu\nu},\cO^a(0)]&=& (\cS_{\mu\nu}){}_b{}^a\cO^b(0),
\ee
where $\cS_{\mu\nu}$ are matrices satisfying the same algebra as $M_{\mu\nu}$, and $a,b$ are indices for the $\SO(d)$ representation of $\cO$.\footnote{The funny index contractions in (\ref{eq:rotationatorigin}) ensure that $M_{\mu\nu}$ and $\cS_{\mu\nu}$ have the same commutation relations (exercise!).}\footnote{Because our commutation relations (\ref{eq:mmcommutator}) for $\SO(d)$ differ from the usual conventions by a factor of $i$, the generators $\cS_{\mu\nu}$ will be {\it anti-}hermitian, $\cS^\dag=-\cS$.}  We often suppress spin indices and write the right-hand side as simply $\cS_{\mu\nu}\cO(0)$. The action (\ref{eq:rotationatorigin}), together with the commutation relations of the Poincare group, determines how rotations act away from the origin.

\begin{figure}
\begin{center}
\includegraphics[width=0.5\textwidth]{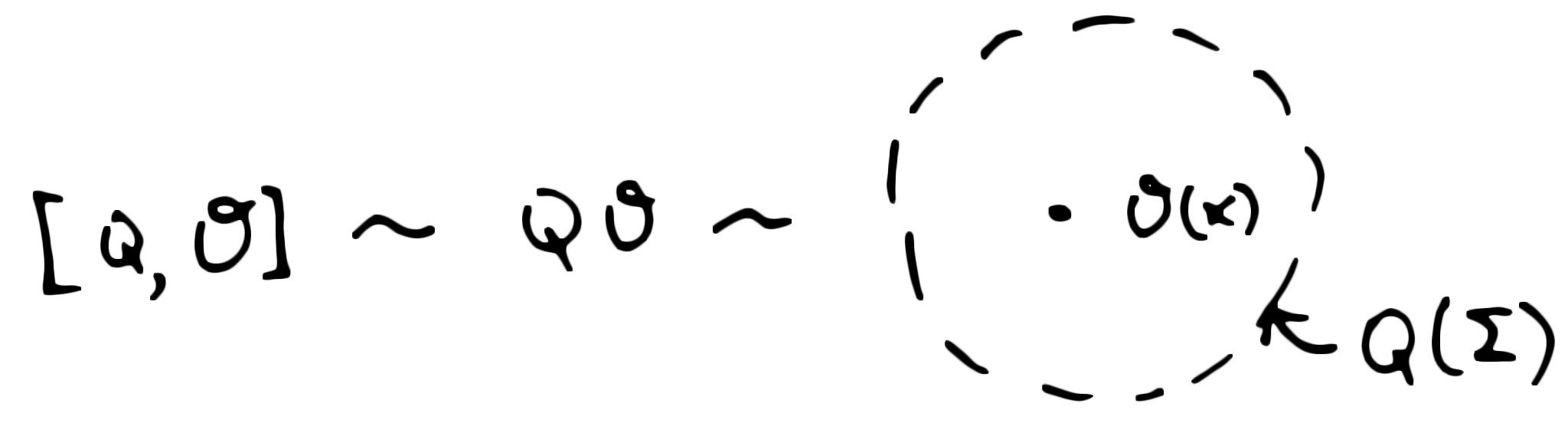}
\includegraphics[width=0.5\textwidth]{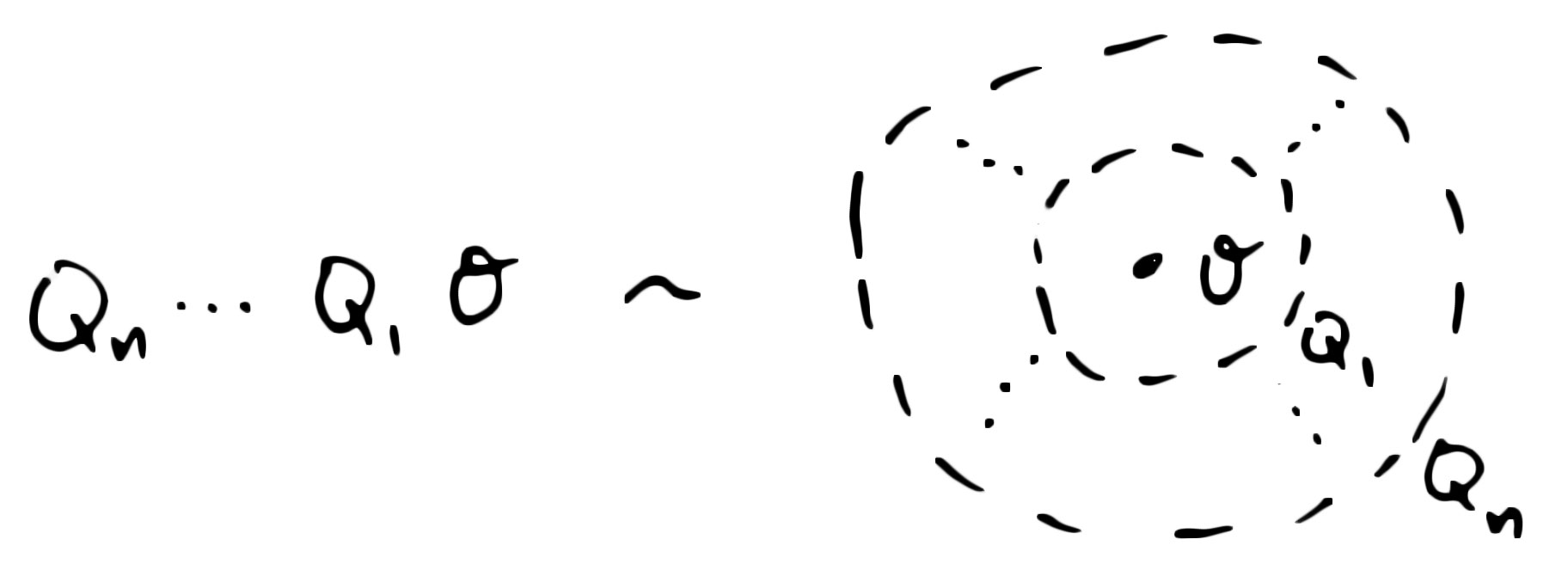}
\end{center}
\caption{The shorthand notation $Q\cO$ stands for surrounding $\cO$ with a surface operator $Q(\Sigma)$. Equivalently, it stands for $[Q,\cO]$ in any quantization of the theory. \label{fig:commutatorissurround}}
\end{figure}

To see this, it is convenient to adopt shorthand notation where commutators of charges with local operators are implicit, $[Q,\cO] \to Q \cO$, see figure~\ref{fig:commutatorissurround}.  This notation is valid because of the Jacobi identity (more formally, the fact that adjoint action gives a representation of a Lie algebra).  In path integral language, $Q_n\cdots Q_1 \cO(x)$ means surrounding $\cO(x)$ with topological surface operators where $Q_n$ is the outermost surface and $Q_1$ is the innermost.  The conformal commutation relations tell us how to re-order these surfaces.

Acting with a rotation on $\cO(x)$, we have
\be
M_{\mu\nu}\cO(x) &=& M_{\mu\nu}e^{x\.P}\cO(0) \nn\\
&=& e^{x\.P}(e^{-x\.P} M_{\mu\nu} e^{x\.P})\cO(0)\nn\\
&=& e^{x\.P}(-x_\mu P_\nu + x_\nu P_\mu+M_{\mu\nu})\cO(0)\nn\\
&=& (x_\nu \ptl_\mu - x_\mu\ptl_\nu+\cS_{\mu\nu})e^{x\.P}\cO(0)\nn\\
&=& (m_{\mu\nu}+\cS_{\mu\nu})\cO(x).\label{eq:actionbyrotation}
\ee
In the third line, we've used the Poincare algebra and the Hausdorff formula
\be
e^{A}Be^{-A} = e^{[A,\.]}B
= B+[A,B]+\frac 1 {2!}[A,[A,B]]+\dots.
\ee

\subsection{Scale+Poincare Representations}

In a scale-invariant theory, it's also natural to diagonalize the dilatation operator acting on operators at the origin,
\be
\label{eq:dilatationcondition}
[D,\cO(0)]&=&\Delta \cO(0).
\ee
The eigenvalue $\Delta$ is the {\it dimension\/} of $\cO$.
\begin{exercise}
Mimic the computation (\ref{eq:actionbyrotation}) to derive the action of dilatation on $\cO(x)$ away from the origin,
\be
\label{eq:dilatationaction}
[D,\cO(x)] &=& (x^\mu\ptl_\mu + \Delta)\cO(x).
\ee
\end{exercise}

Equation (\ref{eq:dilatationaction}) is constraining enough to fix two-point functions of scalars up to a constant.  Firstly, by rotation and translation invariance, we must have
\be
\<\cO_1(x)\cO_2(y)\>&=&f(|x-y|),
\ee
for some function $f$.

In a scale-invariant theory with scale-invariant boundary conditions, the simultaneous action of $D$ on all operators in a correlator must vanish, as illustrated in figure~\ref{fig:wardidentityford}.  Moving $D$ to the boundary gives zero.\footnote{It is also interesting to consider non-scale-invariant boundary conditions. These can be interpreted as having a nontrivial operator at $\oo$.}  On the other hand, shrinking $D$ to the interior gives the sum of its actions on the individual operators.  By the Ward identity (\ref{eq:dilatationaction}), this is
\be
\label{eq:wardidentityforcorrelator}
0 &=& \p{x^\mu\ptl_\mu + \Delta_1+y^\mu\ptl_\mu+\Delta_2}f(|x-y|).
\ee
We could alternatively derive (\ref{eq:wardidentityforcorrelator}) by working in some quantization, where it follows from trivial algebra and the fact that $D|0\> = 0$,
\be
0 &=& \<0|[D,\cO_1(x)\cO_2(y)]|0\>\nn\\
&=& \<0|[D,\cO_1(x)]\cO_2(y)+\cO_1(x)[D,\cO_2(y)]|0\>\nn\\
&=& \p{x^\mu\ptl_\mu + \Delta_1+y^\mu\ptl_\mu+\Delta_2}\<0|\cO_1(x)\cO_2(y)|0\>.
\ee
Either way, the solution is
\be
f(|x-y|) &=& \frac{C}{|x-y|^{\Delta_1+\Delta_2}}.
\ee

If we had an operator with negative scaling dimension, then its correlators would grow with distance, violating cluster decomposition. This is unphysical, so we expect dimensions $\De$ to be positive. Shortly, we will prove this fact for unitary conformal theories (and derive even stronger constraints on $\De$).

\begin{figure}
\begin{center}
\includegraphics[width=0.75\textwidth]{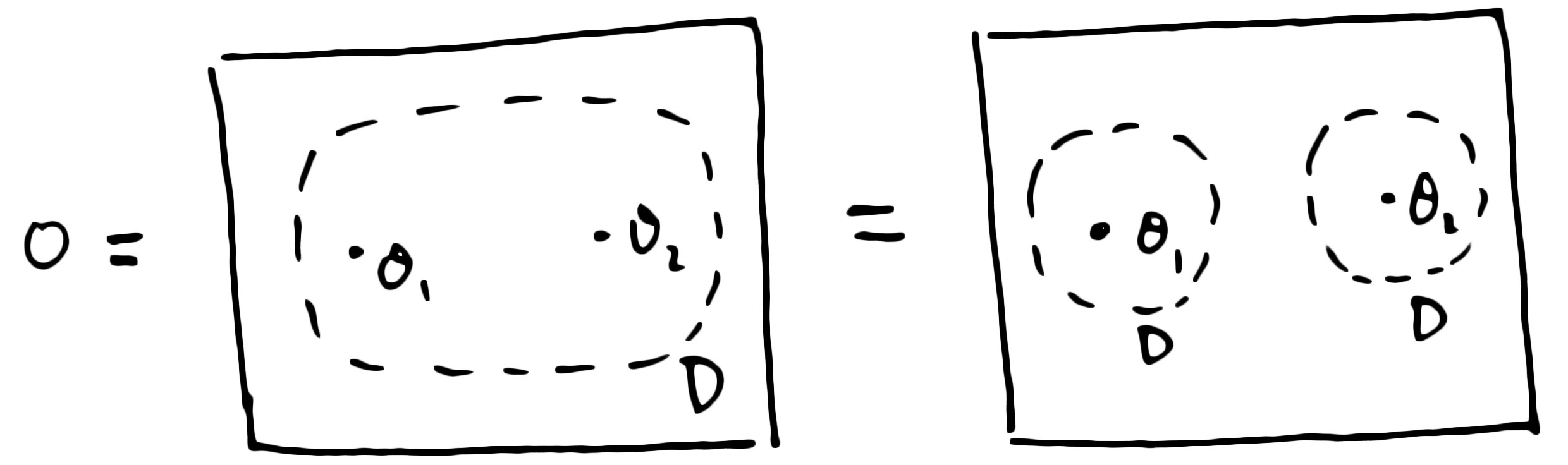}
\end{center}
\caption{The Ward identity for scale invariance of a two-point function. \label{fig:wardidentityford}}
\end{figure}

\subsection{Conformal Representations}

Note that $K_\mu$ is a lowering operator for dimension,
\be
D K_\mu \cO(0) &=& ([D,K_\mu] + K_\mu D)\cO(0)\nn\\
&=& (\De-1)K_\mu \cO(0).
\ee
(Again, we're using shorthand notation $[Q,\cO]\to Q\cO$.)  Thus, given an operator $\cO(0)$, we can repeatedly act with $K_\mu$ to obtain operators $K_{\mu_1}\dots K_{\mu_n}\cO(0)$ with arbitrarily low dimension.  Because dimensions are bounded from below in physically sensible theories, this process must eventually terminate.  That is, there must exist operators such that
\be
\label{eq:primarycondition}
[K_\mu,\cO(0)] &=& 0\qquad\textrm{(primary operator)}.
\ee
Such operators are called ``primary."  Given a primary, we can construct operators of higher dimension, called ``descendants," by acting with momentum generators, which act like raising operators for dimension,
\be
\cO(0) &\to& P_{\mu_1}\cdots P_{\mu_n}\cO(0)\qquad\textrm{(descendant operators)}\nn\\
\De &\to& \De+n.
\ee
For example, $\cO(x)=e^{x\.P}\cO(0)$ is an (infinite) linear combination of descendant operators.
The conditions (\ref{eq:rotationatorigin},~\ref{eq:dilatationcondition},~\ref{eq:primarycondition}) are enough to determine how $K_\mu$ acts on any descendant using the conformal algebra.  For example,
\begin{exercise}
Let $\cO(0)$ be a primary operator with rotation representation matrices $\cS_{\mu\nu}$ and dimension $\Delta$.  Using the conformal algebra, show
\be
[K_\mu, \cO(x)] &=& (k_\mu + 2\De x_\mu - 2x^\nu \cS_{\mu\nu})\cO(x),
\label{eq:actionofK}
\ee
where $k_\mu$ is the conformal Killing vector defined in~(\ref{eq:extraconformalgenerators}). 
\end{exercise}

To summarize, a primary operator satisfies
\be
\label{eq:isotropyaction}
\,[D,\cO(0)] &=& \De\cO(0)\nn\\
\,[M_{\mu\nu},\cO(0)] &=& \cS_{\mu\nu}\cO(0)\nn\\
\,[K_\mu,\cO(0)] &=& 0.
\ee
From these conditions, we can construct a representation of the conformal algebra out of 
$\cO(0)$ and its descendants,
\be
\label{eq:conformalrepresentation}
\begin{array}{c|c}
\textrm{operator} & \textrm{dimension}
\\
\hline
\vdots & \\
P_{\mu_1}P_{\mu_2}\cO(0) & \De+2\\
\uparrow & \\
P_{\mu_1} \cO(0) & \De+1\\
\uparrow &\\
\cO(0) & \De.
\end{array}
\ee
The action of conformal generators on each state follows from the conformal algebra.  This should remind you of the construction of irreducible representations of $\SU(2)$ starting from a highest-weight state.  In this case, our primary is a {\it lowest-weight\/} state of $D$, but the representation is built in an analogous way.\footnote{Generically, the representation (\ref{eq:conformalrepresentation}) is an {\it induced representation} $\mathrm{Ind}^G_H(R_H)$, where $H$ is the subgroup of the conformal group generated by $D,M_{\mu\nu},K_\mu$ (called the isotropy subgroup), $R_H$ is the finite-dimensional representation of $H$ defined by (\ref{eq:isotropyaction}), and $G$ is the full conformal group. It is also called a parabolic Verma module.  Sometimes the operator $\cO$ satisfies ``shortening conditions" where a linear combination of descendants vanishes. (A conserved current is an example.)  In this case, the Verma module is reducible and the actual conformal multiplet of $\cO$ is one of the irreducible components.}  It turns out that any local operator in a unitary CFT is a linear combination of primaries and descendants. We will prove this in section~\ref{sec:onlyprimariesanddescendants}.

\begin{exercise}
Show that (\ref{commutator}), (\ref{eq:actionbyrotation}), (\ref{eq:dilatationaction}), and (\ref{eq:actionofK}) can be summarized as
\be
\label{eq:generatorsummary}
[Q_\e,\cO(x)] &=& \p{\e\.\ptl + \frac{\De}{d}(\ptl\.\e) - \frac 1 2 (\ptl^\mu \e^\nu)\cS_{\mu\nu}}\cO(x).
\ee
\end{exercise}

\begin{exercise}
Deduce that $T^{\mu\nu}$ is primary by comparing (\ref{eq:generatorsummary}) with (\ref{eq:conformaltransfofT}). 
\end{exercise}

\subsection{Finite Conformal Transformations}

An exponentiated charge $U=e^{Q_\e}$ implements a finite conformal transformation.  Denote the corresponding diffeomorphism $e^{\e}$ by $x\mapsto x'(x)$. By comparing with (\ref{eq:conformalinfinitesimal}) and (\ref{eq:conformalfinite}), we find that (\ref{eq:generatorsummary}) exponentiates to
\be
\label{eq:finiteprimarytransformation}
U \cO^a(x) U^{-1} &=& \Omega(x')^\De D(R(x'))_b{}^a\cO^b(x'),
\ee
where as before
\be
\pdr{x'^\mu}{x^\nu} &=& \Omega(x')R^\mu{}_\nu(x'),\qquad R^\mu{}_\nu(x')\in\SO(d).
\ee
Here, $D(R)_b{}^a$ is a matrix implementing the action of $R$ in the $\SO(d)$ representation of $\cO$, for example
\begin{align}
D(R) &= 1 && \textrm{(scalar representation)},\nn\\
D(R)_\mu{}^\nu &= R_\mu{}^\nu && \textrm{(vector representation)},\label{eq:vectorrep}\nn\\
 &\cdots&& \cdots
\end{align}
and so on.

We could have started the whole course by taking (\ref{eq:finiteprimarytransformation}) as the definition of a primary operator. But the connection to the underlying conformal algebra will be crucial in what follows, so we have chosen to derive it.

\begin{exercise}
Show that the transformation (\ref{eq:finiteprimarytransformation}) composes correctly to give a representation of the conformal group.  That is, show
\be
U_{g_1}U_{g_2}\cO^a(x) U_{g_2}^{-1} U_{g_1}^{-1} &=& U_{g_1g_2}\cO^a(x)U_{g_1g_2}^{-1}
\ee
where $x\mapsto g_{i}(x)$ are conformal transformations, $g_1g_2$ denotes composition $x\mapsto g_1(g_2(x))$, and $U_{g}$ is the unitary operator associated to $g$.
\end{exercise}

\section{Conformal Correlators}

\subsection{Scalar Operators}
\label{sec:conformalcorrelatorsscalars}

We have already seen that scale invariance fixes two-point functions of scalars up to a constant
\be
\label{eq:scaletwoptfunction}
\<\cO_1(x_1)\cO_2(x_2)\> &=& \frac{C}{|x_1-x_2|^{\De_1+\De_2}} \qquad\textrm{(SFT)}.
\ee

For primary scalars in a CFT, the correlators must satisfy a stronger Ward identity,
\be
\label{eq:scalarconformalcorrelator}
\<\cO_1(x_1)\dots\cO_n(x_n)\> &=& \<(U\cO_1(x_1)U^{-1}) \cdots (U\cO_n(x_n)U^{-1})\>\nn\\
 &=& \Omega(x_1')^{\De_1}\cdots\Omega(x_n')^{\De_n}\<\cO_1(x_n')\cdots \cO_n(x_n')\>.
\ee
Let us check whether this holds for (\ref{eq:scaletwoptfunction}).

\begin{exercise}
Show that for a conformal transformation,
\be
\label{eq:conformaltransformationofdistance}
(x-y)^2 &=& \frac{(x'-y')^2}{\Omega(x')\Omega(y')}.
\ee
Hint: This is obviously true for translations, rotations, and scale transformations. It suffices to check it for inversions $I:x\to\frac{x}{x^2}$ (why?).
\end{exercise}
Using (\ref{eq:conformaltransformationofdistance}), we find
\be
\frac{C}{|x_1-x_2|^{\De_1+\De_2}} &=& \Omega(x_1')^{\frac{\De_1+\De_2}{2}}\Omega(x_2')^{\frac{\De_1+\De_2}{2}}\frac{C}{|x_1'-x_2'|^{\De_1+\De_2}}.
\ee
Consistency with (\ref{eq:scalarconformalcorrelator}) then requires $\De_1=\De_2$ or $C=0$.  In other words,
\be
\<\cO_1(x_1)\cO_2(x_2)\> &=& \frac{C\de_{\De_1\De_2}}{x_{12}^{2\De_1}}\qquad\textrm{(CFT, primary operators)},
\ee
where $x_{12}\equiv x_1-x_2$.
\begin{exercise}
Recover the same result using the Ward identity for $K_\mu$
\be
\<[K_\mu,\cO_1(x_1)]\cO_2(x_2)\>+\<\cO_1(x_1)[K_\mu,\cO_2(x_2)]\> &=& 0.
\ee
\end{exercise}

Conformal invariance is also powerful enough to fix a three-point function of primary scalars, up to an overall coefficient.  Using (\ref{eq:conformaltransformationofdistance}), it's easy to check that the famous formula \cite{Polyakov:1970xd}
\be
\label{eq:conformalthreeptfunction}
\<\cO_1(x_1)\cO_2(x_2)\cO_3(x_3)\> &=& \frac{f_{123}}{x_{12}^{\De_1+\De_2-\De_3}x_{23}^{\De_2+\De_3-\De_1}x_{31}^{\De_3+\De_1-\De_2}},
\ee
with $f_{123}$ constant, satisfies the Ward identity (\ref{eq:scalarconformalcorrelator}).

With four points, there are nontrivial conformally invariant combinations of the points called ``conformal cross-ratios,"
\be
\label{eq:definitionofcrossratios}
u = \frac{x_{12}^2 x_{34}^2}{x_{13}^2 x_{24}^2},\qquad
v = \frac{x_{23}^2 x_{14}^2}{x_{13}^2 x_{24}^2}.
\ee
The reason that there are exactly two independent cross-ratios can be understood as follows.
\begin{itemize}
\item Using special conformal transformations, we can move $x_4$ to infinity.
\item Using translations, we can move $x_1$ to zero.
\item Using rotations and dilatations, we can move $x_3$ to $(1,0,\dots,0)$.
\item Using rotations that fix $x_3$, we can move $x_2$ to $(x,y,0,\dots,0)$.
\end{itemize}

\begin{figure}
\begin{center}
\includegraphics[width=0.4\textwidth]{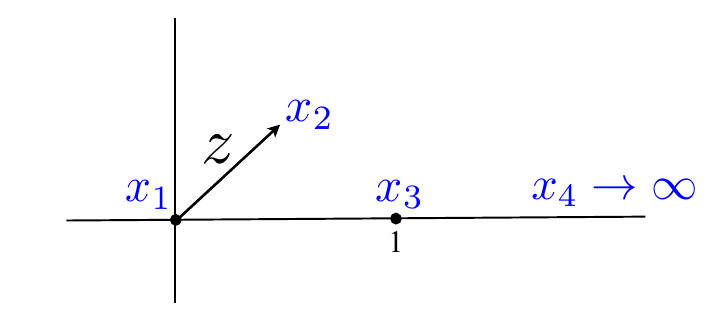}
\end{center}
\caption{\label{fig:zplane} Using conformal transformations, we can place four points on a plane in the configuration shown above (figure from \cite{Hogervorst:2013sma}).}
\end{figure}

This procedure leaves exactly two undetermined quantities $x,y$, giving two independent conformal invariants. Evaluating $u$ and $v$ for this special configuration of points (figure~\ref{fig:zplane}) gives
\be
u=z\bar z,\qquad v=(1-z)(1-\bar z),
\ee
where $z\equiv x+iy$.

Four-point functions can depend nontrivially on the cross-ratios.  For a scalar $\f$ with dimension $\De_\phi$, the formula
\be
\label{eq:fourptfunctionofprimaries}
\<\f(x_1)\f(x_2)\f(x_3)\f(x_4)\> &=& \frac{g(u,v)}{x_{12}^{2\De_\f}x_{34}^{2\De_\f}}
\ee
satisfies the Ward identity (\ref{eq:scalarconformalcorrelator}) for any function $g(u,v)$. 
\begin{exercise}
Generalize (\ref{eq:fourptfunctionofprimaries}) to the case of non-identical scalars $\f_i(x)$ with dimensions $\De_i$.
\end{exercise}

The left-hand side of (\ref{eq:fourptfunctionofprimaries}) is manifestly invariant under permutations of the points $x_i$.  This leads to consistency conditions on $g(u,v)$,
\begin{align}
\label{eq:trivialcrossing}
g(u,v) &= g(u/v,1/v) && \textrm{(from swapping $1\leftrightarrow 2$ or $3\leftrightarrow 4$)},\\
\label{eq:crossingsymmetry}
g(u,v) &= \p{\frac{u}{v}}^{\De_\f} g(v,u) && \textrm{(from swapping $1\leftrightarrow 3$ or $2\leftrightarrow 4$)}.
\end{align}
All other permutations can be generated from the ones above.  We will see shortly that $g(u,v)$ is  actually determined in terms of the dimensions $\De_i$ and three-point coefficients $f_{ijk}$ of the theory.  Equation~(\ref{eq:trivialcrossing}) will be satisfied for trivial reasons.  However (\ref{eq:crossingsymmetry}) will lead to powerful constraints on the $\De_i, f_{ijk}$.

\subsection{Spinning Operators}

The story is similar for operators with spin.  For brevity, we give the answers without doing any computations.  The embedding space formalism provides a transparent and practical way to derive all of these results \cite{Costa:2011mg}, so it's not worth dwelling on them here.

Two-point functions of spinning operators are fixed by conformal invariance.  They are nonzero only if the operators have identical dimensions and spins.  For example, a two-point function of spin-1 operators with dimension $\Delta$ is given by
\be
\label{eq:twoptfunctionofspin1}
\<J^\mu(x)J_\nu(y)\> &=& C_J \frac{I^\mu{}_{\nu}(x-y)}{(x-y)^{2\De}},\\
 I^{\mu}{}_{\nu}(x)&\equiv& \de^\mu_\nu-2\frac{x^\mu x_\nu}{x^2},
 \label{eq:Itensor}
\ee
where $C_J$ is a constant. Note that $I^\mu{}_{\nu}(x)$ is the orthogonal matrix associated with an inversion, $\pdr{x'^\mu}{x^\nu}=\Omega(x) I^\mu{}_\nu(x)$.
\begin{exercise}
Check that (\ref{eq:twoptfunctionofspin1}) is consistent with conformal symmetry.  Hint: it is enough to check inversions.
\end{exercise}

Two-point functions of operators in more general spin representations can be constructed from the above.  For spin-$\ell$ traceless symmetric tensors,
\be
\label{eq:twopointfunctionofspinL}
\<J^{\mu_1\dots\mu_\ell}(x)J_{\nu_1\dots\nu_\ell}(0)\> &=& C_J \p{\frac{I^{(\mu_1}{}_{\nu_1}(x)\cdots I^{\mu_\ell)}{}_{\nu_\ell}(x)}{x^{2\De}} - \mathrm{traces}},
\ee
where we can symmetrize either the $\mu$'s or $\nu$'s (or both).  Subtracting traces means adding terms proportional to $\de^{\mu_i\mu_j}$ and $\de_{\nu_i\nu_j}$ so that the result is separately traceless in the $\mu$ indices and the $\nu$ indices (not necessarily under $\mu$-$\nu$ contractions).

It is sometimes conventional to normalize $J$ so that $C_J=1$ in (\ref{eq:twoptfunctionofspin1}), (\ref{eq:twopointfunctionofspinL}).  An exception is if $J$ already has a natural normalization.  For example, the normalization of the stress tensor is fixed by demanding that $T^{\mu\nu}$ satisfy the appropriate Ward identities.  In this case, $C_T$ is physically meaningful.

Three-point functions are fixed up to a finite number of coefficients.  For example, a three-point function of scalars $\phi_{1}$, $\phi_2$ and a spin-$\ell$ operator $J_{\mu_1\dots\mu_\ell}$ is determined up to a single coefficient $f_{\f_1\f_2 J}$,
\be
\label{eq:scalarscalarspinL}
\<\f_1(x_1)\f_2(x_2)J^{\mu_1\dots\mu_\ell}(x_3)\> &=& \frac{f_{\f_1\f_2 J}(Z^{\mu_1}\cdots Z^{\mu_\ell} - \mathrm{traces})}{x_{12}^{\De_1+\De_2-\De_3+\ell}x_{23}^{\De_2+\De_3-\De_1-\ell}x_{31}^{\De_3+\De_1-\De_2-\ell}},\nn\\
Z^\mu &\equiv& \frac{x_{13}^\mu}{x_{13}^2}-\frac{x_{23}^\mu}{x_{23}^2}.
\ee
When multiple operators have spin, there can be more than one linearly independent structure consistent with conformal invariance.

Formula (\ref{eq:scalarscalarspinL}) applies when $J^{\mu\nu}$ is the stress tensor.  In that case, the coefficient $f_{\phi_1\phi_2 T}$ is fixed by demanding that integrals of $T^{\mu\nu}$ give the correct action of the conformal charges $Q_\e$ (see the exercise in Jo\~ao Penedones' notes \cite{Joao}). The result is 
\be
\label{eq:stresstensorward}
f_{\phi_1\phi_2 T} &=& -\frac{d\De_1}{d-1}\frac 1 {S_d} C_{12},
\ee
where $S_d$ is the volume of the unit sphere $S^{d-1}$ and $C_{12}$ is the coefficient in the two-point function $\<\phi_1(x)\phi_2(0)\>=C_{12}x^{-2\De_1}$ (note $C_{12}$ vanishes unless $\Delta_1=\Delta_2$). The coefficient $f_{\phi_1\phi_2 J}$ is fixed by Ward identities whenever $J$ is a conserved current.

\section{Radial Quantization and the State-Operator Correspondence}

So far, we've written lots of commutation relations, and carefully pointed out that they are true in any quantization of the theory. Now we'll really put that idea to use.  In general, we should to choose quantizations that respect symmetries.  In a scale-invariant theory, it's natural to foliate spacetime with spheres around the origin and consider evolving states from smaller spheres to larger spheres using the dilatation operator (figure~\ref{fig:radialquantization}).  This is called ``radial quantization." The sphere $S^{d-1}$ has an associated Hilbert space $\cH$. We can act on $\cH$ by inserting operators on the surface of the sphere. For example, to act with a symmetry generator $Q$, we insert the surface operator $Q(S^{d-1})$ into the path integral (figure~\ref{fig:chargeactionradialquantization}).

\begin{figure}
\begin{center}
\includegraphics[width=0.35\textwidth]{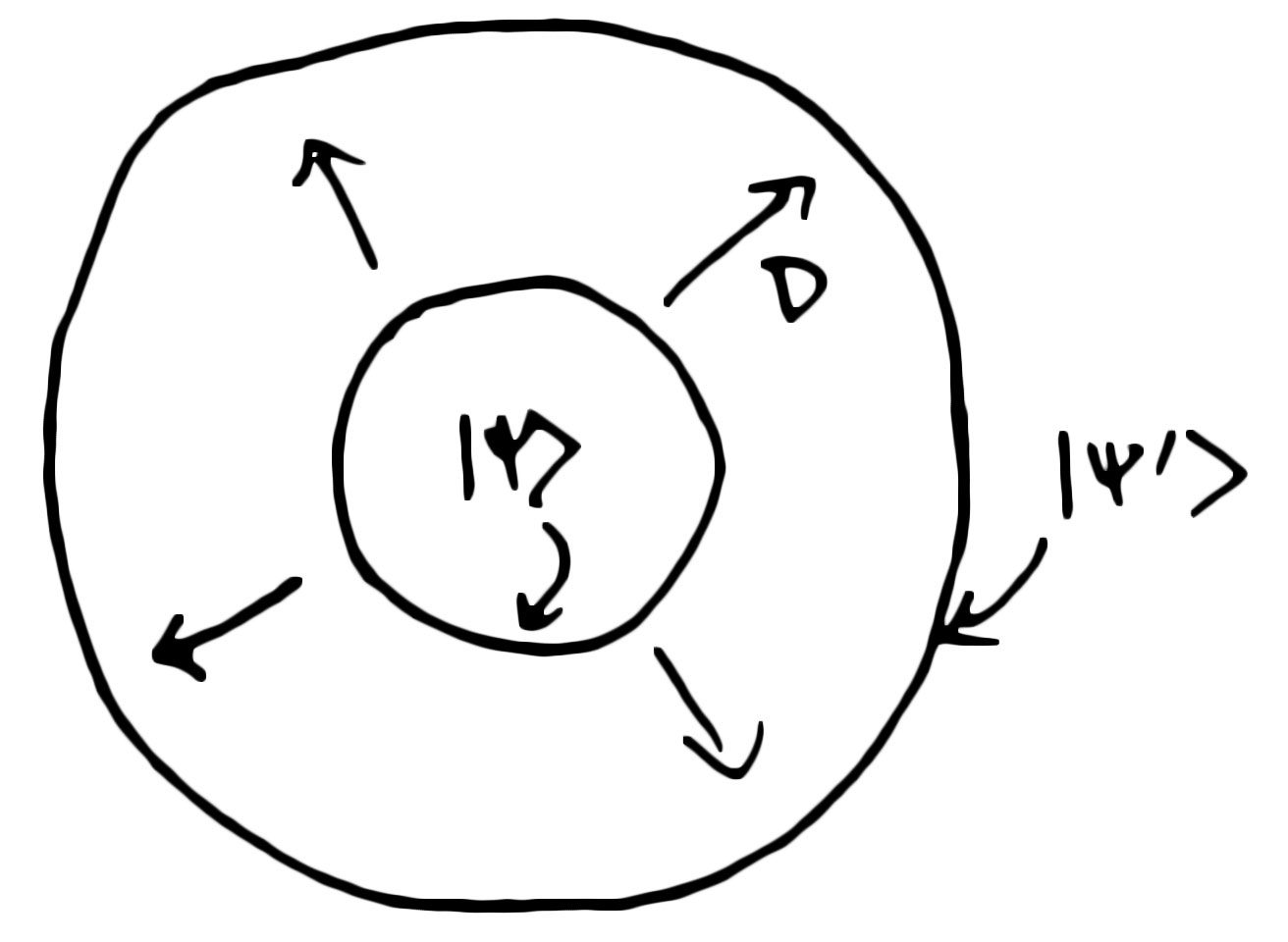}
\end{center}
\caption{In radial quantization, states live on spheres, and we evolve from one state to another with the dilatation operator. \label{fig:radialquantization}}
\end{figure}

\begin{figure}
\begin{center}
\includegraphics[width=0.35\textwidth]{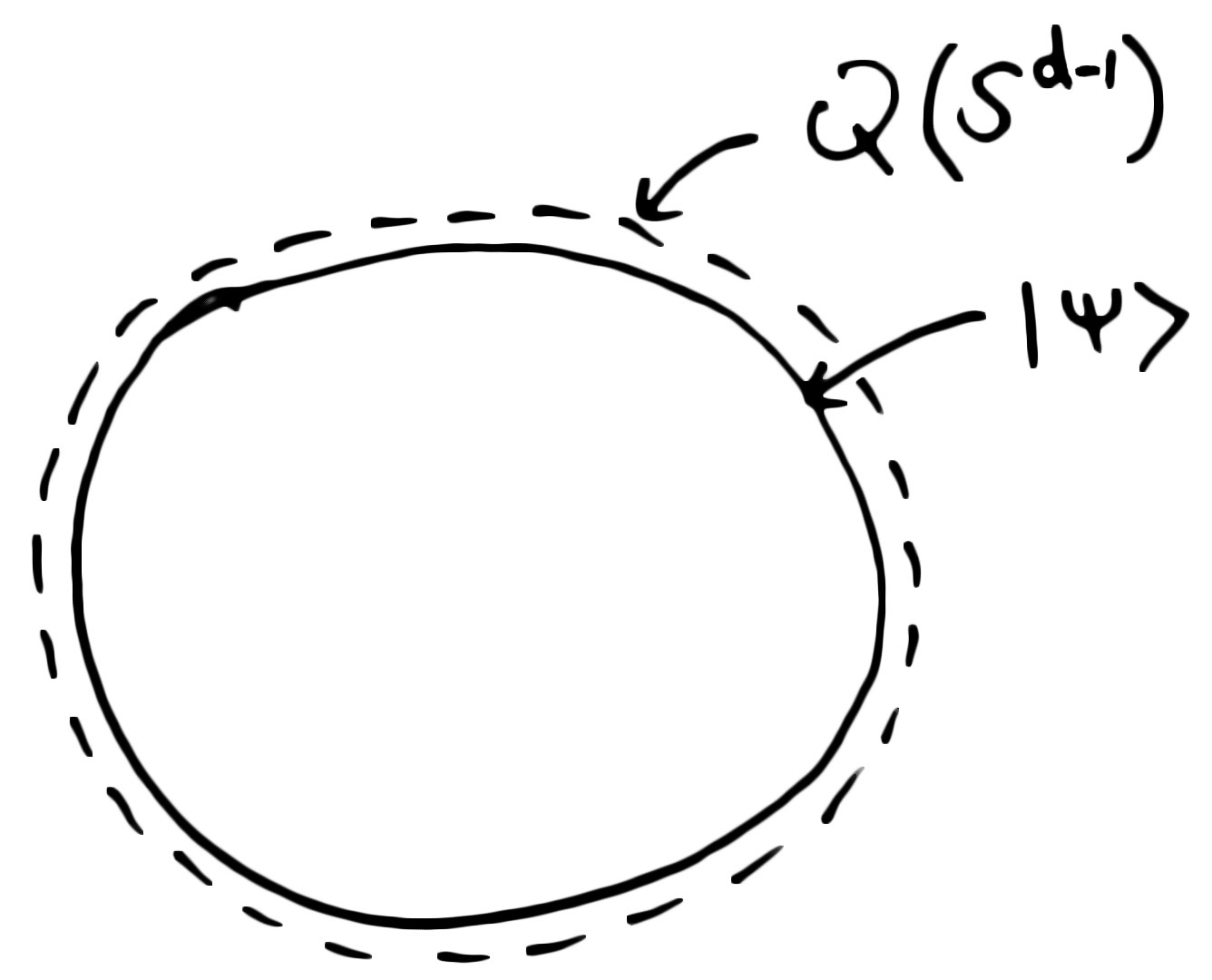}
\end{center}
\caption{We act with a charge in radial quantization by inserting $Q(S^{d-1})$ just outside the sphere on which the state is defined. \label{fig:chargeactionradialquantization}}
\end{figure}

In radial quantization, a correlation function gets interpreted as a radially ordered product,
\be
\<\cO_1(x_1)\cdots \cO_n(x_n)\> &=& \<0|\mathcal{R}\{ \cO_1(x_1)\cdots \cO_n(x_n)\}|0\>\nn\\
&\equiv & \theta(|x_n|-|x_{n-1})\cdots \theta(|x_2|-|x_1|) \<0|\cO(x_n)\cdots\cO(x_1)|0\>\nn\\
&&+\mathrm{permutations}.
\ee
Of course, we can perform radial quantization around different points.  The same correlation function then gets interpreted as an expectation value of differently ordered operators acting on different states in different (but isomorphic) Hilbert spaces (figure~\ref{fig:radialquantdifferentpoints}).  This is completely analogous to changing reference frames in Lorentz invariant theories.  The radial ordering prescription is consistent because operators at the same radius but different angles on the sphere commute, just as spacelike-separated operators commute in the usual quantization.

\begin{figure}
\begin{center}
\includegraphics[width=0.75\textwidth]{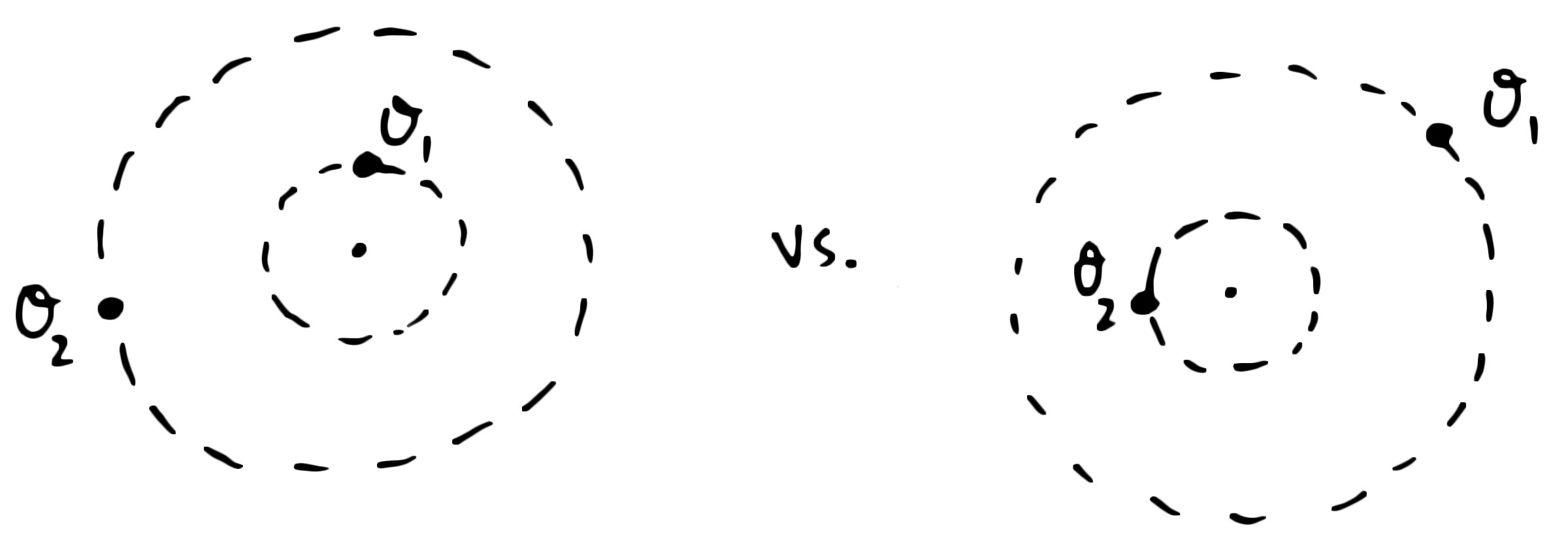}
\end{center}
\caption{When we perform radial quantization around different points, the same correlator gets interpreted as a product of operators with different orderings.  \label{fig:radialquantdifferentpoints}}
\end{figure}

\subsection{Operator $\Longrightarrow$ State}
\label{sec:operatorimpliesstate}

The simplest way to prepare a state in radial quantization is to perform the path integral over the interior $B$ of the sphere, with no operator insertions inside $B$.  This gives the vacuum state $|0\>$ on the boundary $\ptl B$ (figure~\ref{fig:radialvacuum.jpg}).  It's easy to see that $|0\>$ is invariant under all symmetries because a topological surface on the boundary of $B$ can be shrunk to zero inside $B$ (figure~\ref{fig:vacuuminvariant}).

\begin{figure}
\begin{center}
\includegraphics[width=0.3\textwidth]{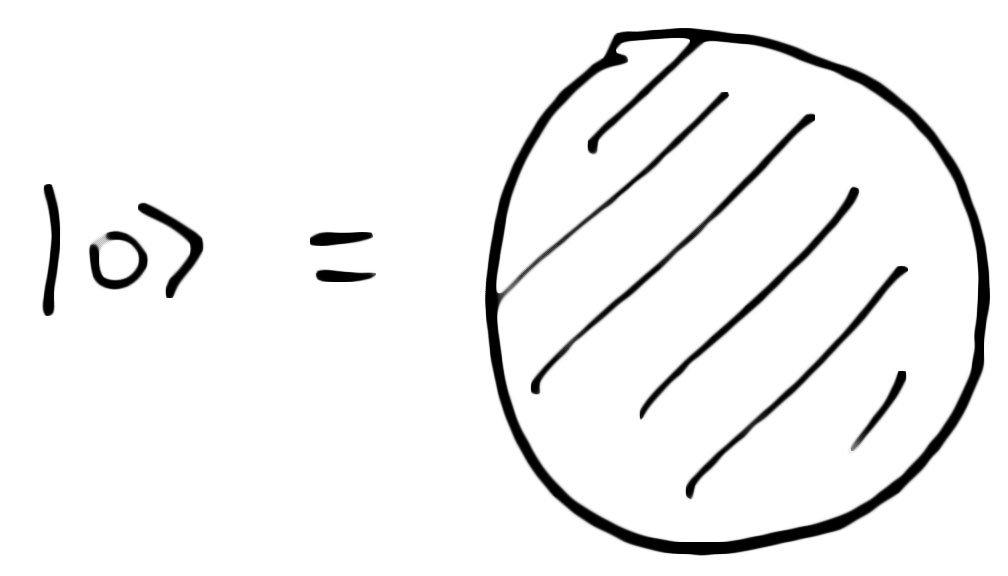}
\end{center}
\caption{The vacuum in radial quantization is given by the path integral over the interior of the sphere, with no operator insertions.  \label{fig:radialvacuum.jpg}}
\end{figure}

\begin{figure}
\begin{center}
\includegraphics[width=0.7\textwidth]{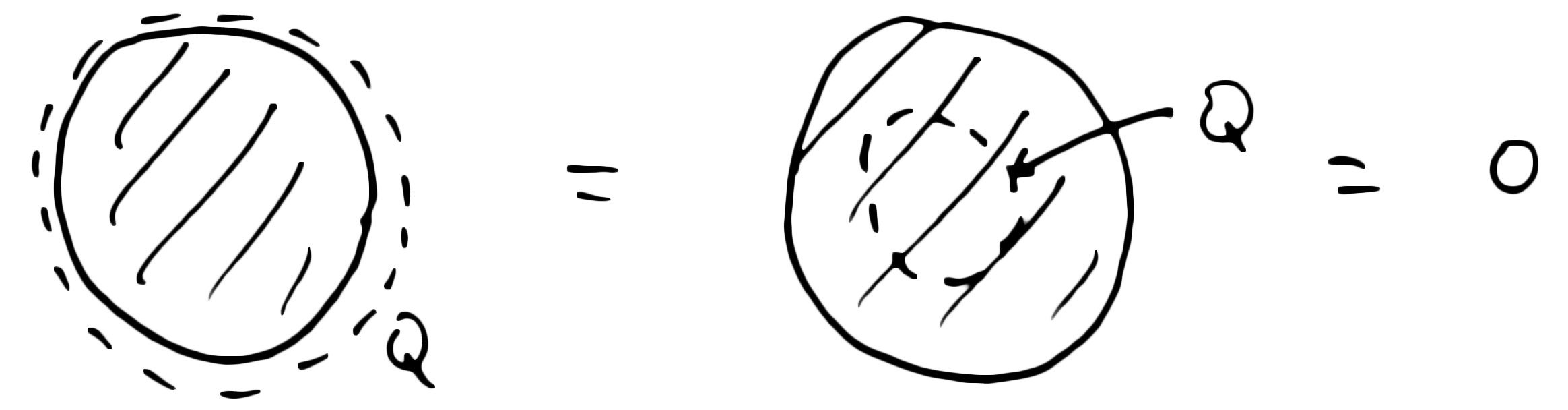}
\end{center}
\caption{The vacuum is automatically invariant under all symmetries.  \label{fig:vacuuminvariant}}
\end{figure}

To be explicit, suppose our CFT is given by the path integral over a scalar field $\phi$.  The Hilbert space in radial quantization is spanned by ``field eigenstates" $|\phi_b\>$, where $\phi_b(\bn)$ is a field configuration on the sphere $\bn\in \ptl B$.  The subscript ``$b$" indicates that $\phi_b$ is defined only on the boundary $\ptl B$ and not in the interior.  A general state is a linear combination of field eigenstates
\be
|\psi\> &\equiv& \int D\phi_b |\phi_b\>\<\phi_b|\psi\>.
\ee
Here, $\int D\phi_b$ represents a $d-1$-dimensional path integral over fields on $\ptl B$.

For the vacuum, the coefficients $\<\phi_b|0\>$ are given by the path integral over the interior with boundary conditions $\phi_b$ and no operator insertions,
\be
\<\phi_b |0\> &=&  \int_{\substack{\phi(1,\bn)=\phi_b(\bn) \\ r \leq 1}} D\phi(r,\bn) e^{-S[\phi]}.
\ee

A more exciting possibility is to insert an operator $\cO(x)$ inside $B$ and then perform the path integral,
\be
\<\phi_b|\cO(x)|0\> &=& \int_{\substack{\phi(1,\bn)=\phi_b(\bn) \\ r \leq 1}} D\phi(r,\bn) \cO(x) e^{-S[\phi]}.
\ee
This defines a state called $\cO(x)|0\>$, see figure~\ref{fig:radialexcited}.
By inserting different operators inside $B$, we can prepare a variety of states on the boundary $\ptl B$. In this language, $|0\>$ is prepared by inserting the unit operator.

\begin{figure}
\begin{center}
\includegraphics[width=0.4\textwidth]{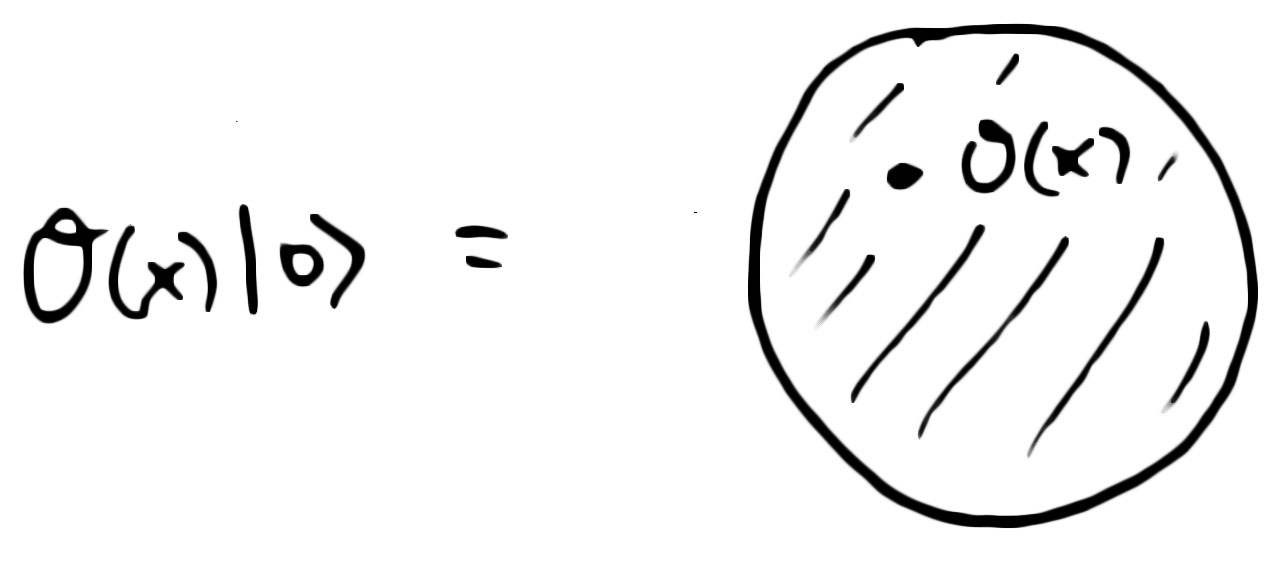}
\end{center}
\caption{The state $\cO(x)|0\>$ is given by inserting $\cO(x)$ inside the sphere and performing the path integral over the interior.  \label{fig:radialexcited}}
\end{figure}

\subsection{Operator $\Longleftarrow$ State}

This construction also works backwards. Let $|\cO_i\>$ be eigenstates of the dilatation operator
\be
D |\cO_i\> &=& \De_i |\cO_i\>.
\ee
The $|\cO_i\>$ can themselves be used as operators: we cut spherical holes $B_i$ out of the path integral centered around positions $x_i$ and glue in the states $|\cO_i\>$ at the boundary of the holes, as in figure~\ref{fig:correlatorofstates}.
\begin{figure}
\begin{center}
\includegraphics[width=0.7\textwidth]{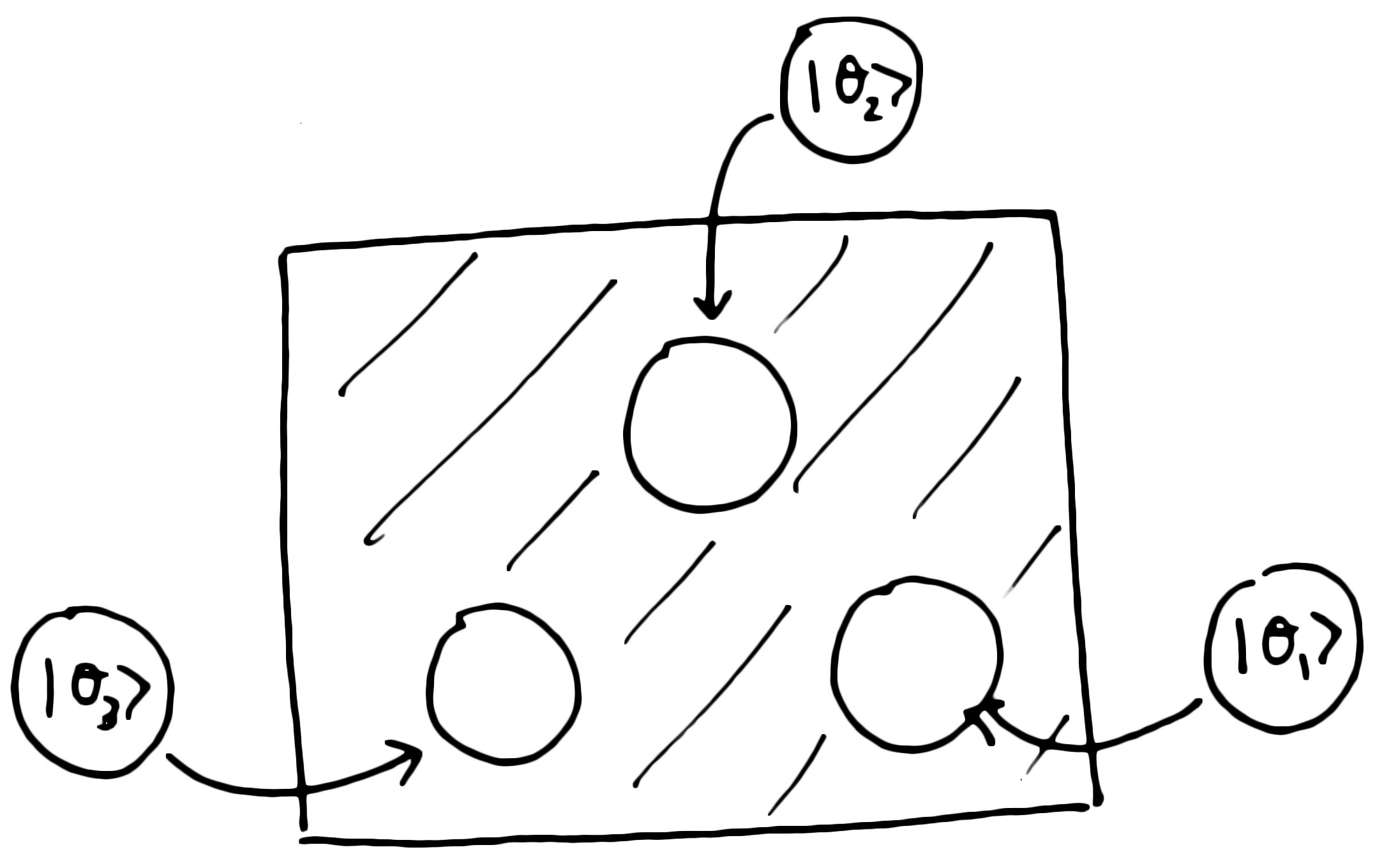}
\end{center}
\caption{A correlator of states is defined by cutting holes out of the path integral and gluing states into the holes.  \label{fig:correlatorofstates}}
\end{figure}
This gives a quantity that behaves exactly like a correlator of local operators.
In the scalar field example, the gluing procedure gives
\be
\<\cO_1(x_1)\cdots \cO_n(x_n)\> &=& \int \prod_i D\phi_{bi} \<\phi_{bi}|\cO_i\> \int_{\substack{\phi_{\ptl i}=\phi_{bi}\\ x \notin B_i }} D\phi(x)\, e^{-S},
\ee
where the path integral $D\phi(x)$ is performed over the region outside the balls $B_i$, and the integrals $D\phi_{bi}$ are over field configurations on the boundaries $\ptl B_i$. Here, $\phi_{\ptl i}$ denotes the restriction of the bulk field $\phi(x)$ to the $i$-th boundary $\ptl B_i$.

This construction only works when the $x_i$ are far enough apart that the balls $B_i$ don't overlap.  If they're too close together, we can use
\be
\<\cO(x_1)\cdots \cO(x_n)\> &=& \lambda^{\sum_i \Delta_i}\<\cO_1(\lambda x_1)\cdots \cO_n(\lambda x_n)\>,
\ee
with $\l$ sufficiently large to define the correlator.  Since the $x_i$ can now be arbitrarily close together, we have defined local operators.\footnote{A more careful construction of the state $\implies$ operator map that doesn't require this rescaling trick is given in Polchinski \cite{Polchinski:1998rq} volume 1, chapter 2.}

\subsection{Operator $\Longleftrightarrow$ State}

So far I've been vague about what I mean by a local operator.  But now, we can give a rigorous definition: we will simply {\it define\/} a local operator to be an eigenstate of $D$ in radial quantization.\footnote{The dilatation operator is diagonalizable in unitary (reflection positive) CFTs.  However, there exist interesting non-unitary theories where $D$ has a nontrivial Jordan block decomposition.  In these cases, we define a local operator as a finite-dimensional representation of $D$.} With this definition, the two constructions above are inverse to each other, with the identification
\be
\cO(0)\quad &\longleftrightarrow& \quad \cO(0)|0\>\equiv |\cO\>.
\ee
This is the ``state-operator correspondence."

It is straightforward to see how the conformal group acts on states in radial quantization.  A primary operator creates a state that is killed by $K_\mu$ and transforms in a finite-dimensional representation of $D$ and $M_{\mu\nu}$,
\begin{align}
\label{eq:operatortostateconditionK}
\,[K_\mu,\cO(0)]&=0 &\longleftrightarrow && K_\mu|\cO\>&=0,\\
\label{eq:operatortostateconditionD}
\,[D,\cO(0)] &= \De\cO(0) &\longleftrightarrow&& D|\cO\>&=\De|\cO\>,\\
\label{eq:operatortostateconditionM}
\,[M_{\mu\nu},\cO(0)]&=\cS_{\mu\nu}\cO(0) &\longleftrightarrow&& M_{\mu\nu}|\cO\> &= \cS_{\mu\nu}|\cO\>.
\end{align}
This follows by acting on $|0\>$ with the operator equations above and using the fact that $|0\>$ is killed by $K,D,$ and $M$.

A conformal multiplet in radial quantization is given by acting with momentum generators on a primary state
\be
|\cO\>, P_\mu|\cO\>, P_\mu P_\nu|\cO\>, \dots \qquad\textrm{(conformal multiplet)}.
\ee
This is equivalent to acting with derivatives of $\cO(x)$ at the origin, for example
\be
\ptl_\mu\cO(x)|_{x=0}|0\> &=& [P_\mu,\cO(0)]|0\>\ \ =\ \ P_\mu|\cO\>.
\ee
The operator $\cO(x)$ creates an infinite linear combination of descendants,
\be
\cO(x)|0\>\ =\  e^{x\.P}\cO(0)e^{-x\.P}|0\>\ =\ e^{x\.P}|\cO\>\ =\ \sum_{n=0}^\oo \frac{1}{n!}(x\.P)^n|\cO\>.\quad
\ee

As with the classification of operators, the action of the conformal algebra on a multiplet in radial quantization is determined by the commutation relations of the algebra. In fact the required computations look {\it exactly identical\/} to the computations we did to determine the action of conformal generators on operators (\ref{eq:actionbyrotation}, \ref{eq:dilatationaction}, \ref{eq:actionofK}).  This is because by surrounding operators with charges supported on spheres, we were secretly doing radial quantization all along!

\subsection{Another View of Radial Quantization}

To study a conformal Killing vector $\e$, it is often helpful to perform a Weyl rescaling of the metric $g\to \Omega(x)^2 g$ so that $\e$ becomes a regular Killing vector, i.e.\ an isometry.  We can turn a dilatation into an isometry by performing a Weyl rescaling from $\R^d$ to the cylinder $\R\x S^{d-1}$,
\be
ds_{\R^d}^2 &=& dr^2 + r^2 ds_{S^{d-1}}^2\nn\\
&=& r^2\p{\frac{dr^2}{r^2} + ds_{S^{d-1}}^2}\nn\\
&=& e^{2\tau}(d\tau^2 + ds_{S^{d-1}}^2) = e^{2\tau} ds_{\R\x S^{d-1}}^2,
\ee
where $r=e^\tau$.

Dilatations $r\to\l r$ become shifts of radial time $\tau\to\tau+\log \l$.  Radial quantization in flat space is equivalent to the usual quantization on the cylinder.  States live on spheres and time evolution is generated by acting with $e^{-D\tau}$.  While the development of radial quantization in the previous sections relied only on scale invariance, the cylinder picture relies on conformal invariance because we have performed a nontrivial Weyl rescaling.

Let us build a more detailed dictionary between the two pictures.  Under a Weyl rescaling, correlation functions of local operators transform as\footnote{In even dimensions, the partition function itself can transform with a Weyl anomaly $\<1\>_g=\<1\>_{\Omega^2 g}e^{S_\mathrm{Weyl}[g]}$.  This will not be important for our discussion, so we have divided through by the partition function.}
\be
\label{eq:weyltransformation}
\frac{\<\cO_1(x_1)\cdots\cO_n(x_n)\>_g}{\<1\>_g} &=& \p{\prod_i \Omega(x_i)^{\De_i}}\frac{\<\cO_1(x_1)\cdots\cO_n(x_n)\>_{\Omega^2g}}{\<1\>_{\Omega^2 g}}.
\ee
This is a nontrivial claim --- if we implement the Ising model in flat space, compute expectation values and take the continuum limit, it's not obvious that the answer should be simply related to the same lattice theory on the cylinder.\footnote{Comparing the flat and cylindrical Ising models is relatively easy in 2d, but harder in 3d since $S^2$ is curved. See \cite{Brower:2014gsa} for a recent attempt.}    In general it isn't, but at the critical value of the coupling when the theory becomes conformal, tracelessness of the stress tensor implies insensitivity to Weyl rescalings, and the answers become related.

\begin{exercise}
By integrating by parts in (\ref{eq:dilatationaction}), show that 
\be
\label{eq:tracetcontact}
T_\mu^\mu(x) \cO(y) &=& \De \de(x-y)\cO(y).
\ee
An insertion of $T_\mu^\mu$ is the response of the theory to an infinitesimal Weyl transformation $g\to e^{2\de\omega} g$. Derive (\ref{eq:weyltransformation}) by exponentiation.\footnote{We cheated here by only deriving (\ref{eq:tracetcontact}) in flat space.  In curved space there is an additional contribution to $T_\mu^\mu$ coming from the Weyl anomaly.  This factor cancels in (\ref{eq:weyltransformation}). There could also be modifications to the contact term (\ref{eq:tracetcontact}). However, in a conformally flat metric, we can simply define the curved space operator $\cO(x)$ so that it satisfies (\ref{eq:tracetcontact}). For instance, we may modify the Weyl factor so that it is constant in a tiny neighborhood of $\cO(x)$ and the flat-space calculation applies. This definition might not be consistent with other independent definitions. For instance, if $\cO(x)$ is the stress tensor, it gives a different answer from the canonical definition (\ref{eq:definitionofstresstensor}) because of the Weyl anomaly.}
\end{exercise}

Thus, given an operator $\cO(x)$ in $\R^d$, it is natural to define a cylinder operator
\be
\label{eq:definitionofcylinderop}
\cO_\mathrm{cyl.}(\tau,\bn) &\equiv& e^{\De \tau} \cO_\mathrm{flat}(x=e^\tau \bn).
\ee
We often omit the subscripts ``cyl." and ``flat," relying on the coordinates to indicate which type of operator we're discussing.
\begin{exercise}
Using (\ref{eq:weyltransformation}), compute a two-point function of cylinder operators
\be
\<\cO(\tau_1,\bn_1)\cO(\tau_2,\bn_2)\>.
\ee
Verify that it is time-translationally invariant on the cylinder. Show that in the limit of large time separation $\tau=\tau_2-\tau_1 \gg 1$, the two-point function has an expansion in terms of the form $e^{-(\De+n)\tau}$ with integer $n\geq 0$.  Interpret these as coming from the exchange of states in the conformal multiplet of $\cO$.
\end{exercise}

\section{Reflection Positivity and Unitarity Bounds}

\subsection{Reflection Positivity}
\label{sec:reflectionpositivity}

In Lorentzian signature, we are interested in unitary theories: theories where the conserved charges (including the Hamiltonian) are Hermitian operators so that they generate unitary transformations.  Unitarity in Lorentzian signature is equivalent to a property called ``reflection positivity" in Euclidean signature.\footnote{We make some brief comments about Euclidean vs. Lorentzian field theory and analytic continuation in appendix~\ref{app:analyticcontinuation}.}

Consider a Lorentzian theory with a local operator $\cO_L$ and Hermitian energy-momentum generators $(H,\bP_L)$ ($L$ is for ``Lorentzian").  We have the textbook formula
\be
\label{eq:textbookinlorentz}
\cO_L(t,\bx) &=& e^{iHt-i\bx\.\bP_L}\cO_L(0,0)e^{-iHt+i\bx\.\bP_L}.
\ee
Let $\cO_L(0,0)$ be Hermitian.  It follows from (\ref{eq:textbookinlorentz}) that $\cO_L(t,\bx)$ is Hermitian too.

Now, let us Wick-rotate to Euclidean signature,
\be
\label{eq:wickrotatedoperator}
\cO_E(t_E,\bx) \equiv \cO_L(-it_E,\bx)
= e^{Ht_E-i\bx\.\bP_L}\cO_L(0,0)e^{-Ht_E+i\bx\.\bP_L}.
\ee
The Euclidean operator satisfies
\be
\cO_E(t_E,\bx)^\dag &=& \cO_E(-t_E,\bx).
\ee
To Wick-rotate an operator with spin, we conventionally add factors of $-i$ to the time components,\footnote{These factors are needed to make Euclidean correlation functions manifestly covariant under $\SO(d)$ rotations.} e.g.\ for a vector operator $\cO_L^\mu$,
\be
\cO_E^0(t_E,\bx) &=& -i \cO_L^0(-it_E,\bx),\nn\\
\cO_E^i(t_E,\bx) &=& \cO_L^i(-it_E,\bx).
\ee
This leads to
\be
\label{eq:reflectionforhermitianconjugation}
\cO_E^{\mu_1\dots\mu_\ell}(t_E,\bx)^\dag &=& \Theta^{\mu_1}{}_{\nu_1}\cdots \Theta^{\mu_\ell}{}_{\nu_\ell} \cO_E^{\nu_1\cdots\nu_\ell}(-t_E,\bx),
\ee
where $\Theta^\mu{}_\nu = \de^\mu_\nu-2\de^\mu_0\de_\nu^0$ is a reflection in the time-direction.

Thus, the way Hermitian conjugation acts on a Euclidean operator depends on which direction we call time.  Whether an operator is Hermitian or not depends on how we quantize the theory! This is very different from Lorentzian signature, where the conjugation properties of operators don't depend on a choice of reference frame.

As an example, consider the momentum generators
\be
P^\mu &=& -\int d^{d-1}\bx\, T^{\mu 0}(0,\bx).
\ee
(From now on, we work in the Euclidean theory and omit the $E$ subscripts.)
Using (\ref{eq:reflectionforhermitianconjugation}), we have
\be
T^{i 0}(0,\bx)^\dag &=& -T^{i0}(0,\bx),\nn\\
T^{00}(0,\bx)^\dag &=& T^{00}(0,\bx).
\ee
It follows that $P^0$ is Hermitian, and the $P^i$ are {\it antihermitian}.  We may write
\be
P^0 = H,\qquad
P^j = -iP^j_L,
\ee
with $H,P_L$ Hermitian, and then (\ref{eq:integratedtranslations}) agrees with the formula we got from Wick rotation (\ref{eq:wickrotatedoperator}).  If we had quantized with a different time direction, say the $x_1$-direction, then we would conclude that $P^1$ is Hermitian, while $P^0,P^2,\dots,P^{d-1}$ are antihermitian.

To reiterate, {\it the way conjugation acts on operators depends on how we quantize our theory}.  This makes sense, because Hermitian conjugation is something you do to operators on Hilbert spaces, and different quantizations have different Hilbert spaces.

This raises the question: given a Euclidean path integral, how do we know if it computes the Wick-rotation of a unitary Lorentzian theory?  One important condition is that norms of states should be positive.  Consider some in-state $|\psi\>$ given by acting on the vacuum with a bunch of operators at negative Euclidean time
\be
|\psi\> &=& \cO(-t_{E1})\cdots\cO(-t_{En})|0\>.
\ee
For brevity, we suppress the spatial positions of the operators.
The conjugate state is given by
\be
\<\psi| &=& (\cO(-t_{E1})\cdots\cO(-t_{En})|0\>)^\dag \nn \\
&=& \<0|\cO(t_{En})\cdots\cO(t_{E1}).
\ee
That is, $\<\psi|$ is given by taking the vacuum in the future and positioning operators in a time-reflected way.  Thus, the condition
\be
\<\psi|\psi\> &\geq& 0
\ee
says that a time-reflection symmetric configuration should have a positive path integral, see figure~\ref{fig:reflectionpositivity}.  This is called ``reflection positivity."

\begin{figure}
\begin{center}
\includegraphics[width=0.5\textwidth]{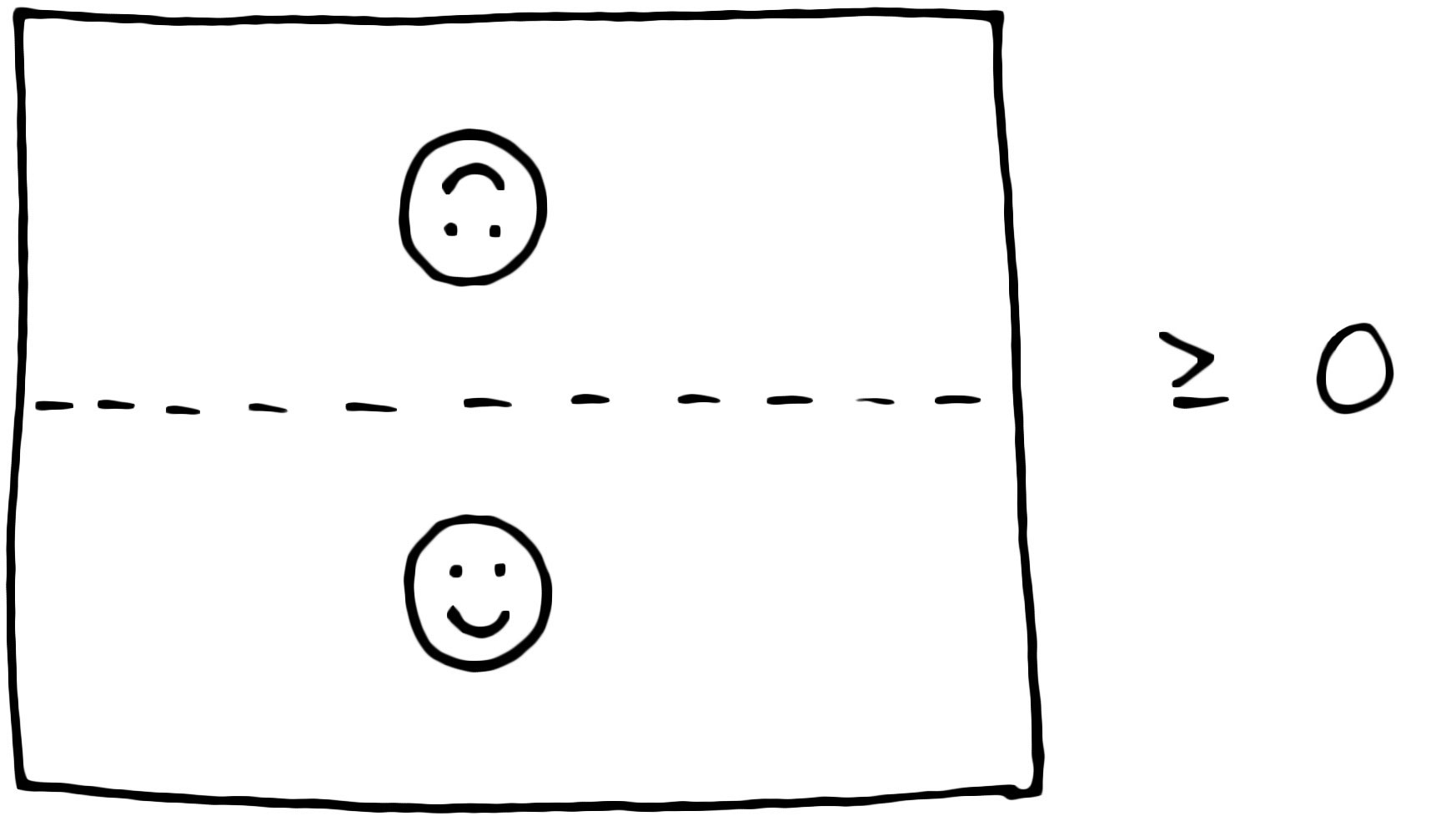}
\end{center}
\caption{Reflection positivity.  \label{fig:reflectionpositivity}}
\end{figure}

If a Euclidean theory is the Wick-rotation of a unitary Lorentzian theory, then it will be reflection positive.  However, some theories are more naturally defined in Euclidean signature.  In this case, reflection positivity must be checked.  It often suffices to check it in any microscopic theory in the same universality class as the CFT we're interested in.
\begin{exercise}
Consider the 2d Ising lattice correlator shown in figure~\ref{fig:isingreflectioncorrelator}.  Show that it can be written as a sum of squares, and is hence positive. (Hint: first sum over spins off the line $L$, and then sum over spins on $L$.) Generalize your proof to argue that the 2d Ising model is reflection-positive.
\end{exercise}

\begin{figure}
\begin{center}
\includegraphics[width=0.3\textwidth]{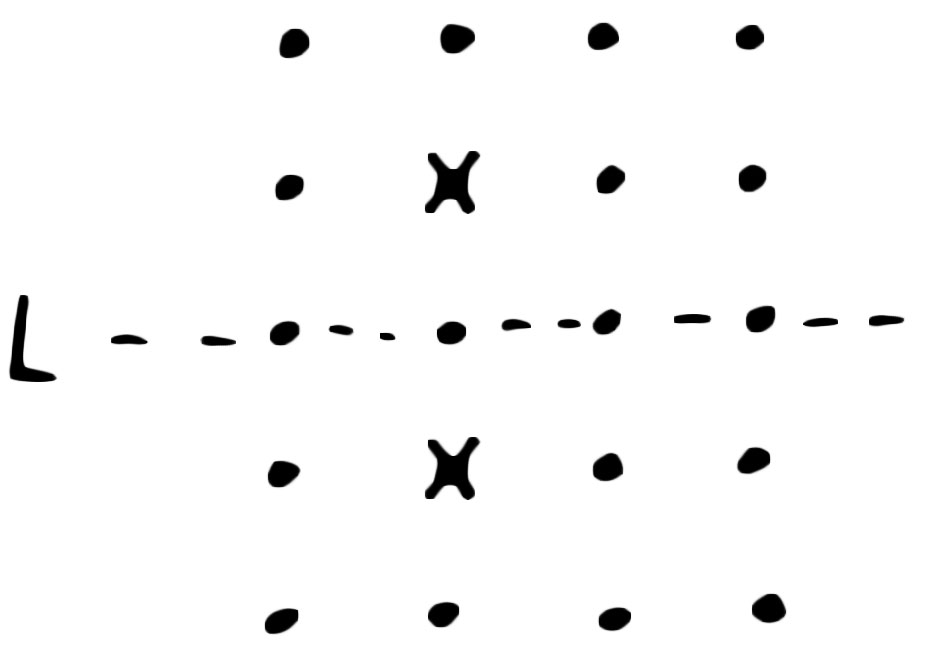}
\end{center}
\caption{A two point function on a $4\x 5$ Ising lattice with free boundary conditions, with spin operators  inserted at the sites marked with an X.  \label{fig:isingreflectioncorrelator}}
\end{figure}

The Osterwalder-Schrader reconstruction theorem says that, given a collection of Euclidean correlators satisfying reflection positivity (and some additional technical assumptions), we can reconstruct a unitary Lorentzian quantum field theory by analytic continuation \cite{GlimmJaffe}.  So reflection positivity in Euclidean signature and unitarity in Lorentzian signature are essentially equivalent, and we will use the terms interchangeably. 

\subsubsection{Real vs.\ Complex Operators}
\label{sec:realvscomplex}

Because Hermitian conjugation is tricky in Euclidean signature, it is helpful introduce some extra terminology.  We call a local operator ``real" if it is Hermitian in Lorentzian signature.  In Euclidean signature, real operators satisfy (\ref{eq:reflectionforhermitianconjugation}).  By contrast, for a complex operator $\cO_L^\dag = \cO_L^*$, we have
\be
\cO_E^{\mu_1\dots\mu_\ell}(t_E,\bx)^\dag &=& \Theta^{\mu_1}{}_{\nu_1}\cdots \Theta^{\mu_\ell}{}_{\nu_\ell} \cO_E^{\nu_1\cdots\nu_\ell*}(-t_E,\bx).
\ee

Later we will need the following result.
If $\f_1,\f_2$ are real scalars and $\cO$ is a real operator with spin $\ell$ in a unitary theory, then the three-point coefficient $f_{\f_1\f_2\cO}$ is real. This is easiest to see in Lorentzian signature when the operators are spacelike separated $x_{ij}^2>0$. Because local operators commute at spacelike separation, we have
\be
\<0|\f_1(x_1)\f_2(x_2) \cO^{\mu_1\cdots\mu_\ell}(x_3)|0\>^* &=& \<0|\f_1(x_1)\f_2(x_2) \cO^{\mu_1\cdots\mu_\ell}(x_3)|0\>.\quad
\ee
Substituting (\ref{eq:scalarscalarspinL}) gives $f_{\f_1\f_2\cO}^*=f_{\f_1\f_2\cO}$.

\subsection{Reflection Positivity on the Cylinder}

Reflection-positivity (or unitarity) has interesting consequences for CFTs on the cylinder.  The Hermitian conjugate of a real cylinder operator is
\be
\cO_\mathrm{cyl.}(\tau,\bn)^{\dag_\mathrm{rad}} &=& \cO_\mathrm{cyl.}(-\tau,\bn).
\ee
Using (\ref{eq:definitionofcylinderop}), this becomes
\be
\cO_\mathrm{flat}(x)^{\dag_\mathrm{rad}} &=& x^{-2\De}\cO_\mathrm{flat}\p{\frac{x^\mu}{x^2}}.\label{eq:radialconjugation}
\ee
Above, we have written $\dag_\mathrm{rad}$ to emphasize that Hermitian conjugation in radial quantization is different from Hermitian conjugation in the usual $P^0$ quantization.  From now on we write simply $\dag$, and hope that the meaning will be clear from context.

The right-hand side of (\ref{eq:radialconjugation}) is just the image of $\cO(x)$ under an inversion $I:x^\mu\to \frac{x^\mu}{x^2}$.  The same is true for operators with spin, where the full formula (\ref{eq:finiteprimarytransformation}) gives
\be
\label{eq:conjugateforspin}
\cO^{\mu_1\cdots\mu_\ell}(x)^\dag &=& I^{\mu_1}{}_{\nu_1}(x)\cdots I^{\mu_\ell}{}_{\nu_\ell}(x) x^{-2\De} \cO^{\nu_1\dots\nu_\ell}\p{\frac{x}{x^2}},\nn\\
I^\mu{}_\nu(x) &=& \de^\mu_\nu - \frac{2 x^\mu x_\nu}{x^2}.
\ee
\begin{exercise}
Check that the two-point function of spin-1 operators (\ref{eq:twoptfunctionofspin1}) satisfies reflection-positivity on the cylinder if $C_J>0$.
\end{exercise}
Applying (\ref{eq:conjugateforspin}) to the stress tensor, we find that 
 the action of conjugation on the conformal charges in radial quantization is
\be
Q_\e^\dag &=& -Q_{I\e I}.
\ee
In particular, we have
\be
\label{eq:mantihermitian}
M_{\mu\nu}^\dag &=& -M_{\mu\nu},\nn\\
D^\dag &=& D,\nn\\
P_\mu^\dag &=& K_\mu.
\ee

These facts let us calculate properties of correlation functions purely algebraically.  As an example, consider a two-point function.  Letting $\tl y = y/y^2$, we have
\be
\<\cO(y)\cO(x)\> &=& \<0|(y^{-2\De}\cO(\tl y))^\dag \cO(x)|0\>\nn\\
&=& y^{-2\De}\<0|(e^{\tl y\.P}\cO(0)e^{-\tl y\.P})^\dag e^{x\.P}\cO(0)e^{-x\.P}|0\>\nn\\
&=& y^{-2\De}\<0|e^{-\tl y\.K}\cO(0)^\dag e^{\tl y\.K} e^{x\.P}\cO(0)e^{-x\.P}|0\>\nn\\
&=& y^{-2\De}\<0|\cO(0)^\dag e^{\tl y\.K} e^{x\.P}\cO(0) |0\>\nn\\
&=& y^{-2\De}\<\cO|e^{\tl y\.K} e^{x\.P}|\cO\>,
\label{eq:twopointfromalgebra}
\ee
where we've defined
\be
\<\cO| \equiv \<0|\cO(0)^\dag = \lim_{y\to \oo} y^{2\De} \<0|\cO(y).
\ee
By expanding the exponentials, we can evaluate (\ref{eq:twopointfromalgebra}) using the conformal algebra.  For example, the first couple terms are
\be
\<\cO(y)\cO(x)\> &=& y^{-2\De}\p{\<\cO|\cO\> + \frac{y^\mu}{y^2}x^\nu\<\cO|K_\mu P_\nu|\cO\>+\dots},
\ee
where we've used that $K|\cO\>=\<\cO|P=0$ because $\cO$ is primary.  Using the conformal commutation relations,
\be
\<\cO|K_\mu P_\nu|\cO\> &=& \<\cO|[K_\mu,P_\nu]|\cO\>\nn\\
&=& \<\cO|(2D\de_{\mu\nu}-2M_{\mu\nu})|\cO\>\nn\\
&=& 2\De\de_{\mu\nu}\<\cO|\cO\>.
\label{eq:normoffirstdescendant}
\ee
Thus,
\be
\<\cO(y)\cO(x)\> &=& y^{-2\De} \<\cO|\cO\>\p{1 + 2\De\frac{y\.x}{y^2}+\dots}.
\ee
This exactly matches the expansion of $\<\cO|\cO\>/(x-y)^{2\De}$ in small $|x|/|y|$!  You can imagine computing all the higher terms and matching the whole series expansion.

Let us also prove our earlier claim that a two-point function of operators in different irreducible spin representations must vanish.  Consider a primary operator $\cO^a$ transforming in a nontrivial unitary representation of $\SO(d)$. The dual state transforms in the dual representation, so we will write it with a lowered index $(|\cO^a\>)^\dag=\<\cO_a|$.  Consider the matrix element $\<\cO_a|M_{\mu\nu}|\cO^b\>$.  Using that $M_{\mu\nu}$ is antihermitian (\ref{eq:mantihermitian}), we can act with it on both the bra and the ket:
\begin{align}
-((\cS_{\mu\nu})_c{}^a)^*\<\cO_c|\cO^b\> = \<\cO_a|M_{\mu\nu}|\cO^b\> = \<\cO_a|\cO^c\>(\cS_{\mu\nu})_c{}^b.
\end{align}
But $\cS_{\mu\nu}$ is antihermitian as well, so as a matrix equation this says
\be
\cS_{\mu\nu}N=N \cS_{\mu\nu},
\ee
where $N_a{}^b\equiv\<\cO_a|\cO^b\>$.  By Schur's lemma, $N_a{}^b$ must vanish if $a$ and $b$ are indices of different irreducible representations.  If $a,b$ are indices for a single irreducible representation, then $N$ is proportional to the identity.

\begin{exercise}
This computation is not directly relevant to the course, but it is instructive for getting used to radial ordering.  Consider a three-point function of scalars
\be
\<\cO_i(x_1)\cO_j(x_2)\cO_k(x_3)\> &=& \<0|\mathcal{R}\{\cO_i(x_1)\cO_j(x_2)\cO_k(x_3)\}|0\>\nn\\
&=&\theta(|x_3|\geq |x_2| \geq |x_1|)\<0|\cO_k(x_3)\cO_j(x_2)\cO_i(x_1)|0\>\nn\\
&&+\mathrm{permutations}.
\ee
Consider the operator $e^{2\pi i\mathcal{D}_1}$ where
\be
\mathcal{D}_1 &=& x_1\.\ptl_1+\De_1.
\ee
Using the fact that $e^{2\pi i\mathcal{D}_1}\cO_1(x_1)= e^{2\pi i D}\cO_1(x_1)e^{-2\pi i D}$, compute the action of $e^{2\pi i \mathcal{D}_1}$ on each of the terms above.  You will get different answers for each of the different operator orderings.

Now determine the action of $e^{2\pi i \mathcal{D}_1}$ on the known answer for the scalar three-point function (\ref{eq:conformalthreeptfunction}).  Check that the two answers agree.
\end{exercise}

\subsection{Unitarity Bounds}

Thinking about the theory on the cylinder gives a natural inner product on states in radial quantization.  Unitarity (or reflection positivity) implies that the norms of states must be nonnegative.  By demanding positivity for every state in a conformal multiplet, we obtain bounds on dimensions of primary operators \cite{Mack:1975je,Jantzen,Minwalla:1997ka}.  We have already seen an example in (\ref{eq:normoffirstdescendant}). We found
\be
|P_0|\cO\>|^2 = \<\cO|K_0 P_0|\cO\> = 2\De\<\cO|\cO\>.
\ee
Unitarity implies $\De\geq 0$.

Let us do the same exercise, this time for an operator $\cO^a$ in a nontrivial irreducible representation $R_\cO$ of $\SO(d)$. We normalize $\cO$ so that
\be
\<\cO_b|\cO^a\> &= \de^a_b.
\ee
Taking inner products between first-level descendants and using the conformal algebra, we find
\be
(P_\mu|\cO^a\>)^\dag P_\nu|\cO^b\> = \<\cO_a|K_\mu P_\nu|\cO^b\>
= 2\De\de_{\mu\nu}\de_a^b-2(\cS_{\mu\nu})_a{}^b.\label{innerproduct}
\ee
The state $P_\nu|\cO^b\>$ lives in the representation $V\otimes R_\cO$ of $\SO(d)$, where $V$ is the vector representation.
Unitarity implies that (\ref{innerproduct}) must be positive-definite as a matrix acting on this space.  This implies
\be
\De\geq \textrm{max-eigenvalue}((\cS_{\mu\nu})_a{}^b).
\ee
Let us write
\be
(\cS_{\mu\nu})_a{}^b &=& \frac 1 2 (L^{\a\b})_{\mu\nu}(\cS_{\a\b})_a{}^b\nn\\
(L^{\a\b})_{\mu\nu} &\equiv& \de^\a_\mu\de^\b_\nu - \de^\a_\nu\de^\b_\mu,
\ee
where $(L^{\a\b})_{\mu\nu}$ is the generator of rotations in the vector representation $V$.  Writing $A=\a\b$ for an adjoint index of $\SO(d)$, and thinking of $L^A,\cS_A$ as operators on $V\otimes R_\cO$, this becomes
\be
L^A \cS_A &=& \frac 1 2 \p{(L+\cS)^2-L^2 -\cS^2}\nn\\
&=& \frac 1 2\p{-\mathrm{Cas}(V\otimes R_\cO)+\mathrm{Cas}(V)+\mathrm{Cas}(R_\cO)},
\ee
where we've used the familiar trick from basic quantum mechanics to get a linear combination of Casimir operators.\footnote{The quadratic Casimir is $-L^2$ because our generators are antihermitian and differ from the conventional ones by a factor of $i$.}

Let's specialize to the case where $R_\cO=V_\ell$ is the spin-$\ell$ traceless symmetric tensor representation.  $V_\ell$ has Casimir $\ell(\ell+d-2)$.  To get the maximal eigenvalue of $L\.\cS$, we need the minimal Casimir of
\be
V\otimes V_\ell=V_{\ell-1}\oplus \dots\qquad(\ell>0).
\ee
Here the ``$\dots$" are irreducible representations with larger Casimirs.  Thus,
\be
\De &\geq& \frac 1 2\p{-\mathrm{Cas}(V_{\ell-1})+\mathrm{Cas}(V)+\mathrm{Cas}(V_\ell)}\nn\\
&=& \ell+d-2.
\ee
This computation was valid only for $\ell> 0$, since for scalars $V\otimes V_{\ell=0}=V$.

One can also consider more complicated descendants.
\begin{exercise}
Compute the norm of $P_\mu P^\mu |\cO\>$, where $\cO$ is a scalar.  Show that unitarity implies either $\De=0$ or $\De\geq \frac{d-2}{2}$.  This gives a stronger condition than what we derived above ($\De\geq 0$) for scalars.
\end{exercise}

For traceless symmetric tensors in general conformal field theories, these inequalities are the best you can do (other descendants give no new information).  In theories with more symmetry, like supersymmetric theories or 2d CFTs, unitarity bounds can be more interesting.  A classic reference for unitarity bound computations is \cite{Minwalla:1997ka}.  In the math literature, unitarity bounds for higher-dimensional CFTs were essentially computed long ago by Jantzen \cite{Jantzen}, though the relevance of that work for physics has only been emphasized recently \cite{Yamazaki:2016vqi,Penedones:2015aga}.

In summary, we have the unitarity bounds
\be
\De &=& 0 \ \textrm{(unit operator), or}\nn\\
\De &\geq& \begin{cases}
\frac{d-2}{2} & \ell = 0,\\
\ell+d-2 & \ell > 0.
\end{cases}
\label{eq:unitaritysummary}
\ee

\subsubsection{Null States and Conserved Currents}

If $\De$ saturates the bounds (\ref{eq:unitaritysummary}), the conformal multiplet will have a null state.  For the unit operator, all descendants are null.  For a scalar with dimension $\frac{d-2}{2}$, the null state is
\be
P^2|\cO\> &=& 0.
\ee
In operator language, this says $\ptl^2 \cO(x)=0$, which means $\cO$ satisfies the Klein-Gordon equation, and is thus a free scalar that decouples from the rest of the CFT.

For a spin-$\ell$ operator, the null state is\footnote{The null state has spin $\ell-1$ because the unitarity bound came from $V_{\ell-1}\subset V\otimes V_\ell$.  Something special happens for vectors in 2d, where $V\otimes V = \mathbf{1}\oplus\mathbf{1}\oplus V_{2}$, with the extra $\mathbf{1}$ appearing because of the antisymmetric $\e_{\mu\nu}$ symbol. Unitarity then implies that $J^{\mu}$ and $\e^{\mu\nu}J_\nu$ are each separately conserved.}
\be
P_\mu | \cO^{\mu\mu_2\cdots\mu_\ell}\> &=& 0.
\ee
In operator language, this becomes the equation for a conserved current
\be
\ptl_\mu \cO^{\mu\mu_2\cdots\mu_\ell}(x) &=& 0.
\ee
The implication also works the other way, so
\be
\De = \ell + d-2 \qquad \textrm{if and only if}\qquad \textrm{$\cO$ is a conserved current}.
\ee
Some important examples are global symmetry currents ($\ell=1$, $\Delta=d-1$) and the stress tensor ($\ell=2$, $\Delta=d$).  For CFTs in $d\geq 3$, the presence of currents with spin $\ell\geq 3$ implies that the theory is free \cite{Maldacena:2011jn,Alba:2015upa}.\footnote{One must also assume the existence of exactly one stress tensor, since otherwise the theory could contain a free subsector, decoupled from the rest.}

\subsection{Only Primaries and Descendants}
\label{sec:onlyprimariesanddescendants}

With a positive-definite inner product, we can now prove that all operators in unitary CFTs are linear combinations of primaries and descendants. We will use one additional physical assumption: that the partition function of the theory on $S^{d-1}\x S^1_\beta$ is finite,
\be
\mathcal{Z}_{S^{d-1}\x S^1_\beta} = \Tr(e^{-\beta D}) < \oo.
\ee
This means that $e^{-\beta D}$ is trace-class, and hence diagonalizable with a discrete spectrum (by the spectral theorem).\footnote{Assuming $e^{-\b D}$ is trace-class may be too strong for some applications. Boundedness of $e^{-\b D}$ suffices for $D$ to be diagonalizable (with a possibly continuous spectrum). An interesting example is Liouville theory, which has a divergent partition function and continuous spectrum, but still has many properties of a sensible CFT, like an OPE.} It follows that $D$ is also diagonalizable, with real eigenvalues because $D$ is Hermitian.

Now consider a local operator $\cO$, and assume for simplicity it is an eigenvector of dilatation with dimension $\De$.  By finiteness of the partition function, there are a finite number of primary operators $\cO_p$ with dimension less than or equal to $\De$.  Using the inner product, we may subtract off the projections of $\cO$ onto the conformal multiplets of $\cO_p$ to get $\cO'$.  Now suppose (for a contradiction) that $\cO'\neq 0$.  Acting on it with $K_\mu$'s, we must eventually get zero (again by finiteness of the partition function), which means there is a new primary with dimension below $\De$, a contradiction.  Thus $\cO'=0$, and $\cO$ is a linear combination of states in the multiplets $\cO_p$.

\section{The Operator Product Expansion}

If we insert two operators $\cO_i(x)\cO_j(0)$ inside a ball and perform the path integral over the interior, we get some state on the boundary.  Because every state is a linear combination of primaries and descendants, we can decompose this state as
\be
\label{eq:opeinitial}
\cO_i(x)\cO_j(0)|0\> &=& \sum_{k}C_{ijk}(x,P)\cO_k(0) |0\>,
\ee
where $k$ runs over primary operators and $C_{ijk}(x,P)$ is an operator that packages together primaries and descendants in the $k$-th conformal multiplet (figure~\ref{fig:ope}).

\begin{figure}
\begin{center}
\includegraphics[width=0.6\textwidth]{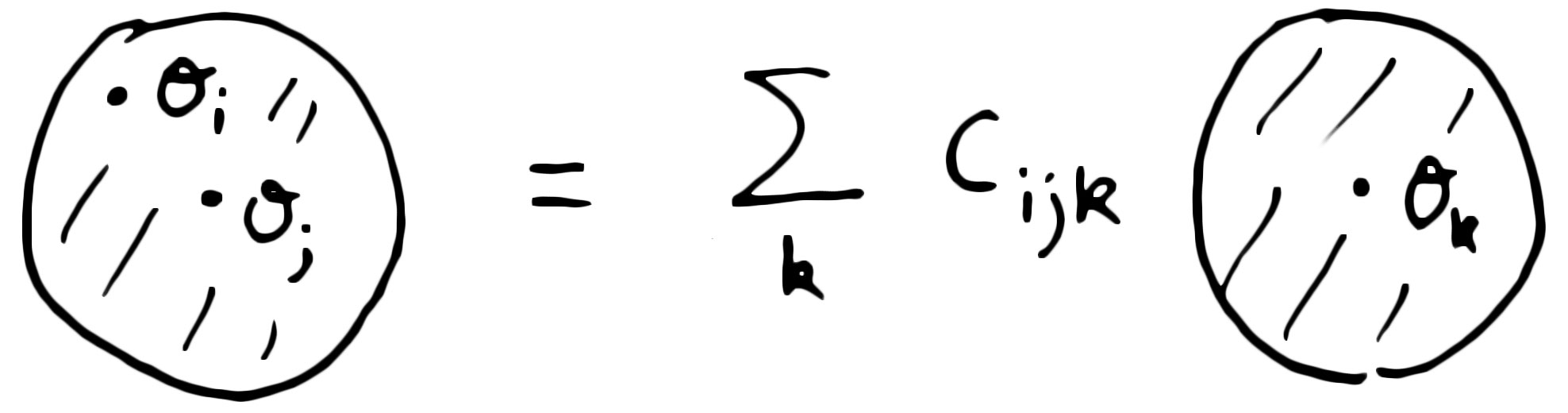}
\end{center}
\caption{A state created by two operator insertions can be expanded as a sum of primary and descendant states.  \label{fig:ope}}
\end{figure}

Eq.~(\ref{eq:opeinitial}) is an exact equation that can be used in the path integral, as long as all other operators are outside the sphere with radius $|x|$.  Using the state-operator correspondence, we can write
\be
\label{eq:opeinitial2}
\cO_i(x_1)\cO_j(x_2) &=& \sum_{k}C_{ijk}(x_{12},\ptl_2)\cO_k(x_2),\qquad\textrm{(OPE)}
\ee
where it is understood that (\ref{eq:opeinitial2}) is valid inside any correlation function where the other operators $\cO_n(x_n)$ have $|x_{2n}|\geq |x_{12}|$.  Eq.~(\ref{eq:opeinitial2}) is called the Operator Product Expansion (OPE).

We could alternatively perform radial quantization around a different point $x_3$, giving
\be
\label{eq:opealternative}
\cO_i(x_1)\cO_j(x_2) &=& \sum_k C'_{ijk}(x_{13},x_{23},\ptl_3)\cO_k(x_3),
\ee
where $C'_{ijk}(x_{13},x_{23},\ptl_3)$ is some other differential operator (figure~\ref{fig:radialquantotherpoint}).  The form  (\ref{eq:opeinitial2}) is usually more convenient for computations, but the existence of (\ref{eq:opealternative}) is important. It shows that we can do the OPE between two operators whenever it's possible to draw any sphere that separates the two operators from all the others.

\begin{figure}
\begin{center}
\includegraphics[width=0.6\textwidth]{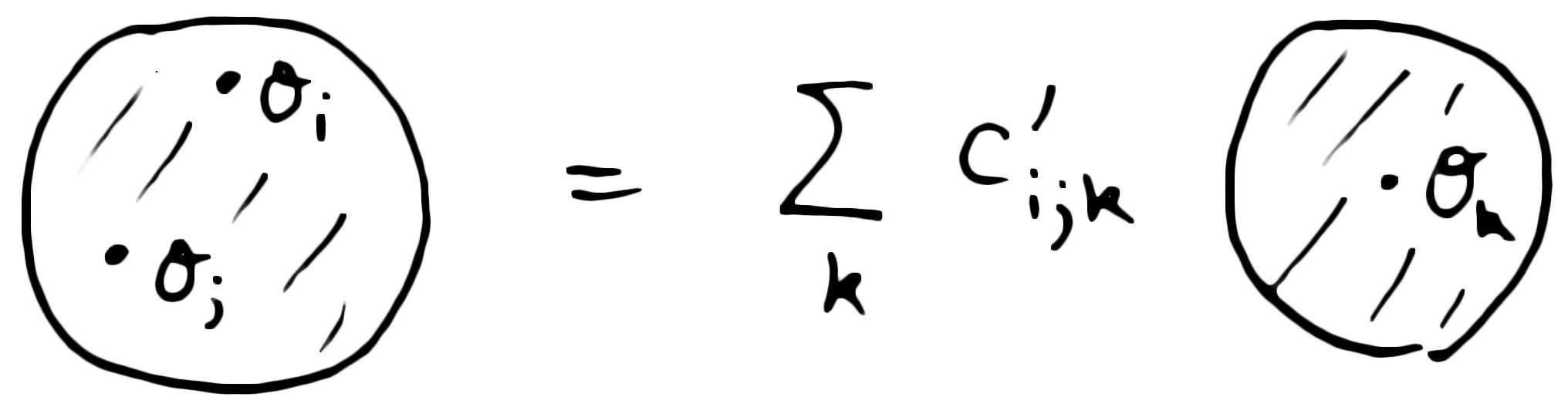}
\end{center}
\caption{It isn't necessary for one of the operators to be at the origin.  \label{fig:radialquantotherpoint}}
\end{figure}

We are being a bit schematic in writing the above equations.  It's possible for all the operators to have spin.  In this case, the OPE looks like
\be
\cO_i^a(x_1)\cO_j^b(x_2) &=& \sum_k C_{ijk}^{ab}{}_c(x_{12},\ptl_2)\cO_k^c(x_2),
\ee
where $a,b,c$ are indices for (possibly different) representations of $\SO(d)$.

\subsection{Consistency with Conformal Invariance}

Conformal invariance strongly restricts the form of the OPE\@.  For simplicity, suppose $\cO_i$, $\cO_j$, and $\cO_k$ are scalars.  
\begin{exercise}
By acting on both sides of (\ref{eq:opeinitial}) with $D$, prove that $C_{ijk}(x,\ptl)$ has an expansion of the form
\be
\label{eq:opeexpansionexample}
C_{ijk}(x,\ptl) &\propto& |x|^{\De_k-\De_i-\De_j}\p{1 +\# x^\mu\ptl_\mu + \# x^\mu x^\nu\ptl_\mu\ptl_\nu+\# x^2 \ptl^2 + \dots}.\nn\\
\ee
\end{exercise}
This is just a fancy way of saying we can do dimensional analysis and that $\cO_i$ has length-dimension $-\De_i$. We're also implicitly using rotational invariance by contracting all the indices appropriately. We could have proved this too by acting with $M_{\mu\nu}$.

We get a more interesting constraint by acting with $K_\mu$. In fact, consistency with $K_\mu$ completely fixes $C_{ijk}$ up to an overall coefficient. In this way, we can determine the coefficients in (\ref{eq:opeexpansionexample}).

This computation is a little annoying (exercise!), so here's a simpler way to see why the form of the OPE is fixed, and to get the coefficients in (\ref{eq:opeexpansionexample}).  Take the correlation function of both sides of (\ref{eq:opeinitial2}) with a third operator $\cO_k(x_3)$ (we will assume $|x_{23}|\geq |x_{12}|$, so that the OPE is valid),
\be
\label{eq:threetotwo}
\<\cO_i(x_1)\cO_j(x_2)\cO_k(x_3)\> &=& \sum_{k'} C_{ijk'}(x_{12},\ptl_2)\<\cO_{k'}(x_2)\cO_k(x_3)\>.
\ee
The three-point function on the left-hand side is fixed by conformal invariance, and is given in  (\ref{eq:conformalthreeptfunction}). We can choose an orthonormal basis of primary operators, so that $\<\cO_k(x_2)\cO_{k'}(x_3)\>= \de_{kk'}x_{23}^{-2\De_k}$.  The sum then collapses to a single term, giving
\be
\label{eq:matchingthreept}
\frac{f_{ijk}}{x_{12}^{\De_i+\De_j-\De_k}x_{23}^{\De_j+\De_k-\De_i}x_{31}^{\De_k+\De_i-\De_j}} &=& C_{ijk}(x_{12},\ptl_2)x_{23}^{-2\De_k}.
\ee
This determines $C_{ijk}$ to be proportional to $f_{ijk}$, times a differential operator that depends only on the $\De_i$'s. The operator can be obtained by matching the small $|x_{12}|/|x_{23}|$ expansion of both sides of (\ref{eq:matchingthreept}).
\begin{exercise}
\label{exercise:seriesfordiffops}
Consider the special case $\De_i=\De_j=\De_\f$, and $\De_k=\De$.  Show
\be
\label{eq:identicalscalaropeoperator}
C_{ijk}(x,\ptl) &=& f_{ijk} x^{\De-2\De_\f}\p{1+\frac 1 2 x\.\ptl + \a x^\mu x^\nu\ptl_\mu\ptl_\nu + \b x^2 \ptl^2+\dots},\nn\\
\ee
where
\be
\label{eq:opedescendantcoefficients}
\a &=& \frac{\De+2}{8(\De+1)},\quad\textrm{and}\quad \b=-\frac{\De}{16(\De-\frac{d-2}{2})(\De+1)}.
\ee
\end{exercise}

\subsection{Computing Correlators with the OPE}

Equation (\ref{eq:threetotwo}) gives an example of using the OPE to reduce a three-point function to a sum of two-point functions.  In general, we can use the OPE to reduce any $n$-point function to a sum of $n-1$-point functions,
\be
\<\cO_1(x_1)\cO_2(x_2)\cdots\cO_n(x_n)\> &=& \sum_k C_{12k}(x_{12},\ptl_2)\<\cO_k(x_2)\cdots\cO_n(x_n)\>.\nn\\
\ee
Recursing, we reduce everything to a sum of one-point functions, which are fixed by dimensional analysis,
\be
\<\cO(x)\> &=& \begin{cases}
1 & \textrm{if $\cO$ is the unit operator,}\\
0 & \textrm{otherwise.}
\end{cases}
\ee
This gives an algorithm for computing any flat-space correlation function using the OPE\@.  It shows that all these correlators are determined by dimensions $\De_i$, spins, and OPE coefficients $f_{ijk}$.\footnote{The OPE is also valid on any conformally flat manifold.  The difference is that on nontrivial manifolds, non-unit operators can have nonzero one-point functions.  An example is $\R^{d-1}\x S^1_\b$, which has the interpretation as a CFT at finite temperature.  By dimensional analysis, we have $\<\cO\>_{\R^{d-1}\x S^1_\b}\propto \b^{-\De_\cO}\propto T^{\De_\cO}$.\label{foot:finitetemperature}}

\section{Conformal Blocks}

\subsection{Using the OPE}

Let us use the OPE to compute a four-point function of identical scalars.  Recall that Ward identities imply
\be
\<\f(x_1)\f(x_2)\f(x_3)\f(x_4)\> &=& \frac{g(u,v)}{x_{12}^{\De_\f}x_{34}^{\De_\f}},
\ee
where the cross-ratios $u$, $v$ are given by~(\ref{eq:definitionofcrossratios}).

The OPE takes the form
\be
\label{eq:scalarscalarOPE}
\f(x_1)\f(x_2) &=& \sum_\cO f_{\f\f\cO} C_{a}(x_{12},\ptl_2)\cO^{a}(x_2),
\ee
where $\cO^{a}$ can have nonzero spin in general.  For $\cO^a$ to appear in the OPE of two scalars, it must transform in a spin-$\ell$ traceless symmetric tensor representation of $\SO(d)$.
\begin{exercise}
Prove this as follows. Show that $\<\cO^a|\f(x)|\f\>$ vanishes unless $\cO^a$ is a symmetric tensor.  (Tracelessness comes from restricting to irreducible representations of $\SO(d)$.) Argue that if $\<\cO^a|\f(x)|\f\>$ vanishes, then for any descendent $|\psi\>=P\cdots P|\cO\>$, the matrix element $\<\psi|\f(x)|\f\>$ vanishes as well.
\end{exercise}
\begin{exercise}
\label{exercise:elleven}
Using (\ref{eq:conformalthreeptfunction}), show that $f_{\f\f\cO}$ vanishes unless $\ell$ is even.
\end{exercise}

Assuming the points are configured appropriately, we can pair up the operators (12) (34) and perform the OPE between them,\footnote{Although our computation will make it look like we need $x_{3,4}$ to be sufficiently far from $x_{1,2}$, we will see shortly that the answer will be correct whenever we can draw any sphere separating $x_1,x_2$ from $x_3,x_4$.}
\begin{align}
\label{eq:fourptcalc}
\<
\contraction{}{\f}{(x_1)}{\f}
\f(x_1)&\f(x_2)
\contraction{}{\f}{(x_1)}{\f}
\f(x_3)\f(x_4)
\>\nn\\
&= \sum_{\cO,\cO'}f_{\f\f\cO}f_{\f\f\cO'} C_a(x_{12},\ptl_2)C_b(x_{34},\ptl_4)\<\cO^a(x_2)\cO'^b(x_4)\>\nn\\
&= \sum_\cO f_{\f\f\cO}^2 C_a(x_{12},\ptl_2)C_b(x_{34},\ptl_4)\frac{I^{ab}(x_{24})}{x_{24}^{2\De_\cO}}\nn\\
&= \frac{1}{x_{12}^{\De_\f} x_{34}^{\De_\f}}\sum_\cO f_{\f\f\cO}^2 g_{\De_\cO,\ell_\cO}(x_i),
\end{align}
where
\be
\label{eq:olddefinitionofg}
g_{\De,\ell}(x_i) &\equiv& x_{12}^{\De_\f} x_{34}^{\De_\f} C_a(x_{12},\ptl_2)C_b(x_{34},\ptl_4)\frac{I^{ab}(x_{24})}{x_{24}^{2\De}}.
\ee
In (\ref{eq:fourptcalc}), we have chosen an orthonormal basis of operators and used that
\be
\label{eq:canonicallynormalizedtwopt}
\<\cO^a(x)\cO'^b(0)\> &=& \de_{\cO\cO'} \frac{I^{ab}(x)}{x^{2\De_\cO}},
\ee
where $I^{ab}(x)=I^{\mu_1\cdots\mu_\ell,\nu_1\cdots\nu_\ell}(x)$ is the tensor in (\ref{eq:twopointfunctionofspinL}).

The functions $g_{\De,\ell}(x_i)$ are called {\it conformal blocks}.  Although it's not obvious from the way we defined them, it turns out they are actually functions of the conformal cross-ratios $u,v$ alone.  We thus have the conformal block decomposition
\be
g(u,v) &=& \sum_\cO f_{\f\f\cO}^2 g_{\De_\cO,\ell_\cO}(u,v).
\ee
\begin{exercise}
Using the differential operator (\ref{eq:identicalscalaropeoperator}), show
\be
\label{eq:boundaryconditionforblock}
g_{\De,0}(u,v) &=& u^{\De/2}\p{1+\dots}.
\ee
\end{exercise}

\begin{exercise}
Using (\ref{eq:scalarscalarspinL}), argue that $x^{2\De_\f} C_{\f\f\cO}(x,\ptl)$ is independent of $\Delta_\phi$ for any spin of $\cO$. Conclude that $g_{\De,\ell}(u,v)$ is independent of $\Delta_\phi$. (This is a special property of conformal blocks for operators with identical scaling dimensions.)
\end{exercise}

\subsection{In Radial Quantization}

A conformal block represents the contribution of a single conformal multiplet to a four-point function.  It is instructive to understand it in radial quantization.  Along the way, we'll explain why the blocks are functions of the cross-ratios $u,v$ alone.

Let us pick an origin such that $|x_{3,4}|\geq |x_{1,2}|$, so that
\be
\label{eq:fourptradial}
\<\f(x_1)\f(x_2)\f(x_3)\f(x_4)\> &=& \<0|\mathcal{R}\{\f(x_3)\f(x_4)\}\mathcal{R}\{\f(x_1)\f(x_2)\}|0\>.\quad
\ee
For a primary operator $\cO$, let $|\cO|$ be the projector onto the conformal multiplet of $\cO$,
\be
|\cO| &\equiv& \sum_{\a,\b=\cO,P\cO,PP\cO,\dots} |\a\>\mathcal{N}^{-1}_{\a\b}\<\b|,\qquad
\mathcal{N}_{\a\b} \equiv \<\a|\b\>.
\ee
The identity is the sum of these projectors over all primary operators.
\be
\mathbf{1} &=& \sum_\cO |\cO|.
\ee
Inserting this into (\ref{eq:fourptradial}) gives
\be
\label{eq:insertingprojector}
\<\f(x_1)\f(x_2)\f(x_3)\f(x_4)\> &=& \sum_\cO\<0|\mathcal{R}\{\f(x_3)\f(x_4)\}|\cO|\mathcal{R}\{\f(x_1)\f(x_2)\}|0\>.\nn\\
\ee
Each term in the sum is a conformal block times a squared OPE coefficient and some conventional powers of $x_{ij}$,
\be
\label{eq:newdefinitionofg}
\<0|\mathcal{R}\{\f(x_3)\f(x_4)\}|\cO|\mathcal{R}\{\f(x_1)\f(x_2)\}|0\> &=& \frac{f_{\f\f\cO}^2}{x_{12}^{2\De_\f}x_{34}^{2\De_\f}}g_{\De_\cO,\ell_\cO}(u,v).\qquad
\ee
\begin{exercise}
Verify the equivalence between (\ref{eq:newdefinitionofg}) and (\ref{eq:olddefinitionofg}) by performing the OPE between $\f(x_3)\f(x_4)$ and $\f(x_1)\f(x_2)$.
\end{exercise}

This expression makes it clear why $g_{\De,\ell}(u,v)$ is a function of $u$ and $v$: the projector $|\cO|$ commutes with all conformal generators (by construction).  Thus, the object above satisfies all the same Ward identities as a four-point function of primaries, and it must take the form (\ref{eq:fourptfunctionofprimaries}).  In path integral language, we can think of $|\cO|$ as a new type of  surface operator.  Here, we've inserted it on a sphere separating $x_{1,2}$ from $x_{3,4}$.

\subsection{From the Conformal Casimir}

We can now give a simple and elegant way to compute the conformal block, due to Dolan \& Osborn \cite{DO2}.
Recall that the conformal algebra is isomorphic to $\SO(d+1,1)$, with generators $L_{ab}$ given by (\ref{eq:conformalgeneratorssodplus11}).  The Casimir $C=-\frac 1 2 L^{ab}L_{ab}$ acts with the same eigenvalue on every state in an irreducible representation.  
\begin{exercise}
Show that this eigenvalue is given by
\be
C|\cO\> &=& \l_{\De,\ell}|\cO\>,\nn\\
\l_{\De,\ell} &\equiv& \De(\De-d)+\ell(\ell+d-2).
\ee
\end{exercise}
It follows that $C$ gives this same eigenvalue when acting on the projection operator $|\cO|$ from either the left or right,
\be
C|\cO|=|\cO| C = \l_{\De,\ell}|\cO|.
\ee

Let $\cL_{ab,i}$ be the differential operator giving the action of $L_{ab}$ on the operator $\f(x_i)$.  Note that
\be
(\cL_{ab,1}+\cL_{ab,2})\f(x_1)\f(x_2)|0\> &=& \p{[L_{ab},\f(x_1)]\f(x_2)+\f(x_1)[L_{ab},\f(x_2)]}|0\>\nn\\
&=& L_{ab}\f(x_1)\f(x_2)|0\>.
\ee
Thus, 
\be
C\f(x_1)\f(x_2)|0\> &=& \cD_{1,2}\f(x_1)\f(x_2)|0\>,\nn\\
\textrm{where}\qquad\cD_{1,2} &\equiv& -\frac 1 2(\cL^{ab}_{1}+\cL^{ab}_{2})(\cL_{ab,1}+\cL_{ab,2}).
\ee
We then have
\begin{align}
\cD_{1,2}\<0|&\mathcal{R}\{\f(x_3)\f(x_4)\}|\cO|\mathcal{R}\{\f(x_1)\f(x_2)\}|0\>\nn\\
&=
\<0|\mathcal{R}\{\f(x_3)\f(x_4)\}|\cO| C\mathcal{R}\{\f(x_1)\f(x_2)\}|0\>\nn\\
&= \l_{\De,\ell}\<0|\mathcal{R}\{\f(x_3)\f(x_4)\}|\cO|\mathcal{R}\{\f(x_1)\f(x_2)\}|0\>.
\end{align}
Plugging in (\ref{eq:newdefinitionofg}), we find that $g_{\De,\ell}$ satisfies the differential equation 
\be
\label{eq:conformalcasimir}
\cD g_{\De,\ell}(u,v) &=& \l_{\De,\ell} g_{\De,\ell}(u,v),
\ee
where the second-order differential operator $\cD$ is given by
\be
\cD &=& 2(z^2(1-z)\ptl_z^2-z^2 \ptl_z) + 2(\bar z^2 (1-\bar z)\ptl_{\bar z}^2-\bar z^2 \ptl_{\bar z})\nn\\
&& + 2(d-2)\frac{z\bar z}{z-\bar z}((1-z)\ptl_z - (1-\bar z)\ptl_{\bar z}).
\ee

Eq.~(\ref{eq:conformalcasimir}), together with the boundary condition~(\ref{eq:boundaryconditionforblock}) (and its generalization to nonzero spin, which we give shortly), then determines the conformal block $g_{\Delta,\ell}(u,v)$.  In even dimensions, the Casimir equation can be solved analytically.  For example, in 2d and 4d \cite{DO1,DO2},
\be
\label{eq:explicitblock2d}
g_{\De,\ell}^{(2d)}(u,v) &=& k_{\De+\ell}(z)k_{\De-\ell}(\bar z) + k_{\De-\ell}(z)k_{\De+\ell}(\bar z),\\
\label{eq:explicitblock4d}
g_{\De,\ell}^{(4d)}(u,v) &=& \frac{z \bar z}{z-\bar z}\p{k_{\De+\ell}(z)k_{\De-\ell-2}(\bar z) - k_{\De-\ell-2}(z)k_{\De+\ell}(\bar z)},\\
k_\beta(x) &\equiv& x^{\beta/2}{}_2F_1\p{\frac \beta 2, \frac \beta 2, \beta, x}.
\ee
In odd dimensions, no explicit formula in terms of elementary functions is known.  However the blocks can still be computed in a series expansion using the Casimir equation or alternative techniques like recursion relations.

\subsection{Series Expansion}

It will be helpful to understand the series expansion of the conformal blocks in more detail.
 The ``radial coordinates" of \cite{Pappadopulo:2012jk,Hogervorst:2013sma} are ideal for this purpose.
 Using conformal transformations, we can place all four operators on a plane in the configuration shown in figure~\ref{fig:rho}.  This makes it clear that the conformal block expansion is valid whenever $|\rho|<1$.

\begin{figure}
\begin{center}
\includegraphics[width=0.55\textwidth]{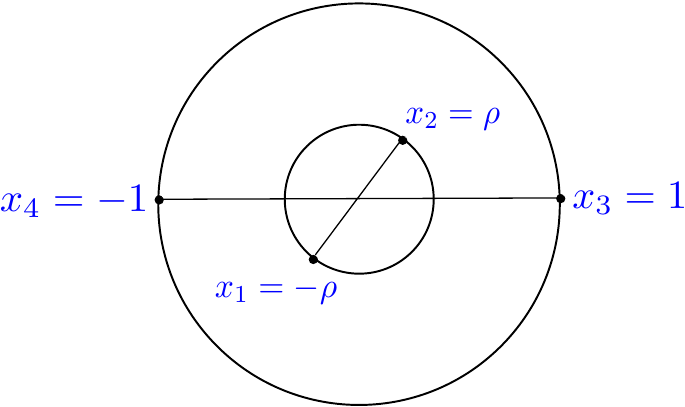}
\end{center}
\caption{Any four points can be brought to the above configuration using conformal transformations. (Figure from \cite{Hogervorst:2013sma}.)  \label{fig:rho}}
\end{figure}

\begin{exercise}
Show that $\rho=re^{i\theta}$ is related to $z$ via
\be
\label{eq:radialcoordinatedefinition}
\rho = \frac{z}{(1+\sqrt{1-z})^2},\qquad z = \frac{4\rho}{(1+\rho)^2}
\ee
(and similarly for $\bar\rho=r e^{-i\theta}$ and $\bar z$).
\end{exercise}

In radial quantization, this corresponds to placing cylinder operators (\ref{eq:definitionofcylinderop}) at diametrically opposite points $\pm \bn$ and $\pm \bn'$ on $S^{d-1}$, with $\cos\th=\bn\.\bn'$, and with the pairs separated by time $\tau=-\log r$ (figure~\ref{fig:cylinderconfig}).  The conformal block is then
\be
\label{eq:blockintermsofpsi}
 g_{\De,\ell}(u,v) &=& \<\psi(\bn)||\cO|e^{-\tau D}|\psi(\bn')\>,
\ee
where we've defined the state\footnote{The factor $2^{\De_\f}=\<\f_\mathrm{cyl.}(0,\bn)\f_\mathrm{cyl.}(0,-\bn)\>^{-1}$ comes from transforming $x_{12}^{-2\De_\f}$ to the cylinder (exercise!).}
\be
|\psi(\bn)\> &\equiv& \frac{2^{\De_\f}}{f_{\f\f\cO}}\phi_\mathrm{cyl.}(0,\bn)\phi_\mathrm{cyl.}(0,-\bn)|0\>.
\ee

\begin{figure}
\begin{center}
\includegraphics[width=0.75\textwidth]{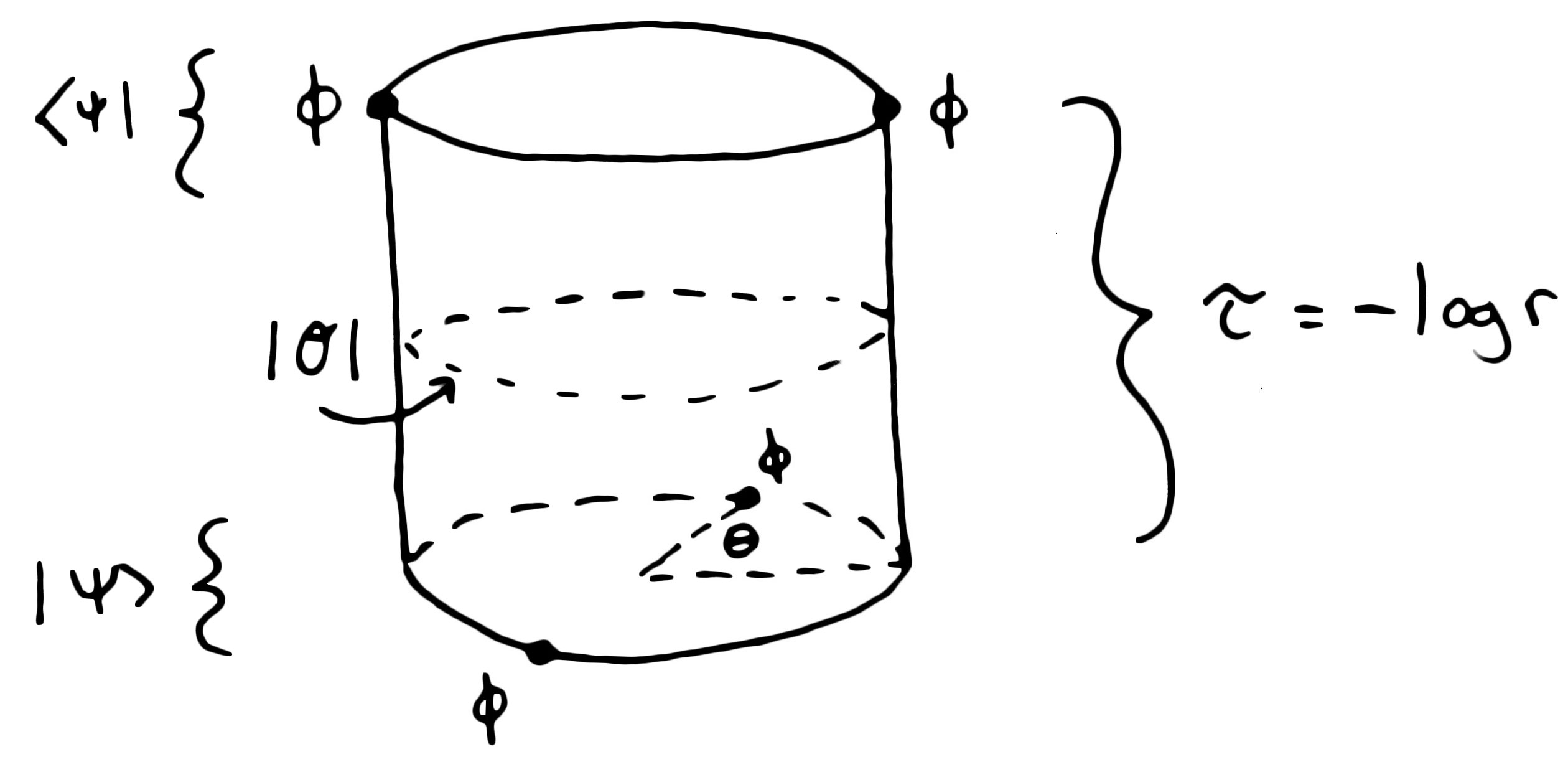}
\end{center}
\caption{Configuration on the cylinder corresponding to (\ref{eq:blockintermsofpsi}).  \label{fig:cylinderconfig}}
\end{figure}

A descendant $P^{\mu_1}\cdots P^{\mu_n}|\cO\>$ has energy $\De+n$ on the cylinder.  Within the $n$-th energy level, the $\SO(d)$ spins that appear are
\be
\label{eq:rangeofjs}
j \in \{\ell+n,\ell+n-2,\dots,\mathrm{max}(\ell-n,\ell+n\,\,\mathrm{mod}\,\,2)\}.
\ee
Consider a set of descendent states $|n,j\>^{\mu_1\cdots\mu_j}$ with energy $\De+n$ and spin $j$. They contribute
\be
r^{\De+n} \<\psi(\bn)|n,j\>^{\mu_1\cdots\mu_j}{}_{\mu_1\cdots\mu_j}\<n,j|\psi(\bn')\>.
\label{eq:definitespinandenergy}
\ee
By rotational invariance,
\be
\<\psi(\bn)|n,j\>^{\mu_1\cdots\mu_j} &\propto& \bn^{\mu_1}\cdots\bn^{\mu_j}-\mathrm{traces}.
\ee
Because $|\psi(\bn)\>=|\psi(-\bn)\>$, $j$ must be even (and thus $n$ is even).
The contraction of two traceless symmetric tensors is a Gegenbauer polynomial,
\be
C_j^{\frac{d-2}{2}}(\bn\cdot\bn') &\propto& (\bn^{\mu_1}\cdots\bn^{\mu_j}-\mathrm{traces})(\bn'_{\mu_1}\cdots\bn'_{\mu_j}-\mathrm{traces}),
\ee
so (\ref{eq:definitespinandenergy}) becomes
\be
r^{\De+n} \<\psi(\bn)|n,j\>^{\mu_1\cdots\mu_j}{}_{\mu_1\cdots\mu_j}\<n,j|\psi(\bn')\> &\propto& r^{\De+n}C_j^{\frac{d-2}{2}}(\cos\th).
\ee

Summing over descendants, we find
\be
\label{eq:seriesexpansion}
g_{\De,\ell}(u,v) &=& \sum_{\substack{n=0,2,\dots \\ j}} B_{n,j}r^{\De+n}C_j^{\frac{d-2}{2}}(\cos\th),\label{eq:seriesforblock}
\ee
where $j$ ranges over the values in (\ref{eq:rangeofjs}) and $B_{n,j}$ are constants.  
Notice a few properties:
\begin{itemize}
\item The leading term in the $r$-expansion comes from the primary state $|\cO\>$ with $n=0$ and $j=\ell$. This can be used as a boundary condition in the Casimir equation to determine the higher coefficients $B_{n,j}$.
\item The $B_{n,j}$ are positive in a unitary theory because they are given by norms of projections of $|\psi\>$ onto energy and spin eigenstates.
\item The $B_{n,j}$ are rational functions of $\De$.  This follows because the Casimir eigenvalue $\l_{\De,\ell}$ is polynomial in $\De$, or alternatively from the fact that the differential operators $C_a(x,\ptl)$ appearing in the OPE (\ref{eq:scalarscalarOPE}) have a series expansion in $x$ with rational coefficients, see exercise~\ref{exercise:seriesfordiffops}. 
\end{itemize}

\begin{exercise}
Expand $g^{(2d)}_{\De,\ell}(u,v)$ and $g^{(4d)}_{\De,\ell}(u,v)$ to the first few orders in $r$, and check these properties.  Verify that some of the coefficients $B_{n,j}$ become negative when $\De$ violates the unitarity bound.
\end{exercise}

\begin{exercise}
By rewriting in terms of $r,\th$ and using (\ref{eq:seriesexpansion}), show that even spin blocks are invariant under $x_1\leftrightarrow x_2$ or $x_3\leftrightarrow x_4$,
\be
\label{eq:invariantunderonetwo}
g_{\De,\ell}(u,v) &=& g_{\De,\ell}\p{\frac{u}{v},\frac 1 v},\qquad(\textrm{$\ell$ even}).
\ee
\end{exercise}

\section{The Conformal Bootstrap}

\subsection{OPE Associativity and Crossing Symmetry}

We've gotten pretty far using symmetries and basic principles of quantum field theory.  We  classified operators into primaries and descendants. We established the OPE, which determines $n$-point functions as sums of $(n-1)$-point functions,
\be
\label{eq:usingOPEtoreducecorrelator}
\<\cO_1(x_1)\cO_2(x_2)\cdots\cO_n(x_n)\> &=& \sum_k C_{12k}(x_{12},\ptl_2)\<\cO_2(x_2)\cdots\cO_n(x_n)\>.\nn\\
\ee
And we showed that the differential operators $C_{ijk}(x,\ptl)$ are determined by conformal symmetry in terms of dimensions $\De_i$, spins, and OPE coefficients $f_{ijk}$. 

Now it's time to implement the last step of the bootstrap program: impose consistency conditions and derive constraints.  Using the OPE, all correlation functions can be written in terms of the ``CFT data" $\De_i,f_{ijk}$.  Now suppose someone hands you a random set of numbers $\De_i,f_{ijk}$.  Does that define a consistent CFT?

\begin{figure}
\begin{center}
\includegraphics[width=0.9\textwidth]{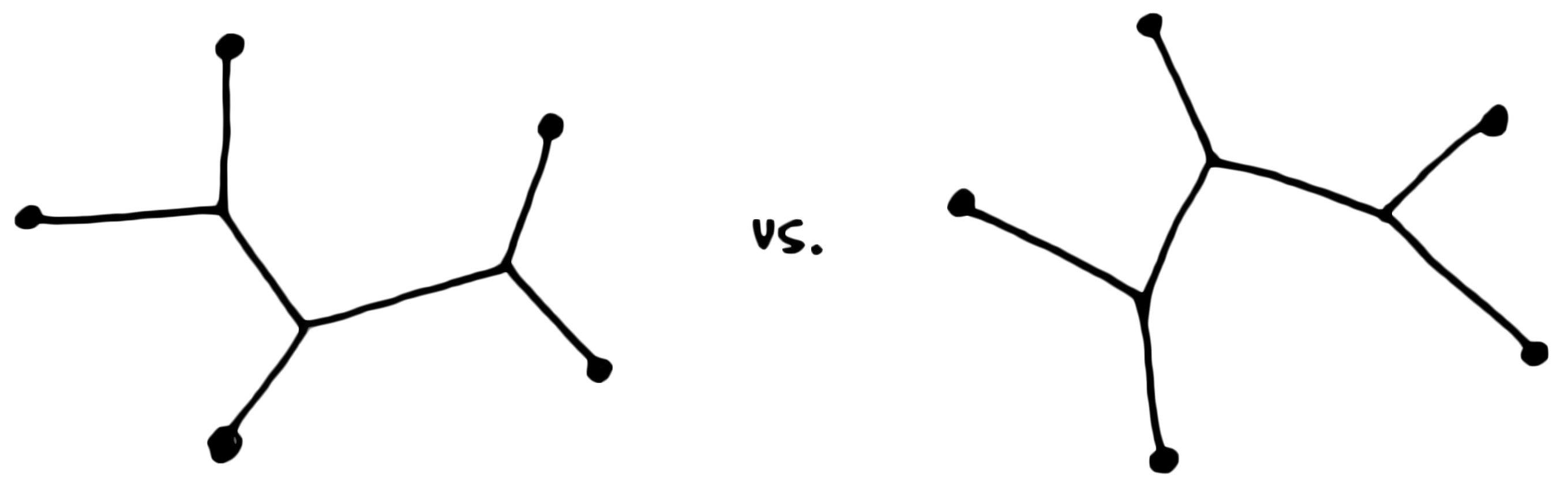}
\end{center}
\caption{Two different ways of evaluating a five-point function using the OPE\@.  Dots represent operators in the correlator, and vertices represent the OPE\@. The two ways differ by a crossing symmetry transformation (\ref{eq:graphicalcrossing}) applied to the left part of the diagram.\label{fig:opedifferentways}}
\end{figure}

The answer is: not always.  By doing the OPE (\ref{eq:usingOPEtoreducecorrelator}) between different pairs of operators in different orders (see figure~\ref{fig:opedifferentways}), we get naively different expressions for the same correlator in terms of CFT data.  These expressions should agree. 
This means the OPE should be associative,
\be
\contraction{}{
\frac 1 2\!\!\!\!\!
\cO_1\cO_2
}{}{\cO_3}
\contraction{}{\cO_1}{}{\cO_2}
\cO_1\cO_2\cO_3
&=&
\contraction{}{
\frac 1 2\!\!\!\!\!
\cO_1
}{}{
\contraction{}{\cO_2}{}{\cO_3}
\cO_2\cO_3
}
\contraction{\cO_1}{\cO_2}{}{\cO_3}
\cO_1\cO_2\cO_3,
\ee
or more explicitly,
\be
\label{eq:OPEassociativity}
C_{12i}(x_{12},\ptl_2)C_{i3j}(x_{23},\ptl_3)\cO_j(x_3) &=& C_{23i}(x_{23},\ptl_3)C_{1ij}(x_{13},\ptl_3)\cO_j(x_3).\nn\\
\ee
(We suppress spin indices for simplicity.)
Taking the correlator of both sides with a fourth operator $\cO_4(x_4)$ gives the {\it crossing symmetry equation}

\be
\label{eq:graphicalcrossing}
\hspace{-0.4in}
\setlength{\unitlength}{.4in}
\begin{picture}(7,1.7)(0,0.2)
\linethickness{1pt}
\put(1.7,0.4){\line(1,2){0.3}}
\put(1.7,1.6){\line(1,-2){0.3}}
\put(2,1){\line(1,0){0.8}}
\put(2.8,1){\line(1,2){0.3}}
\put(2.8,1){\line(1,-2){0.3}}
\put(1.2,1){\makebox(0,0){$\mathlarger{\sum}_i^{\phantom\cO}$}}
\put(4,1){\makebox(0,0){$=$}}
\put(1.65,1.8){\makebox(0,0){\small $1$}}
\put(1.65,0.2){\makebox(0,0){\small $2$}}
\put(3.15,1.8){\makebox(0,0){\small $4$}}
\put(3.15,0.2){\makebox(0,0){\small $3$}}
\put(5.38,1.8){\makebox(0,0){\small $1$}}
\put(5.38,0.2){\makebox(0,0){\small $2$}}
\put(6.6,1.8){\makebox(0,0){\small $4$}}
\put(6.6,0.2){\makebox(0,0){\small $3$}}
\put(6.3,1){\makebox(0,0){\small $\cO_i$}}
\put(2.4,1.2){\makebox(0,0){\small $\cO_i$}}
\put(5,1){\makebox(0,0){$\mathlarger{\sum}_i^{\phantom\cO}$}}
\put(6,0.6){\line(0,1){0.8}}
\put(5.5,0.35){\line(2,1){0.5}}
\put(6,0.6){\line(2,-1){0.5}}
\put(6,1.38){\line(2,1){0.5}}
\put(5.5,1.65){\line(2,-1){0.5}}
\end{picture}.
\ee
The left-hand side is the conformal block expansion of $\<\cO_1\cO_2\cO_3\cO_4\>$ in the $12\leftrightarrow 34$ channel, while the right-hand side is the expansion in the $14\leftrightarrow 23$ channel.
 
\begin{exercise}
\label{ex:crossingexercise}
Argue that by choosing different operators $\cO_4$ and taking linear combinations of derivatives, one can recover OPE associativity (\ref{eq:OPEassociativity}) from the crossing equation (\ref{eq:graphicalcrossing}). Conclude that crossing symmetry of all four-point functions implies crossing symmetry of all $n$-point functions (i.e.\ that any way of computing an $n$-point function using the OPE gives the same result).
\end{exercise}

The crossing equation (\ref{eq:graphicalcrossing}) is a powerful but complicated constraint on the CFT data.  The rest of this course will be devoted to studying its implications for the simplest possible case: a four-point function of identical scalars $\<\f(x_1)\f(x_2)\f(x_3)\f(x_4)\>$.  

\subsubsection{Additional Structures and Consistency Conditions}
\label{sec:additionalstructures}

Before jumping in, let us reflect on the implications of exercise~\ref{ex:crossingexercise}:  A solution to the crossing equations~(\ref{eq:graphicalcrossing}) gives a completely nonperturbative definition of flat-space correlation functions of local operators, without the need for a Lagrangian.  This is most of the way towards a full theory. However, some structures associated with local QFTs are missing, and additional structures might bring new consistency conditions.

Firstly, CFTs can admit extended objects like line operators, surface operators, boundaries, and interfaces. These objects have additional data associated with them, and it's possible to write down OPEs and crossing equations that relate this data to itself and the usual CFT data, see e.g. \cite{Liendo:2012hy,Gaiotto:2013nva}.  It is also interesting to consider correlation functions on manifolds not conformally equivalent to flat space. An example includes the theory at finite temperature (discussed in footnote~\ref{foot:finitetemperature}). This introduces more data, for example the one-point functions of local and extended operators on nontrivial manifolds.\footnote{It is known that this data is not determined by the local operator spectrum.  For example, pure Chern-Simons theory has no local operators at all, but has interesting nonlocal observables that depend on the gauge group and level \cite{Witten:1988hf}.  Also, 4d conformal gauge theories admit different sets of line operators for the same set of local operators \cite{Aharony:2013hda}.} Other interesting constraints come from studying CFTs in Lorentzian signature. Examples include bounds from energy positivity \cite{Hofman:2008ar}, dispersion relations \cite{Komargodski:2011vj,Nachtmann:1973mr,Komargodski:2012ek,Komargodski:2016gci}, and causality \cite{Maldacena:2015waa,Hartman:2015lfa}.

The full set of data and consistency conditions associated with a CFT is not known in general. However, we do have examples of constraints on local operators beyond the OPE and crossing equations. The most famous is modular invariance: the requirement that the partition function of a 2d CFT on the torus $T^2$ be invariant (or covariant) under large diffeomorphisms.  Imposing modular invariance is an additional step that must be performed after solving the crossing equations in 2d CFTs \cite{Moore:1988uz}.\footnote{2d is special because the space of states on a spatial slice $S^1\subset T^2$ is the same as the space of states in radial quantization, and thus modular invariance on $T^2$ directly constrains local operators.  This is not true in $d\geq 3$, so it is not clear how modular invariance on $T^d$ constrains local operators in that case.}

\subsection{Crossing Symmetry for Identical Scalars} 

For the rest of this course, we study the crossing equation for a four-point function of identical real scalars $\<\f(x_1)\f(x_2)\f(x_3)\f(x_4)\>$.  Let us summarize the consequences of conformal symmetry and unitarity for this case.
\begin{itemize}
\item We have the OPE
\be
\f(x_1)\f(x_2) &=& \sum_{\cO} f_{\f\f\cO} C_{\mu_1\cdots\mu_\ell}(x_{12},\ptl_2) \cO^{\mu_1\cdots\mu_\ell}(x_2).
\ee
We denote the dimension of $\cO$ by $\De$ and the spin by $\ell$. By exercise~\ref{exercise:elleven}, $\ell$ must be even.

\item We can choose a basis of operators such that the $\cO$'s are real and orthonormal, as in  (\ref{eq:canonicallynormalizedtwopt}).  Unitarity implies that the three-point coefficients $f_{\f\f\cO}$ are real in this basis (section~\ref{sec:realvscomplex}).

\item Each $\cO$ satisfies the unitarity bounds
\be
\label{eq:unitarityboundsummary}
\De &=& 0 \ \textrm{(unit operator), or}\nn\\
\De &\geq& \left\{
\begin{array}{ll}
\frac{d-2}{2} & (\ell=0),\\
\ell+d-2 & (\ell > 0).
\end{array}
\right.
\ee

\item We have the conformal block expansion
\be
\<\f(x_1)\f(x_2)\f(x_3)\f(x_4)\> &=& \frac{g(u,v)}{x_{12}^{2\De_\f}x_{34}^{2\De_\f}}\\
g(u,v) &=& \sum_\cO f_{\f\f\cO}^2 g_{\De,\ell}(u,v),
\ee
where $g_{\De,\ell}(u,v)$ are conformal blocks, and the cross ratios are
\be
u = z\bar z = \frac{x_{12}^2 x_{34}^2}{x_{13}^2 x_{24}^2},\qquad v=(1-z)(1-\bar z) =\frac{x_{23}^2 x_{14}^2}{x_{13}^2 x_{24}^2}.
\ee

\item Crossing symmetry is equivalent to the condition (\ref{eq:crossingsymmetry}) that our four-point function is invariant under $1\leftrightarrow 3$ or $2\leftrightarrow 4$,
\be
\label{eq:crossingeqsummary}
g(u,v) &=& \p{\frac{u}{v}}^{\De_\f} g(v,u).
\ee
Eq.~(\ref{eq:invariantunderonetwo}) shows that invariance of the four-point function under $1\leftrightarrow 2$ or $3\leftrightarrow 4$ is true block-by-block.  All other permutations can be generated from these.

\end{itemize}

We know at least two operators present in the $\f\x\f$ OPE: the unit operator and the stress tensor.  Normalizing $\f$ so that $\<\f(x)\f(0)\>=x^{-2\De_\f}$, we have $f_{\f\f\mathbf{1}}=1$.  The stress tensor three-point coefficient is set by Ward identities to be $f_{\f\f T_{\mu\nu}}\propto \De_\f/\sqrt{C_T}$, where $C_T$ is the coefficient of the two-point function of the canonically normalized stress tensor (\ref{eq:twopointfunctionofspinL}).  The factor of $1/\sqrt{C_T}$ relative to (\ref{eq:stresstensorward}) comes from choosing the basis of operators $\cO$ to be orthonormal.

\subsection{An Infinite Number of Primaries}

To get a feel for the crossing equation (\ref{eq:crossingeqsummary}), let us consider a simple limit: $z\to 0$ with $z=\bar z$. This corresponds to $x_2\to x_1$ with all four operators collinear.

Recall that the blocks go like $g_{\De,\ell}(u,v)\sim (z\bar z)^{\De/2}$ in this limit, so the left-hand side of (\ref{eq:crossingeqsummary}) is dominated by the smallest dimension operator in the OPE, the unit operator:
\be
\mathrm{LHS} &:& 1+\dots \qquad(z\to 0).
\ee

Crossing $u\leftrightarrow v$ corresponds to $(z,\bar z)\to (1-z,1-\bar z)$.  In the limit $z\to 0$, the crossed conformal blocks $g_{\De,\ell}(1-z,1-z)$ go like $\log z$.
\begin{exercise}
Check this for the explicit formulae (\ref{eq:explicitblock2d}) and (\ref{eq:explicitblock4d}).
\end{exercise}
Thus, each term on the right-hand side goes like
\be
\label{eq:rhsterms}
\textrm{each term, RHS} &:& z^{2\De_\f}\log z + \dots \qquad(z\to 0).
\ee
As $z\to 0$, any finite sum of terms of the form (\ref{eq:rhsterms}) vanishes.  Thus, for a sum of operators on the right-hand side to reproduce the unit operator on the left-hand side, we need an infinite number of primary operators.\footnote{This doesn't contradict the textbook statement that rational 2d CFTs contain a finite number of primary operators.  In that context, ``primary" refers to primary operators with respect to the Virasoro algebra.  Here, we are discussing primaries with respect to the global conformal group, which is $\mathrm{SL}(2,\R)\x\mathrm{SL}(2,\R)$ in 2d.  A single Virasoro representation contains an infinite number of global conformal representations.}

One can prove that as $z\to 0$, the sum on the right-hand side is dominated by operators of dimension $\Delta\sim 1/\sqrt z$ \cite{Pappadopulo:2012jk}.  In other words, the unit operator on the left-hand side maps to the large-$\De$ asymptotics of the sum over operators on the right-hand side.  This is a general feature of the crossing equation --- it cannot be satisfied block-by-block.

One can also show \cite{Pappadopulo:2012jk} that the conformal block expansion converges exponentially in $\De$ whenever $|\rho|\leq 1$, where $\rho$ is defined in (\ref{eq:radialcoordinatedefinition}).  In particular, this means that both sides of the crossing equation converge exponentially in a finite neighborhood of the point $z=\bar z=\frac 1 2$, which will play an important role in the next section.

Analyzing different limits of the crossing equation can give other information about the CFT spectrum.  For example, the limit $z\to 0$ with $\bar z$ fixed gives information about operators with large spin \cite{Alday:2007mf,Fitzpatrick:2012yx,Komargodski:2012ek,Alday:2015ewa}.

\subsection{Bounds on CFT Data}
\label{sec:bounds}

The crossing equation (\ref{eq:crossingeqsummary}) has been known for decades.  However, little progress was made in solving it for CFTs in $d\geq 3$ until 2008, in a breakthrough paper by Rattazzi, Rychkov, Tonni, and Vichi \cite{Rattazzi:2008pe}.  Instead of trying to solve the crossing equation exactly, their insight was to derive bounds on CFT data by studying the crossing equation geometrically.  Crucially, their methods let one make rigorous statements about some of the CFT data (for example, the lowest few operator dimensions), without having to compute all of it.

The basic idea is simple.  Let us write the crossing equation as
\be
\label{eq:crossingononeside}
\sum_\cO f_{\f\f\cO}^2 \underbrace{\p{v^{\De_\f}g_{\De,\ell}(u,v)-u^{\De_\f} g_{\De,\ell}(v,u)}}_{F_{\De,\ell}^{\De_\f}(u,v)} &=& 0.
\ee
Abstractly, we can think of the functions $F_{\De,\ell}^{\De_\f}(u,v)$ as vectors $\vec{F}_{\De,\ell}^{\De_\f}$ in the (infinite-dimensional) vector space of functions of $u$ and $v$.  Recall that the coefficients $f_{\f\f\cO}^2$ are positive, so (\ref{eq:crossingononeside}) has the form
\be
\label{eq:abstractcrossingeq}
\sum_{\De,\ell} p_{\De,\ell} \vec{F}_{\De,\ell}^{\De_\f} &=& 0,\qquad p_{\De,\ell}\geq 0,
\ee
where $\De,\ell$ run over dimensions and spins of operators in the $\f\x\f$ OPE.  

\begin{figure}
\begin{center}
\includegraphics[width=0.85\textwidth]{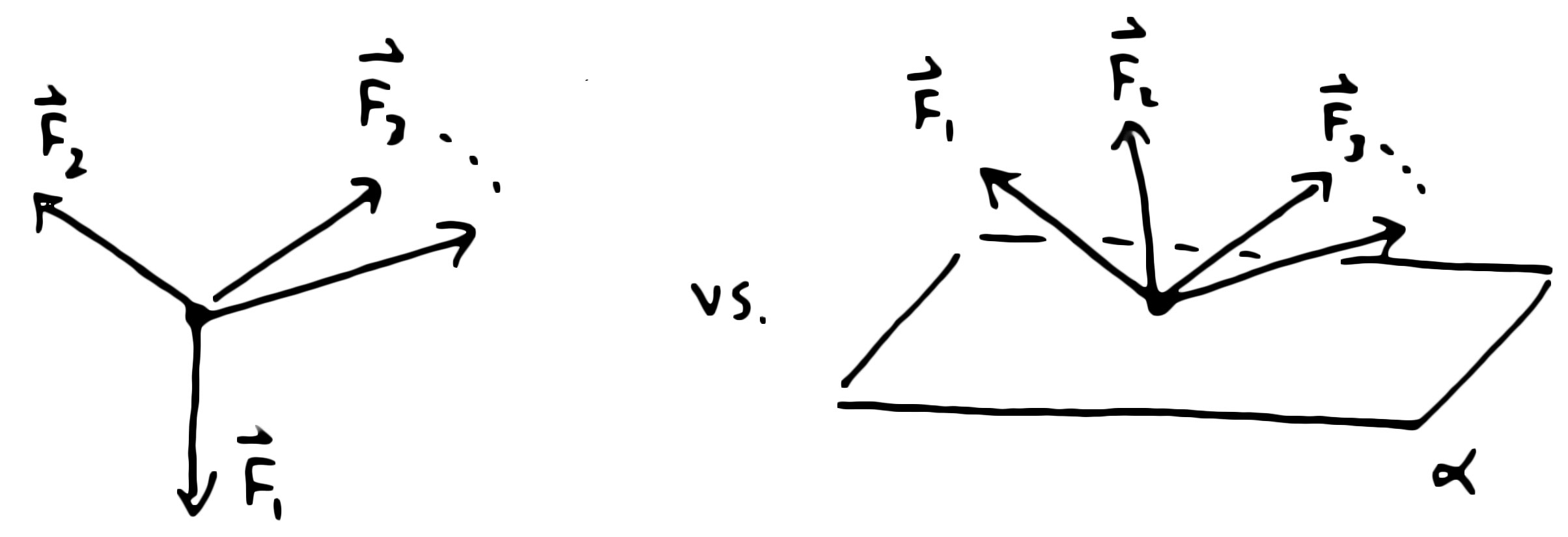}
\end{center}
\caption{On the left, a bunch of vectors that can sum to zero with positive coefficients.  On the right, a bunch of vectors that can't.  In the latter case, it's possible to find a separating plane $\a$.\label{fig:vectorpossibilities}}
\end{figure}

Equation~(\ref{eq:abstractcrossingeq}) says that a bunch of vectors sum to zero with positive coefficients.  This may or may not be possible, depending on the vectors.  The left-hand side of figure~\ref{fig:vectorpossibilities} shows a case where it's possible, and the right-hand side shows a case where it's impossible.  The way to distinguish these cases is to search for a {\it separating plane} $\alpha$ through the origin such that all the vectors $\vec{F}_{\De,\ell}^{\De_\f}$ lie on one side of $\alpha$.  If $\alpha$ exists, then the  $\vec{F}_{\De,\ell}^{\De_\f}$ cannot satisfy crossing, for any choice of coefficients $p_{\De,\ell}=f_{\f\f\cO}^2$.  This suggests the following procedure for bounding CFT data.

\begin{algorithm}[Bounding Operator Dimensions]
\label{thealgorithm}
\qquad

\begin{enumerate}
\item Make a hypothesis for which dimensions and spins $\De,\ell$ appear in the $\f\x\f$ OPE.
\item Search for a linear functional $\alpha$ that is nonnegative acting on all $\vec F_{\Delta,\ell}^{\Delta_\f}$ satisfying the hypothesis,
\be
\alpha(\vec F_{\Delta,\ell}^{\De_\f})\geq 0,
\ee
and strictly positive on at least one operator (usually taken to be the unit operator).
\item If $\alpha$ exists, the hypothesis is wrong.  We see this by applying $\alpha$ to both sides of (\ref{eq:abstractcrossingeq}) and finding a contradiction.
\end{enumerate}
\end{algorithm}

A slight modification of this algorithm also lets one bound OPE coefficients \cite{Caracciolo:2009bx}.

\subsection{An Example Bound}
\label{sec:boundsexample}

Let's work through an example.\footnote{An early version of this example is due to Sheer El-Showk, and this specific implementation is due to Jo\~ao Penedones and Pedro Vieira.}  Consider a 2d CFT with a real scalar primary $\f$ of dimension $\De_\f=\frac 1 8$.
Project the crossing equation onto a two-dimensional subspace with the linear map
\be
\label{eq:vecv}
\vec v(F) &=& \p{H\p{\frac 1 2,\frac 3 5} - H\p{\frac 1 2, \frac 1 3}, H\p{\frac 1 2,\frac 3 5} - H\p{\frac 1 3, \frac 1 4}}\in \R^2,\nn\\
&&\quad \textrm{where}\quad H(z,\bar z) = \frac{F(u,v)}{u^{\De_\f}-v^{\De_\f}},\nn\\
&&\qquad \qquad \qquad\,\,\,\,\,\,\, u=z\bar z,\nn\\
&&\qquad \qquad \qquad\,\,\,\,\,\,\, v= (1-z)(1-\bar z).
\ee
By linearity, the vectors $\vec v(F_{\De,\ell}^{\De_\f})$ also sum to zero with positive coefficients,
\be
\label{eq:finitedimensionalcrossing}
\sum_{\De,\ell} p_{\De,\ell} \vec v(F_{\De,\ell}^{\De_\f}) &=& 0.
\ee

In figure~\ref{fig:twodvectorsexample}, we plot $\vec v(F_{\De,\ell}^{\De_\f})$ for all $\De,\ell$ satisfying the unitarity bounds (\ref{eq:unitarityboundsummary}), where the conformal blocks are given by (\ref{eq:explicitblock2d}).  We have normalized the vectors so that they are easy to see, since changes in normalization can be absorbed into the coefficients $p_{\De,\ell}$.

\begin{figure}
\begin{center}
\includegraphics[width=0.9\textwidth]{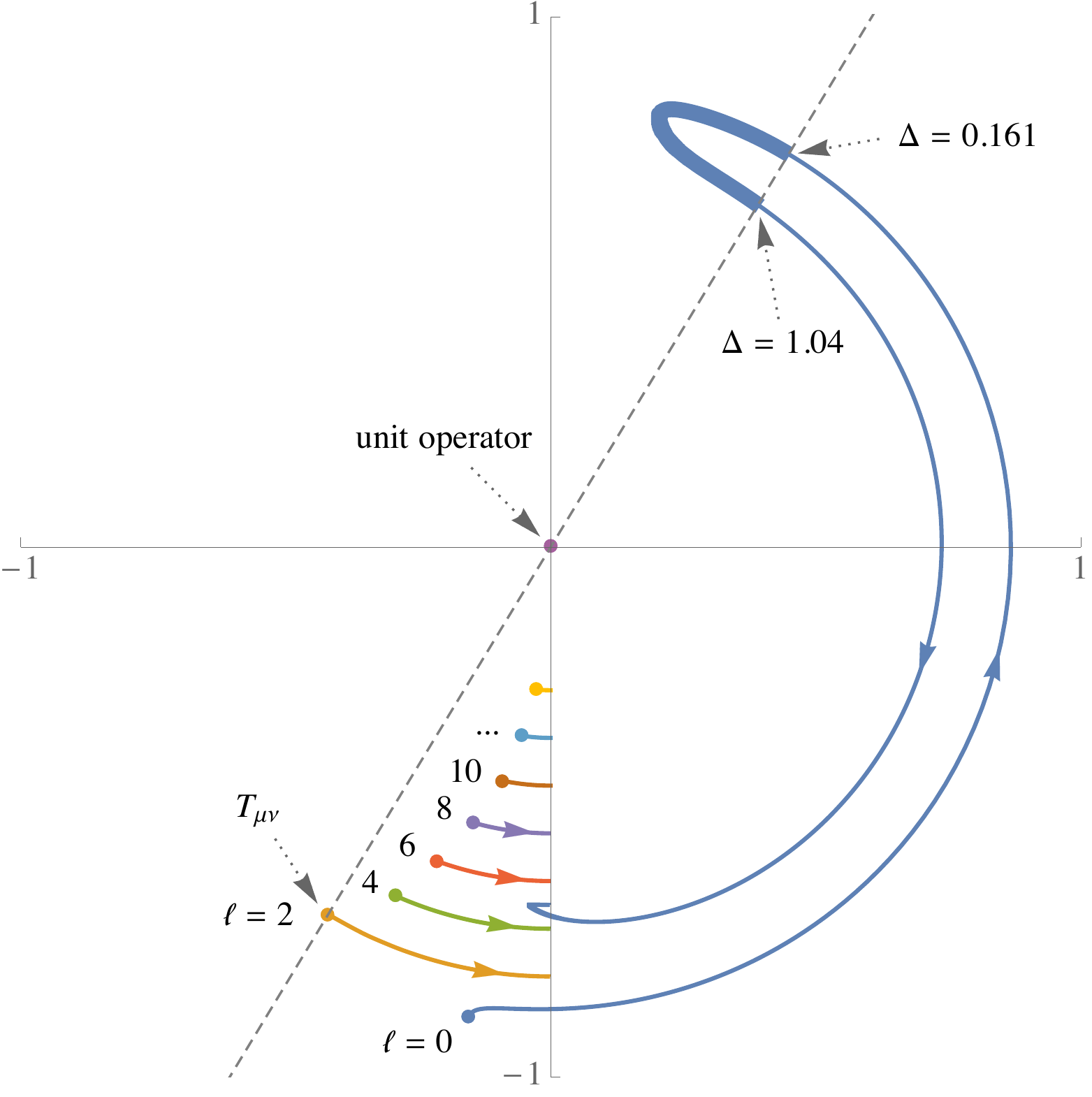}
\end{center}
\caption{Vectors $\vec v(F_{\De,\ell}^{\De_\f})$ for all values of $\De,\ell$ satisfying the 2d unitarity bound $\De\geq \ell$, with $\ell$ even.  Dots represent vectors at the unitarity bound $\De=\ell$.  As $\De$ varies, $\vec v(F_{\De,\ell}^{\De_\f})$ sweeps out a curve starting at the dot and approaching  the negative $y$-axis as $\De\to \oo$. The curves for spins $\ell=16,18,\dots$ look similar and converge quickly as $\ell\to \oo$, so we have not included them in the figure.  All vectors are normalized for visual simplicity, except for the unit operator $\vec v(F_{0,0}^{\De_\f})=\vec 0$.  The dashed line splits the figure into two half-spaces with the stress tensor $\vec v(F_{2,2}^{\De_\f})$ on the boundary.  The thicker region of the $\ell=0$ curve, in a different half-space from the rest of the figure, corresponds to scalars with dimension $\De\in[0.161,1.04]$.}
\label{fig:twodvectorsexample} 
\end{figure}

As $\De$ varies from $\ell$ (the unitarity bound) to $\oo$, $\vec v(F_{\De,\ell}^{\De_\f})$ sweeps out a curve.  The curves for higher spin operators $\ell\geq 2$ are extremely simple, converging quickly at large $\De$.  The scalar curve is more interesting. It circles counterclockwise partway around the origin before circling back and converging as $\De\to \oo$.  The region $\De\in[0.161,1.04]$ of the scalar curve lies in a different half space from the other curves.  To satisfy (\ref{eq:finitedimensionalcrossing}), we must include at least one vector from this region.  Thus, we immediately conclude
\begin{claim}
In a unitary 2d CFT with a real operator $\f$ of dimension $\De_\f=\frac 1 8$, there must exist a scalar in the $\f\x\f$ OPE with dimension $\De\in[0.161,1.04]$.
\label{claimboundone}
\end{claim}
\begin{proof}
We have already given the proof, but let us rephrase it in terms of Algorithm~\ref{thealgorithm}.
\begin{itemize}
\item Suppose (for a contradiction) that there are no scalars with $\De\in[0.161,1.04]$ in the $\f\x\f$ OPE.
\item Let
\be
\alpha(F) &=& \vec u \. \vec v(F),
\ee
where $\vec v(F)$ is defined in (\ref{eq:vecv}) and $\vec u\in \R^2$ is normal to the dashed line, pointing to the bottom right in figure~\ref{fig:twodvectorsexample}.  Note that $\alpha(F_{\De,\ell}^{\De_\f})\geq 0$ for all $\De,\ell$ satisfying our hypothesis.  Further, $\alpha$ is strictly positive on at least one vector appearing in the $\f\x\f$ OPE. (To establish this, we could rotate $\vec u$ slightly so that $\alpha$ is strictly positive on the stress tensor vector. Alternatively, we could use the fact that at least one higher dimension operator must appear in the $\f\x\f$ OPE.)
\item Applying $\alpha$ to both sides of (\ref{eq:abstractcrossingeq}), we find a contradiction: $0>0$.
\end{itemize}
\end{proof}

\begin{exercise}
\label{exercise:asymptotics}
Check that $\alpha(F_{\De,\ell}^{\De_\f})\geq 0$ is true asymptotically as $\De\to \oo$ and $\ell\to \oo$.  Convince yourself that the proof of Claim~\ref{claimboundone} could be made rigorous to a mathematician's standards.
\end{exercise}

Fiddling around with two-dimensional vectors has yielded a surprisingly strong result.  The 2d Ising CFT is an example of a unitary theory with a real scalar $\s$ (the spin operator) with dimension $\De_\s=\frac 1 8$.  The lowest dimension scalar in the $\s\x\s$ OPE is the energy operator $\e$, which has $\De_\e=1$.  So our upper bound $\De_\mathrm{scalar}\leq 1.04$ is within $4$ percent of being saturated by an actual theory!

\subsection{Numerical Techniques}
\label{sec:numericaltechniques}

The bound $\De_\mathrm{scalar}\in[0.161,1.04]$ is not particularly special.  If we had picked a different two-dimensional subspace (\ref{eq:vecv}), we would have gotten different numbers.  We might also consider higher-dimensional subspaces and derive even stronger results.  Although it is possible to prove bounds by hand as we did in the previous subsection, computerized searches are the current state-of-the-art.  In this section, we describe some of the techniques involved.

The hard part of Algorithm~\ref{thealgorithm} is the middle step: finding a functional $\alpha$ such that
\be
\label{eq:positivityconstraints}
\alpha(\vec F_{\De,\ell}^{\De_\f}) \geq 0,\qquad\textrm{for all $\De,\ell$ satisfying our hypothesis}.
\ee
If we want to use a computer, we have two immediate difficulties:
\begin{enumerate}
\item The space of possible $\alpha$'s is infinite-dimensional.
\item There are an infinite number of positivity constraints (\ref{eq:positivityconstraints}) --- one for each $\De,\ell$ satisfying our hypothesis.  Our hypothesis usually allows $\ell$ to range from $0$ to $\oo$, and $\De$ to vary continuously (aside from a few discrete values).\footnote{This is due to ignorance about the spectrum. Although physical CFT spectra should be discrete, we don't know exactly which discrete values $\De$ takes, and so we must include positivity constraints for continuously varying $\De$.}
\end{enumerate}

The first difficulty is easy to fix.  Instead of searching the infinite-dimensional space of all functionals, we restrict to a finite-dimensional subspace.  If we find $\alpha$ in our subspace that satisfies the positivity constraints, we can immediately rule out our hypothesis about the spectrum.  If we can't find $\alpha$, then we can't conclude anything about the spectrum: either no functional exists, or we just weren't searching a big enough subspace.

In the example from section~\ref{sec:boundsexample}, we restricted $\alpha$ to linear combinations of the components of $\vec v(F)$ in (\ref{eq:vecv}).  For numerical computations, we usually take linear combinations of derivatives around the crossing-symmetric point $z=\bar z = \frac 1 2$,
\be
\label{eq:choiceoffunctionals}
\alpha(F) &=& \sum_{m+n\leq \Lambda} a_{mn}\ptl_z^m \ptl_{\bar z}^n F(z,\bar z)|_{z=\bar z=\frac 1 2},
\ee
where $\Lambda$ is some cutoff.  The functional $\alpha$ is now parameterized by a finite number of coefficients $a_{mn}$, and a computer can search over these coefficients.\footnote{Note that $F(z,\bar z)$ is symmetric under $z\leftrightarrow \bar z$ (because $u$ and $v$ are), so we can restrict $m\leq n$.  Also, $F(z,\bar z)$ is odd under $(z,\bar z)\to (1-z,1-\bar z)$, so we can restrict to $m+n$ odd. This gives $\frac 1 2\lfloor \frac{\Lambda+1}{2}\rfloor(\lfloor \frac{\Lambda+1}{2}\rfloor+1)$ coefficients.}\footnote{Sometimes these bounds appear to converge as $\Lambda$ increases, justifying post-hoc the choice of subspace (\ref{eq:choiceoffunctionals}).  However, this subspace is not always obviously the best choice.  New results might come from studying different points in the $z,\bar z$ plane, integrating against kernels $K(z,\bar z)$, or doing something more exotic.  For example, the limit $z\to 0$, with $\bar z$ fixed is known to encode interesting information about high spin operators. Finding the optimal space of functionals is an open problem.}

Getting around the second difficulty takes more care.  To solve the inequalities (\ref{eq:positivityconstraints}) on a computer, we must encode them with a finite amount of data.  It is usually sufficient to restrict $\ell\leq \ell_\mathrm{max}$ for some large cutoff $\ell_\mathrm{max}$.  After we find $\alpha$, we can go back and check afterwards that it satisfies $\alpha(F_{\De,\ell}^{\De_\f})\geq 0$ for $\ell>\ell_\mathrm{max}$, as in exercise~\ref{exercise:asymptotics}.

To deal with the continuous infinity of $\De$'s, three techniques have been used in the literature:
\begin{itemize}
\item Discretize $\De$ with a small spacing and impose a cutoff $\De_\mathrm{max}$.  We then have a finite set of linear inequalities for $a_{mn}$, which can be solved using {\it linear programming}.  This is the approach taken in the original paper on CFT bounds \cite{Rattazzi:2008pe}.

\item Use a version of the simplex algorithm (underlying many linear programming solvers) that is customized to handle continuously varying constraints, see \cite{El-Showk:2014dwa,Paulos:2014vya}.

\item Approximate the constraints (\ref{eq:positivityconstraints}) as positivity conditions on polynomials and use {\it semidefinite programming} \cite{Poland:2011ey,Kos:2013tga,Kos:2014bka,Simmons-Duffin:2015qma}.  Appendix~\ref{app:semidefinite} explains the basic idea.

\end{itemize}

\subsection{Improving on our Hand-Computed Bound}

Let us compute an upper bound on the lowest-dimension scalar in the $\f\x\f$ OPE using a computer search. We will assume a $\Z_2$ symmetry under which $\f$ is odd so that $\f$ doesn't appear in its own OPE\@. The procedure is as follows
\begin{enumerate}
\item Pick a value $\De_0$ and assume that all scalars in the $\f\x\f$ OPE have dimension $\De\geq \De_0$.

\item Use a computer to search for $a_{mn}$ such that
\begin{align}
\label{eq:positivityconditioninexample}
\sum_{m+n\leq \Lambda} & a_{mn}  \ptl_z^m\ptl_{\bar z}^n F^{\De_\f}_{\De,\ell}(z,\bar z)|_{z=\bar z = \frac 1 2} \geq 0,\nn\\
\textrm{for all }\quad\ell &= 0,2,\dots,\ell_\mathrm{max},\quad\De \geq \begin{cases}
\De_0 & (\ell=0), \\
\ell+d-2 & (\ell>0).
\end{cases}
\end{align}

\item If (\ref{eq:positivityconditioninexample}) is solvable, there must exist a scalar with dimension below $\De_0$.

\end{enumerate}

The best bound is the critical value $\De_{0}^\mathrm{crit.}$ above which (\ref{eq:positivityconditioninexample}) is solvable and below which it is not. To find it, we can perform a binary search in $\De_0$, running the algorithm above at each step. By additionally varying $\De_\f$, we obtain a $\De_\f$-dependent upper bound on the lowest-dimension scalar in the $\f\x\f$ OPE.

An implementation of this procedure is included with the semidefinite program solver {\tt SDPB} \cite{Simmons-Duffin:2015qma}.\footnote{See {\tt mathematica/Bootstrap2dExample.m} at {\tt https://github.com/davidsd/sdpb}.}  See also \cite{Behan:2016dtz} for a Python interface to {\tt SDPB} and \cite{Paulos:2014vya} for another user-friendly bootstrap package.  Running the code for $\Lambda=6,8,12,16,20,28$ gives the bounds shown in figure~\ref{fig:2dbound}.\footnote{We use the {\tt SDPB} parameters listed in the appendix of \cite{Simmons-Duffin:2015qma}.}  

\begin{figure}[hrt!]
\begin{center}
\includegraphics[width=\textwidth]{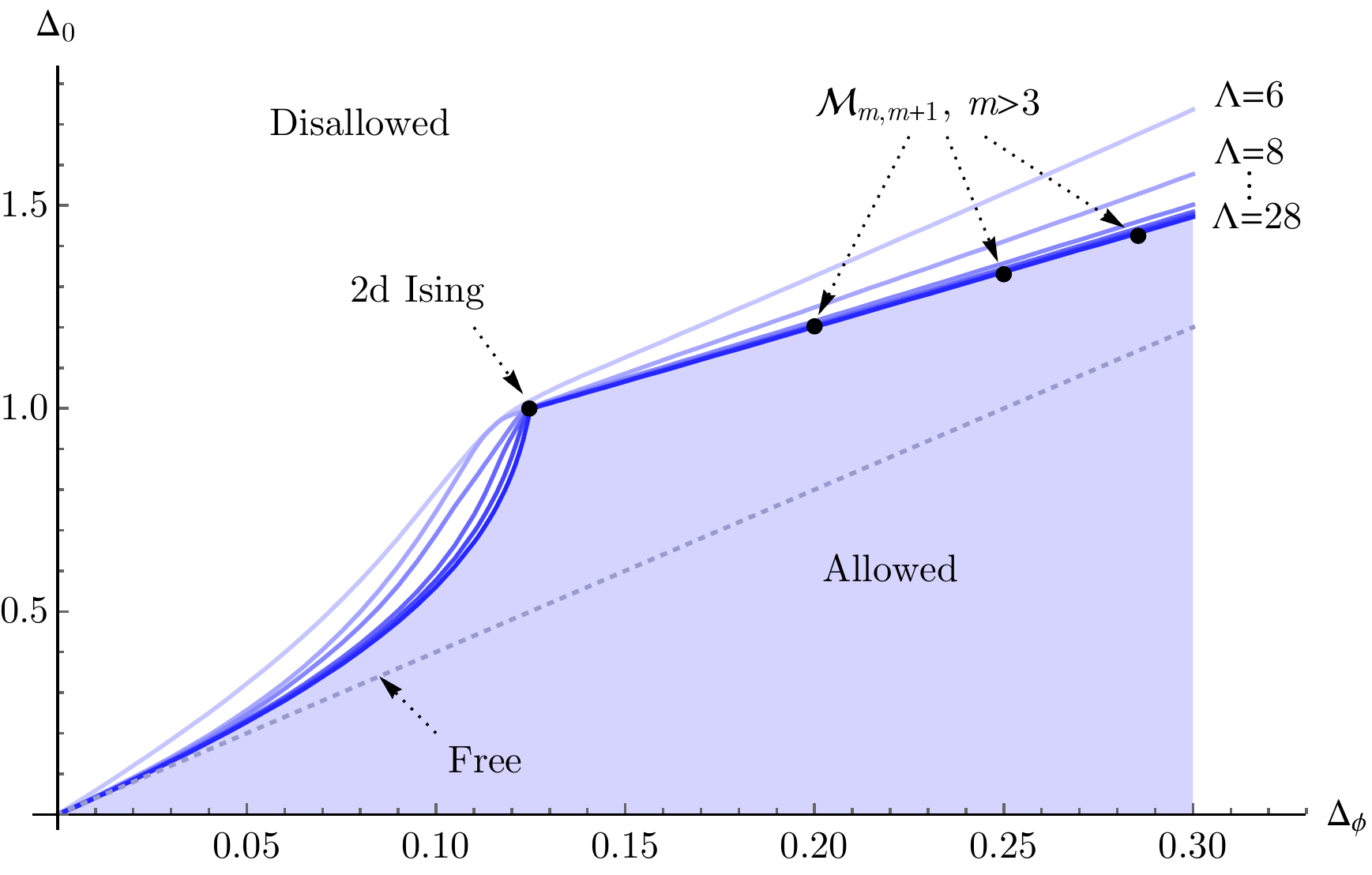}
\end{center}
\caption{\label{fig:2dbound}
Upper bounds on the dimension $\De_0$ of the lowest dimension scalar in the $\f\x\f$ OPE as a function of $\De_\f$, for 2d CFTs with a $\Z_2$ symmetry. The bounds are computed using {\tt SDPB} for $\Lambda=6,8,12,16,20,28$, with the strongest bound (darkest blue curve) corresponding to $\Lambda=28$ (a 105-dimensional space of functionals). The black dots represent the unitary minimal models $\mathcal{M}_{m,m+1}$ with $(\De_\f,\De_0)=(\frac 1 2 - \frac{3}{2(m+1)},2-\frac{4}{m+1})$ for $m=3,4,5,6$, of which the 2d Ising model is the case $m=3$.  The dashed line represents the lowest dimension scalar in an OPE of operators $\cos(k\phi)$ in the free boson theory.  These bounds first appeared in~\cite{Rychkov:2009ij}. It should be possible to improve on the lower bound in section~\ref{sec:boundsexample} as well, though we have not attempted this.
}
\end{figure}

As the cutoff $\Lambda$ on the number of derivatives increases, the bounds $\De_0^\mathrm{crit.}(\De_\f)$ get stronger. Remarkably, the strongest bounds are nearly saturated by interesting physical theories. The most obvious feature of figure~\ref{fig:2dbound} is a {\it kink\/} near the location of the 2d Ising model $(\De_\f,\De_0)=(\frac 1 8,1)$.  (Other exactly soluble unitary minimal models $\mathcal{M}_{m,m+1}$ also lie near the bound.) The bounds for different $\Lambda$ at the 2d Ising point $\De_\f=\frac 1 8$ are given in table~\ref{tab:twodbounds}.  Taking $\Lambda=28$ gives a bound $\De_\e \leq \De_0^\mathrm{crit.}(\frac 1 8)\approx 1.0000005$, within $5\x 10^{-7}$ of the correct value.

\begin{table}[ht]
\tbl{\label{tab:twodbounds} Upper bounds on $\De_\e$ in the 2d Ising model, computed with different cutoffs $\Lambda$ on the number of derivatives.}{
\begin{tabular}{c|c|c|c|c|c|c}
$\Lambda$  & 6 & 8 & 12 & 16 & 20 & 28 \\
\hline
$\De_0^\mathrm{crit.}(\De_\f=\frac 1 8)$ & 1.020 & 1.0027 & 1.00053 & 1.000043 & 1.0000070 & $\sim1.0000005$
\end{tabular}}
\end{table}

\subsection{Numerical Results in 3d}

It is helpful to compare to exact solutions in 2d, but the above results are remarkable because the methods are so general. We input information about 2d global conformal symmetry (nothing about the Virasoro algebra!) and unitarity, and the 2d Ising model pops out.  Wonderfully, the same thing happens in 3d! Again, we compute an upper bound on the lowest dimension scalar in the $\f\x\f$ OPE, this time using the 3d conformal blocks and the 3d unitarity bound.  The resulting bound, shown in figure~\ref{fig:singletbound3d}, has a kink at $(\De_\f,\De_0)\approx(0.518,1.412)$ --- close to the values realized in the 3d Ising CFT \cite{ElShowk:2012ht}.

\begin{figure}[hrt!]
\begin{center}
\includegraphics[width=\textwidth]{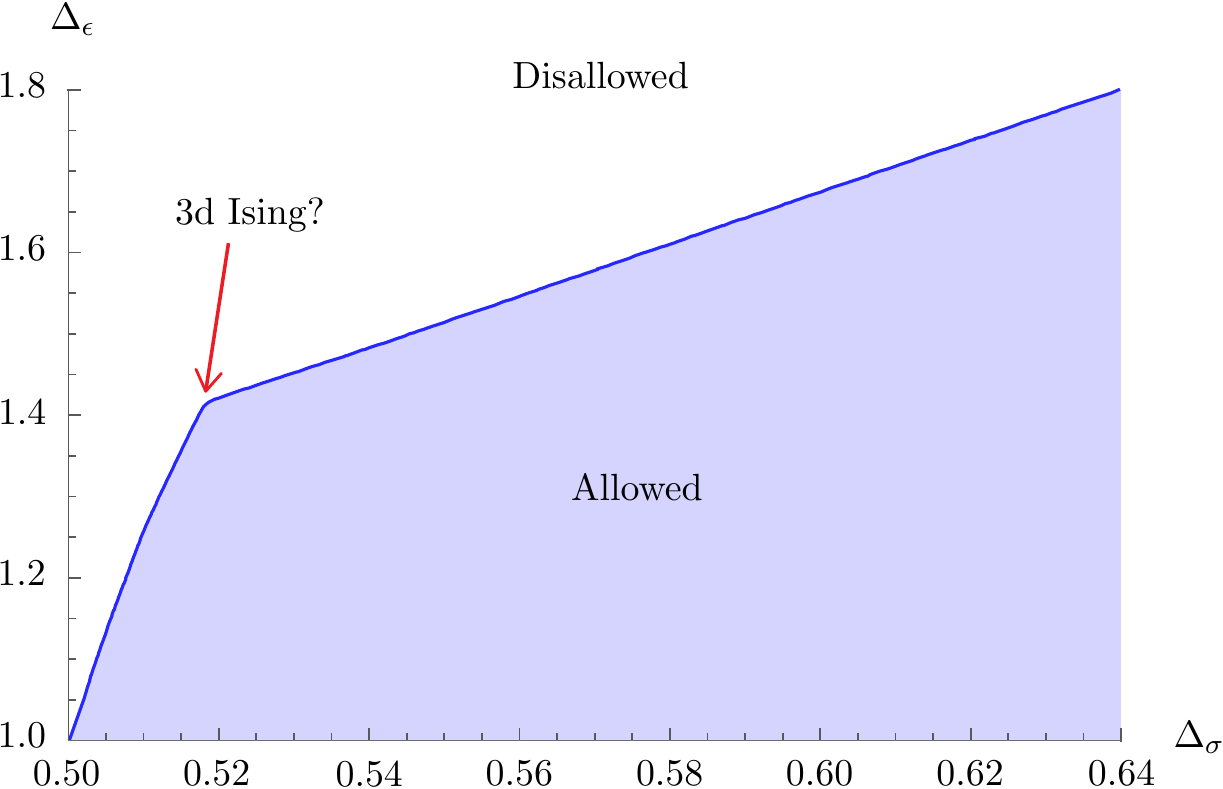}
\end{center}
\caption{\label{fig:singletbound3d} Upper bound on the dimension $\De_\e$ of the lowest dimension scalar in the $\s\x\s$ OPE, where $\s$ is a real scalar primary in a unitary 3d CFT with a $\Z_2$ symmetry, from \cite{El-Showk:2014dwa}. This bound is computed with $\Lambda=24$ (78-dimensional space of derivatives).}
\end{figure}

All the results discussed so far come from studying crossing symmetry of a single four-point function.  However, the techniques can be generalized to systems of correlation functions. The system $\<\s\s\s\s\>$, $\<\s\s\e\e\>$, $\<\e\e\e\e\>$ in the 3d Ising CFT was studied in \cite{Kos:2014bka}.  To get interesting new bounds in this case, it's necessary to input an additional fact: that $\s$ and $\e$ are the only relevant scalars in the theory.\footnote{This is an obvious experimental fact about the 3d Ising CFT. (It would be interesting to prove mathematically.)  Relevant scalars are in one-to-one correspondence with parameters that must be tuned to reach the critical point in some microscopic theory.  The fact that the phase diagram of water is 2-dimensional (the axes are temperature and pressure) tells us that the critical point of water has two relevant operators.}  In practice, this roughly means that we impose positivity conditions $\alpha(F_{\De,\ell})\geq 0$ for $\De=\De_\s,\De_\e$, and $\De\geq 3$.  The resulting bound in figure~\ref{fig:multicorrelator3d} now restricts $(\De_\s,\De_\e)$ to a small island in the space of operator dimensions.

\begin{figure}[hrt!]
\begin{center}
\includegraphics[width=\textwidth]{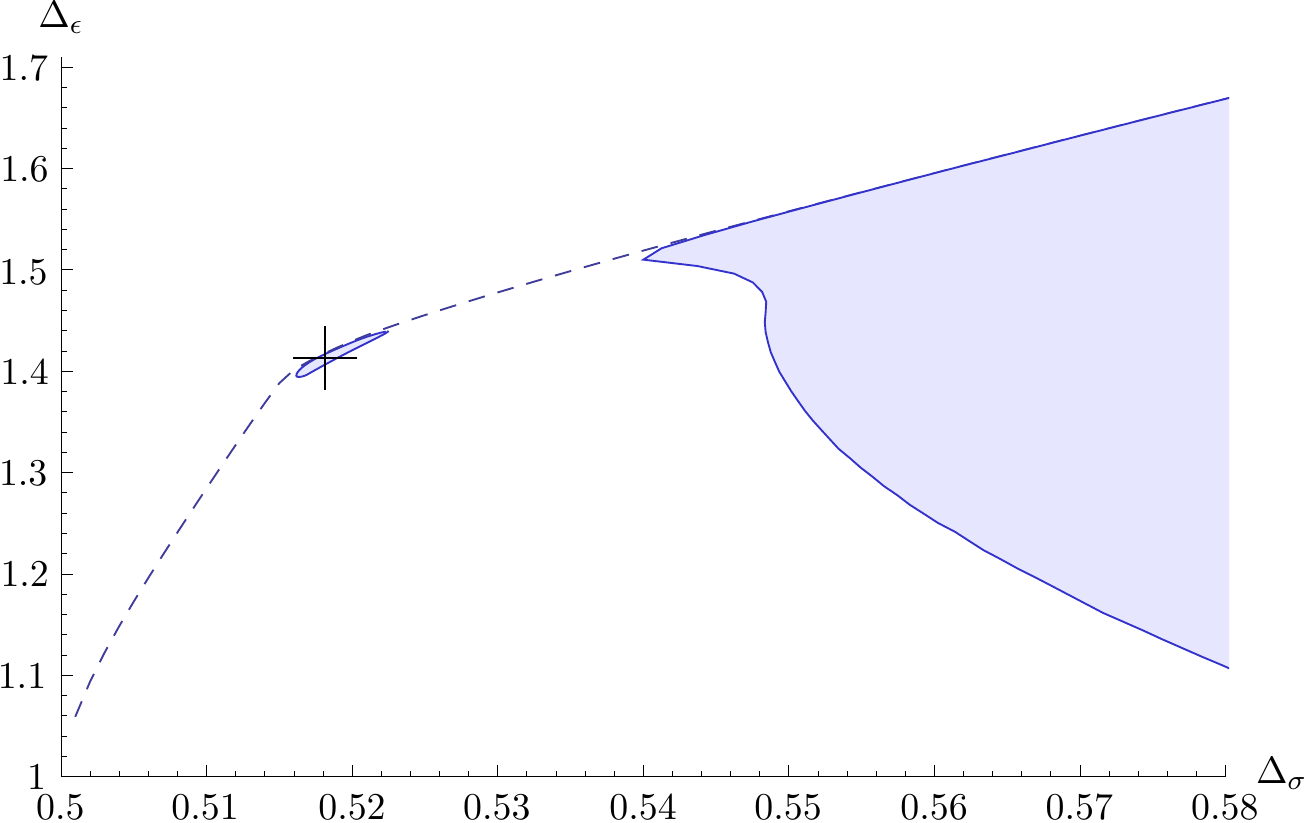}
\end{center}
\caption{\label{fig:multicorrelator3d} Bound on $(\De_\s,\De_\e)$ in a unitary 3d CFT with a $\Z_2$ symmetry and two relevant scalars $\s,\e$ with $\Z_2$ charges $-,+$.  The bound comes from studying crossing symmetry of $\<\s\s\s\s\>$, $\<\s\s\e\e\>$, $\<\e\e\e\e\>$, and is computed with $\Lambda=12$.  The allowed region is now a small island near the 3d Ising point (black cross), with an additional bulk region to the right.}
\end{figure}

The same multiple correlator bound, computed with $\Lambda=43$ using {\tt SDPB}, is shown in figure~\ref{fig:sdpbBound} \cite{Simmons-Duffin:2015qma}.  The island has shrunk substantially, giving a precise determination of the 3d Ising operator dimensions,
\be
(\De_\s,\De_\e) &=& (0.518151(6), 1.41264(6)).
\ee
Figures~\ref{fig:multicorrelator3d} and~\ref{fig:sdpbBound} are conceptually interesting.  Firstly, the striking agreement between Monte Carlo simulations and the conformal bootstrap is strong evidence that the critical 3d Ising model actually does flow to a conformal fixed-point, as originally conjectured by Polyakov \cite{Polyakov:1970xd}.

Secondly, figures~\ref{fig:multicorrelator3d} and~\ref{fig:sdpbBound} give a way to understand the phenomenon of critical universality discussed at the beginning of this course.  If a theory flows to a unitary 3d CFT with a $\Z_2$-symmetry and two relevant scalars $\s,\e$ --- and if $\De_\s,\De_\e$ don't live in the bulk region in figure~\ref{fig:multicorrelator3d} --- then the IR theory must live in the 3d Ising island!  Perhaps future bootstrap studies will shrink the 3d Ising island to a point, proving the IR equivalence of these theories.

\begin{figure}[hrt!]
\begin{center}
\includegraphics[width=\textwidth]{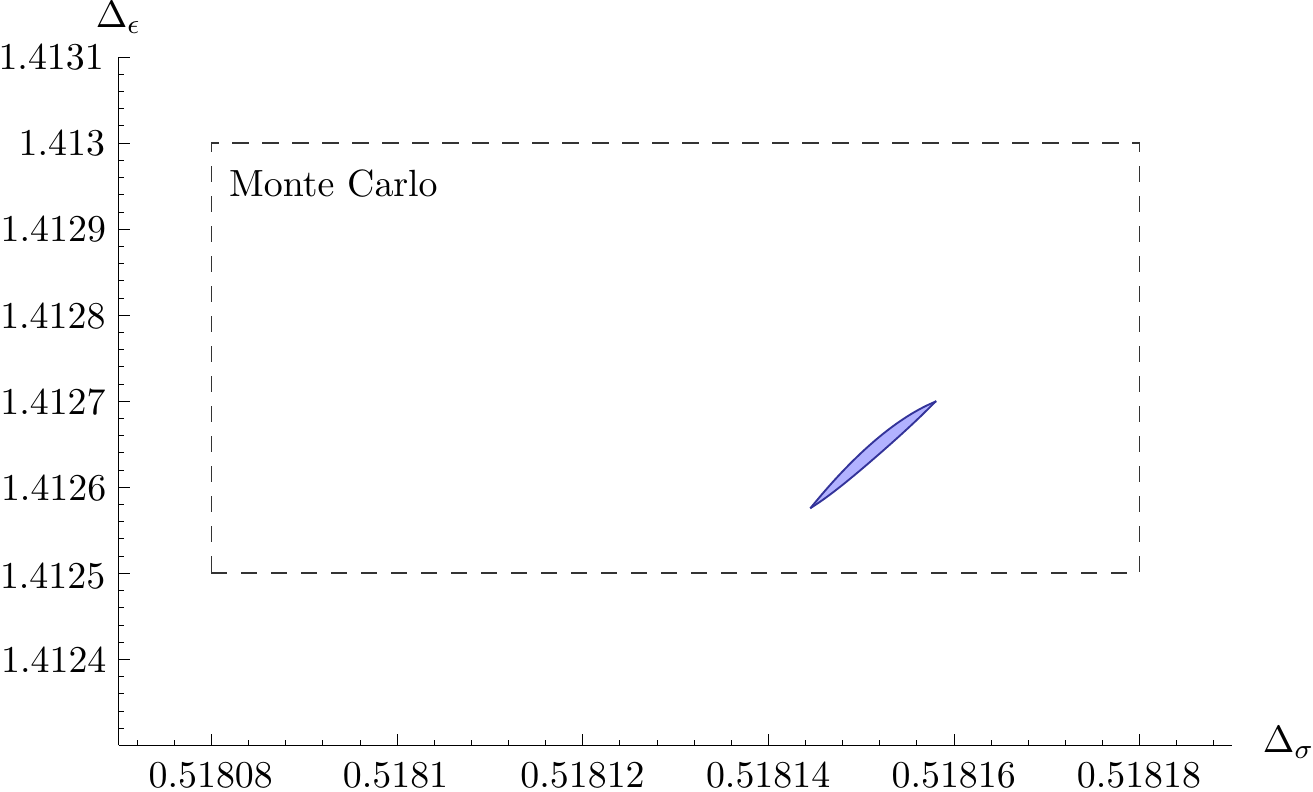}
\end{center}
\caption{\label{fig:sdpbBound} Bound on $(\De_\s,\De_\e)$ in a unitary 3d CFT with a $\Z_2$ symmetry and two relevant scalars $\s,\e$ with $\Z_2$ charges $-,+$.  The bound comes from studying crossing symmetry of $\<\s\s\s\s\>$, $\<\s\s\e\e\>$, $\<\e\e\e\e\>$, and is computed with $\Lambda=43$ using {\tt SDPB}. The allowed region is the blue sliver. The dashed rectangle shows the $68\%$ confidence region for the current best Monte Carlo determinations.}
\end{figure}

\subsection{Open Questions}

The techniques above have been applied to numerous theories in different spacetime dimensions, with different amounts of supersymmetry \cite{Rattazzi:2008pe, Rychkov:2009ij, Caracciolo:2009bx, Poland:2010wg, Rattazzi:2010gj, Rattazzi:2010yc, Vichi:2011ux, Poland:2011ey, Rychkov:2011et, ElShowk:2012ht,Liendo:2012hy, Beem:2013qxa, Kos:2013tga, El-Showk:2013nia, Alday:2013opa, Gaiotto:2013nva,Bashkirov:2013vya, Berkooz:2014yda, El-Showk:2014dwa, Nakayama:2014lva,Nakayama:2014yia, Alday:2014qfa, Chester:2014fya, Kos:2014bka, Caracciolo:2014cxa, Nakayama:2014sba, Golden:2014oqa, Chester:2014mea, Paulos:2014vya, Beem:2014zpa, Simmons-Duffin:2015qma, Bobev:2015vsa, Bobev:2015jxa, Kos:2015mba, Chester:2015qca, Beem:2015aoa, Iliesiu:2015qra,poland2015exploring,Lemos:2015awa,Lin:2015wcg,Chester:2015lej,Chester:2016wrc}.  Because we don't start with a Lagrangian, there's no guarantee when and how a particular physical theory will show up in the bounds. It's an open question which correlators to study to isolate different CFTs.

Other open questions include the following:
\begin{itemize}
\item How do the bounds behave in the limit $\Lambda\to \oo$? Does the Ising island shrink to a point, still using a finite number of correlation functions, or must we study larger systems of crossing equations?
\item How does one efficiently compute higher operator dimensions and OPE coefficients?  The extremal functional method \cite{Poland:2010wg,ElShowk:2012hu,El-Showk:2014dwa} is one way, but it is hard to estimate the associated errors.
\item Can the requirement of unitarity be relaxed? Gliozzi's method of determinants \cite{Gliozzi:2013ysa} has shown success analyzing the crossing equation in nonunitary theories and other situations where positivity is not obviously present \cite{Gliozzi:2014jsa,Gliozzi:2015qsa,Nakayama:2016cim}. Can it be made rigorous?
\item What information is hidden in correlators of higher-spin operators like stress tensors?
\item What can we prove analytically about the crossing equations? Progress has been made in certain limits, for example large-$N$ \cite{Heemskerk:2009pn}, large dimension \cite{Pappadopulo:2012jk}, large spin \cite{Alday:2007mf,Fitzpatrick:2012yx,Komargodski:2012ek,Alday:2015ewa}, and combinations thereof \cite{Fitzpatrick:2014vua,Fitzpatrick:2015qma,Fitzpatrick:2015zha}. Another approach is to study the implications of slightly broken symmetries \cite{Maldacena:2012sf,Rychkov:2015naa,Alday:2015ota,Giombi:2016hkj}. It would be extremely interesting to prove analytical results about the small-$N$, small $\De,\ell$ regime.
\item What additional structures and consistency conditions should we incorporate into the bootstrap? (See section~\ref{sec:additionalstructures}.)
\item What protected information can be computed using supersymmetry? Bootstrap studies recently led to the discovery of beautiful new algebraic structures in the OPE algebra of supersymmetric theories in $3,4,6$ dimensions \cite{Beem:2013sza,Chester:2014mea,Beem:2016cbd}. How do these structures interact with the full non-protected bootstrap? 
\end{itemize}

That's a lot of open questions, and there are certainly many more.
I hope some of you will help find the answers!

\section*{Acknowledgements}

I am grateful to Joe Polchinski and Pedro Vieira for inviting me to give this course, and Tom DeGrand, Oliver DeWolfe, and Sherry Namburi for helping make TASI such a fun experience.  I am also grateful to Justin David, Chethan Krishnan and Gautam Mandal for organizing the Advanced Strings School at ICTS, Bangalore.  A special thanks to the spectacular students at TASI and the Strings School, who asked so many good questions.  Thanks to Chris Beem, Joanna Huey, Filip Kos, Petr Kravchuk, David Poland, and Slava Rychkov for comments.  Thanks to Sheer El-Showk, Jo\~ao Penedones, and Pedro Vieira for the nice example in section~\ref{sec:boundsexample}. I am supported by DOE grant number DE-SC0009988 and a William D. Loughlin Membership at the Institute for Advanced Study.

\begin{appendix}[Quantization of the Lattice Ising Model]
\label{app:latticequantization}

In this section, we show how to interpret the partition function of the Ising model on a square lattice in terms of Hilbert spaces and discrete time evolution.  This is a textbook trick,\footnote{It is the starting point for Onsager's exact solution of the 2d Ising model \cite{PhysRev.65.117}.} but we review it because it clearly illustrates several ideas from section~\ref{sec:quantization}.

Consider the 2d Ising model on an $m\x n$ lattice with periodic boundary conditions.  The spins are given by $s_{i,j}\in \{\pm 1\}$, where $i\in \Z/m\Z$ and $j\in \Z/n\Z$.  The partition function is
\be
\label{eq:isingpartitionfunction}
Z &=& \sum_{\{s\}} \exp\p{-JS_h(s) -J S_v(s)},\nn\\
S_h(s) &\equiv& \sum_{i,j} s_{i,j}s_{i+1,j},\nn\\
S_v(s) &\equiv& \sum_{i,j} s_{i,j} s_{i,j+1},
\ee
where we have separated the action into contributions from horizontal and vertical bonds.

We will think of the $j$-direction as ``time", and introduce a Hilbert space $\cH_m$ associated with a ``slice" of $m$ lattice sites at constant time.  The space $\cH_m$ has a basis state for each spin configuration on the slice,
\be
|s_1,\dots,s_m\>,\qquad s_i\in\{\pm 1\}.
\ee
These are the analogs of the field eigenstates $|\phi_b\>$ in section~\ref{sec:operatorimpliesstate}.
The Pauli spin matrices $\widehat\sigma_i^{\mu}$, $\mu=x,y,z$ act on the $i$-th site.

The operator
\be
 U &\equiv & \exp\p{-J\sum_{i=1}^m \widehat\s^z_i \widehat\s^z_{i+1}}
\ee
encodes the contribution to the partition function from horizontal bonds between $m$ spins in a line. For example on an $m\x 1$ lattice, we would have
\be
 U |s_1,\dots,s_m\> &=& e^{-JS_h(s)}|s_1,\dots,s_m\>.
\ee
The operator
\be
 V &\equiv& \prod_i(e^{-J} + e^{J}\widehat\s_i^x)
\ee
encodes the effects of vertical bonds. For each site, it either preserves the spin, giving a factor $e^{-J}$ associated with aligned spins, or flips it, giving a factor $e^{J}$ associated with anti-aligned spins.  Defining the ``transfer matrix" $ T\equiv VU$, it's easy to check that
\be
Z &=& \Tr_{\cH_m}( T^n).
\ee

We have interpreted the discrete path integral (\ref{eq:isingpartitionfunction}) in terms of operators on a Hilbert space.  The transfer matrix is a discrete analogue of the time-evolution operator $e^{-tH}$.  The path integral variable $s_{i,j}$ maps to the quantum operator
\be 
\label{eq:quantizationmapone}
s_{i,j} &\to& T^{-j} \s_i^z  T^j,
\ee
and correlation functions become traces of time-ordered products, e.g.\footnote{We use the convention $\th(0)=\frac 1 2$.}
\be
\<\s_{i_1,j_1}\s_{i_2,j_2}\> &=& \Tr_{\cH_m}(T^{n+j_2-j_1} \widehat \s_{i_1}^z T^{j_1-j_2} \widehat \s_{i_2}^z)\theta(j_1 - j_2) + (1\leftrightarrow 2).\nn\\
\ee

We could instead have quantized the theory with the horizontal direction as time.  This would give a different Hilbert space $\cH_n$ with dimension $2^n$ instead of $2^m$, a new transfer matrix $T'$ (acting on $\cH_n$), and a different formula for the same path integral
\be
Z &=& \Tr_{\cH_n}(T'^m)=\Tr_{\cH_m}(T^n).
\ee
The new quantization map would be
\be
\label{eq:quantizationmaptwo}
s_{i,j} &\to& T'^{-i} \widehat \s_j^z T'^i.
\ee
Let us emphasize that the operators (\ref{eq:quantizationmapone}) and (\ref{eq:quantizationmaptwo}) are truly different, even though they represent the same path integral variable.  They even act on different-dimensional Hilbert spaces ($2^m$ vs.\ $2^n$)!  Thus, it's not surprising that properties associated to a particular quantization, like their behavior under Hermitian conjugation (section~\ref{sec:reflectionpositivity}), could be different.

\vspace{0.2in}

\end{appendix}

\begin{appendix}[Euclidean vs. Lorentzian and Analytic Continuation]
\label{app:analyticcontinuation}

Here we make some brief comments about Euclidean and Lorentzian correlation functions and analytic continuation between them.  

The first comment is that in Euclidean quantum field theory, out-of-time-order correlators don't make sense.  For example, consider a Euclidean two-point function,
\be
\label{eq:euclideantwopt}
\<0|\cO_1(t_1)\cO_2(t_2)|0\> &=& \<0| \cO_1(0) e^{H(t_2-t_1)} \cO_2(0)|0\>.
\ee
In QFT, the Hamiltonian $H$ is bounded from below and has arbitrarily large positive eigenvalues.  If we take $t_2 > t_1$, then the operator $e^{H(t_2-t_1)}$ is unbounded.  Generically, local operators $\cO_{1,2}(0)$ have nonzero amplitude to create arbitrarily high energy states. Thus, (\ref{eq:euclideantwopt}) is formally infinite.

Because the Euclidean path integral gives a time-ordered product
\be
\label{eq:euclideantimeorderedproduct}
\<\cO_1(t_1)\cO_2(t_2)\> &=& \<0|\cO_1(0)e^{H(t_2-t_1)}\cO_2(0)|0\>\th(t_1-t_2)+\nn\\
&& \<0|\cO_2(0)e^{H(t_1-t_2)}\cO_1(0)|0\>\th(t_2-t_1),
\ee
it is well-defined for any ordering of the time coordinates. Specifically, the operators $e^{H(t_i-t_j)}$ in (\ref{eq:euclideantimeorderedproduct}) are always bounded.

In Lorentzian quantum field theory, however,  non-time-ordered correlators (Wightman functions) are interesting observables.  They can be obtained from time-ordered Euclidean correlators as follows.  First set the Euclidean times equal to small values $t_{Ei}=\e_i$, increasing in the same order as the operator ordering we want.  For example, to place $\cO_1$ later than $\cO_2$, consider
\be
\<\cO_1(\e)\cO_2(0)\> &=& \<0|\cO_1(0)e^{-\e H}\cO_2(0)|0\>,\qquad \e>0.
\ee
Now continue $t_{Ei}$ in the pure imaginary direction to the desired Lorentzian times $it_{Li}$.  Because $e^{H(t_{Ei}-t_{Ej})}$ never becomes unbounded, the operators remain in the same order,
\be
\<\cO_1(\e+i t_{L1})\cO_2(it_{L2})\> &=& \<0|\cO_1(0)e^{-\e H - iH(t_{L1}-t_{L2})}\cO_2(0)|0\>.
\ee
Finally, take $\e\to 0$ to get the desired Wightman function.

To get a time-ordered Lorentzian correlator, there is a simple trick: just simultaneously rotate all Euclidean times $t\to i (1-i\e) t$. Because the ordering of the real parts of $t$ are preserved, the order of the operators will be too. This is Wick rotation.

Many properties of correlators under analytic continuation are clearer when thinking about states and Hamiltonians, as opposed to path integrals.

\vspace{0.2in}

\end{appendix}

\begin{appendix}[Semidefinite Programming]
\label{app:semidefinite}

For our purposes, a semidefinite program solver is an oracle that can solve the following problem:
\be
\label{eq:semidefiniteprogram}
\textrm{Find $\vec a$ such that $\vec a \. \vec P_i(x) \geq 0$ for all $x\geq 0$, $i=1,\dots,N$},
\ee
where $\vec P_i(x)$ are vector-valued polynomials.  There are many freely-available semidefinite program solvers.  {\tt SDPB} \cite{Simmons-Duffin:2015qma} in particular was written for application to the conformal bootstrap.

We would like to write our search in the form (\ref{eq:semidefiniteprogram}).
After restricting to the subspace (\ref{eq:choiceoffunctionals}), our positivity constraints become
\be
\label{eq:positivityinfinitespace}
\sum_{m+n\leq \Lambda} a_{mn}\ptl_z^m \ptl_{\bar z}^n F_{\De,\ell}^{\De_\f}(z,\bar z)|_{z=\bar z = \frac 1 2} \geq 0.
\ee
The trick is to find an approximation
\be
\label{eq:polynomialapproximation}
\ptl_z^m \ptl_{\bar z}^n F_{\De,\ell}^{\De_\f}(z,\bar z)|_{z=\bar z= \frac 1 2} &\approx& \chi_\ell(\De) P^{mn}_\ell(\De),
\ee
where $\chi_\ell(\De) \geq 0$ are positive and $P^{mn}_\ell(\De)$ are polynomials.  Then, dividing (\ref{eq:positivityinfinitespace}) by $\chi_\ell(\De)$ and writing $\De=\De_{\mathrm{min},\ell}+x$, our inequality becomes
\be
\sum_{m+n\leq \Lambda} a_{mn} P_\ell^{mn}(\De_{\mathrm{min},\ell}+x) &\geq& 0.
\ee
This has the right form if we group the coefficients $a_{mn}$ into a vector $\vec a$ and identify $\ell\to i$, $\ell_\mathrm{max}\to N$.  The value $\De_{\mathrm{min},\ell}$ depends on the calculation at hand, see for example (\ref{eq:positivityconditioninexample}).

To get a positive-times-polynomial approximation, we can start with the series expansion (\ref{eq:seriesexpansion}),
\be
g_{\De,\ell}(u,v) &=& \sum_{n,j} B_{n,j} r^{\De+n} C_j^{\frac{d-2}{2}}(\cos\th).
\ee
Recall that the coefficients $B_{n,j}$ are positive rational functions of $\De$.  The crossing-symmetric point $z=\bar z = \frac 1 2$ corresponds to a very small value of $r=r_*=3-2\sqrt 2 \approx 0.17$.  Thus, truncating the series at some large $n_\mathrm{max}$ gives a good approximation,
\be
\ptl_r^a \ptl_\th^b g_{\De,\ell}(u,v)|_{r=r_*,\th=0} &\approx & r_*^\De \frac{P^{ab}_\ell(\De)}{Q_\ell(\De)}  + O(r_*^{\De+n_\mathrm{max}}),
\ee
where $P^{ab}_\ell$ and $Q_\ell$ are polynomials and $Q_\ell(\De)$ is positive.  Since derivatives of $F_{\De,\ell}^{\De_\f}(z,\bar z)$ are linear combinations of derivatives of $g_{\De,\ell}(u,v)$, this establishes (\ref{eq:polynomialapproximation}) with
\be
\chi_\ell(\De) &=& \frac{r_*^\De}{Q_\ell(\De)}.
\ee

When exact formulae for conformal blocks are not available (for example, in odd dimensions), the polynomials $P^{ab}_\ell(\De)$ can be computed using recursion relations \cite{Zamolodchikov:1985ie,Zamolodchikov:1987,Kos:2013tga,Kos:2014bka,Penedones:2015aga,Yamazaki:2016vqi,Iliesiu:2015akf} or differential equations \cite{Hogervorst:2013kva}.

\end{appendix}

\bibliographystyle{ws-rv-van}
\bibliography{TASI-bootstrap-chapter}
\ifarxivsubmission
\else
  \blankpage
\fi
\printindex
\end{document}